\documentclass[a4paper,11pt]{article}
\usepackage[margin=1in]{geometry} 
\usepackage{setspace}
\usepackage[english]{babel}
\usepackage{afterpage}
\usepackage{amsthm}
\usepackage{amssymb}
\usepackage{amsmath}
\usepackage{graphicx}
\usepackage[utf8]{inputenc}
\usepackage{dsfont} 
\usepackage{textcomp}
\usepackage{listings}
\usepackage{color}
\usepackage{titling}
\usepackage{subcaption}
\usepackage[table]{xcolor}
\usepackage{tabularx}
\usepackage{enumerate}
\usepackage{multirow}
\usepackage{algorithm}
\usepackage{enumitem}
\usepackage[noend]{algpseudocode}

\usepackage[babel,german=quotes]{csquotes}
\usepackage[style=bwl-FU, date=year, doi=false, url=false, isbn=false, dashed=false, uniquename=init,maxbibnames=99,uniquelist=false,maxcitenames=3,backend=bibtex]{biblatex}
\bibliography{ISLMC}
\renewbibmacro{in:}{}
\renewbibmacro*{volume+number+eid}{%
  \printfield{volume}
  \addnbthinspace
  \printfield{number}%
  \setunit{\addcomma\space}%
  \printfield{eid}}
\DeclareFieldFormat[article]{number}{\mkbibparens{#1}}

\let\epsilon\varepsilon
\newtheorem{defi}{Definition}[section]
\newtheorem{rema}[defi]{Remark}
\newtheorem{lem}[defi]{Lemma}
\newtheorem{theo}[defi]{Theorem}
\newtheorem{ass}[defi]{Assumption}
\newtheorem{prop}[defi]{Proposition}
\newtheorem{exam}[defi]{Example}

\allowdisplaybreaks

\colorlet{g}{green!70}
\colorlet{p}{red!25}
\colorlet{vp}{red!60}
\colorlet{vg}{green!25}
\def\-{\raisebox{.75pt}{-}}

\newcommand{\bbr}{\mathbb{R}}

\DeclareMathOperator*{\argmin}{arg\,min}
\begin{document}


\title{An Integrated Approach to Importance Sampling and Machine Learning for Efficient Monte Carlo Estimation of Distortion Risk Measures in Black Box Models}

\author{
S\"oren Bettels \hspace{1.3cm}  Stefan Weber\\[1.0ex] \textit{Leibniz Universit{\"a}t Hannover} }
\date{\today\thanks{House of Insurance \& Institute of Actuarial and Financial Mathematics, Leibniz Universit\"at Hannover, Welfengarten 1, 30167 Hannover, Germany. e-mail:  {\tt stefan.weber@insurance.uni-hannover.de}}}

\maketitle

\begin{abstract} 
Distortion risk measures play a critical role in quantifying risks associated 
with uncertain outcomes. Accurately estimating these risk measures in the context of computationally expensive simulation models that lack analytical tractability is fundamental to 
effective risk management and decision making. In this paper, we propose an efficient important
sampling method for distortion risk measures in such models that reduces the computational cost through machine learning. We demonstrate the applicability and efficiency of the Monte Carlo method in numerical experiments on various distortion risk measures and models. 
\end{abstract}
\vspace{0.2cm}

\doublespacing
	
\section{Introduction} \label{sec:introduction}

Many real world simulation models are typically highly complex. Controlled random inputs are transformed by functions that require costly evaluations. These models also provide the basis for risk measurement and control of firms and systems, and this requires a careful analysis of rare events. Important industry examples are internal models of banks and insurance companies that are applied in  their internal risk management process and solvency regulation. 

The goal of this paper is to develop an importance sampling (IS) algorithm for computing an important class of measures of the downside risk, distortion risk measures (DRMs), when very costly computations are required in the mapping of model inputs to outputs. We call such models \emph{black box models} which are characterized by high computational complexity -- sometimes even by opaque mechanisms obscuring the relationship between input and output variables. Besides machine learning (ML) techniques, our simulation approach builds on two main ingredients: (i) efficient IS for quantiles which was developed by \cite{Glynn1996} and \cite{AhnShy2011}, and (ii) representations of distortion risk measures as mixtures of quantiles, cf. \cite{Dhaene2012}.

The quantitative characterization of the downside risk has been studied systematically since the 1990s; the axiomatic foundation for downside risk quantification was laid by \cite{Artzner1999}, \cite{FoellmerSchied2002}, and \cite{FrittelliGianin2002}. An important and extensive class of risk measures are distortion risk measures (DRMs). Particular examples are many distribution-based coherent risk measures, but also frequently used non-convex risk measures. These include value at risk, average value at risk, commonly known as expected shortfall, and range value at risk. Also the insurance premium principles of Wang are included in this class, cf. \cite{Wang1995}, \cite{Wang1996}. 
\noindent
The main innovations of this paper are:
\begin{enumerate}
    \item[(i)]   Using ML techniques, we design an IS algorithm for the Monte Carlo estimation of DRMs in black box models. The construction of an IS distribution requires a discretization of DRM to a mixtures of quantiles, suitable measure changes for quantiles at different levels, and an efficient allocation of the available samples to these levels. Machine learning provides a cheaper alternative for the highly costly evaluation of the black box model. 
    \item [(ii)] We analyze and illustrate the performance of the method in various case studies. For DRMs that focus on extremely rare scenarios, we additionally suggest and test an iterative refinement. Finally, the algorithm is successfully applied in a simple asset-liability management model of an insurance firm.
\end{enumerate}

\subsection*{Literature}

A more extensive treatment of risk measures and DRMs can be found in \cite{FoellmerSchied2016} and 
\cite{FoellmerWeber2015}. DRMs are closely related to Choquet integrals 
introduced by \cite{Choquet1954} and discussed in detail in \cite{Denneberg1994}. DRMs are, for example, studied in 
\cite{Wang1995}, \cite{Wang1996}, \cite{Kusuoka2001}, \cite{Acerbi2002}, \cite{DhaeneEtAl2006}, 
\cite{SongYan2006}, \cite{SongYan2009}, \cite{SongYan2009_2}, \cite{Weber2018} and 
\cite{KimWeber2022}. The representation theorem for DRMs used in the paper can be found in \cite{Dhaene2012}.

Surveys on Monte Carlo simulation and IS are \cite{Glasserman2003} 
and \cite{AsmussenGlynn2007}. Importance sampling techniques for rare-event simulation  include \cite{RubinoTuffin2009}, \cite{Bucklew2004}, 
\cite{BlanchetGlynn2008}, \cite{DupuisWang2002}, \cite{HultNyquist2016}, \cite{AsmussenEtAl2000}, and 
\cite{JunejaShahabuddin2006}. More closely related to DRMs are the following papers.  The asymptotic properties of the IS quantile estimators used in this paper are discussed in 
 \cite{Glynn1996} and generalized in \cite{AhnShy2011}. \cite{Glynn1996} considers the IS estimation of quantiles and suggests four estimators for which asymptotic normality is shown. Results from the theory of large deviations motivate in applications the choice of IS distributions within an exponential class. \cite{AhnShy2011} build on this contribution and study IS for V@R and AV@R. Asymptotic normality is proven under weaker conditions. \cite{AriefEtAl2021} consider rare-event simulation in black box systems focusing on the estimation of probabilities. \cite{GlassermanHeidelbergerShahabuddin2002} study the IS estimation of V@R for heavy-tailed risk factors on the basis of exponential measure changes. \cite{DunkelWeber2007} investigate the estimation of utility-based shortfall risk, combining stochastic approximation and IS. \cite{BrazauskasJonesPuriZitikis2008} focus on the estimation of conditional value at risk, but do not consider IS. The paper proves, for example, the consistency of the estimator and constructs confidence intervals. \cite{SunHong2009} study IS for value at risk and average value risk, exploiting the OCE representation of average value at risk which is due to \cite{RU00}, \cite{RU02}, see also \cite{BTT2007}. Measure changes are selected from an exponential family. \cite{BeutnerZaehle2010} present a modified functional delta method for the estimation of DRMs and derive asymptotic distributions and approximate confidence intervals, mainly motivated from a statistical perspective. They do not consider IS. \cite{PandeyPrashanthSanjay2021} combine a trapezoidal rule and quantile estimation to estimate spectral risk measures. Bounds for the error in probability are proven. IS is not considered.
Estimators of DRMs are often related to $L$-estimators for which the reader is referred to 
\cite{Stigler1974} and \cite{Serfling1980}. 

Surveys on techniques and applications of ML are 
\cite{ShalevSchwartzBenDavid2014} and \cite{MohriRostamizadehTalwalkar2018}. 
Some applications of black box models in finance are discussed in \cite{HuangChaiCho2020}.

\subsection*{Outline}
The paper is structured as follows: In Section 2 we set the scene by briefly introducing DRMs, the considered
quantile estimators, and their asymptotic distribution. The IS method for DRMs is also developed in this section. Section 3 applies the method across various case studies to test its performance. Section 4 discusses an application to asset-liability management of an insurance firm. Auxiliary results are collected in an online appendix. This includes background material on distortion risk measures, asymptotics of quantile estimators in IS, tools from machine learning, some computations and proofs, and additional figures on the basis of data that were obtained in the case studies.

\section{Efficient Estimation of DRM of Black Box Models} \label{sec:EffEstDRM}

\subsection{Setting the Scene} \label{sec:SetScene}

Accurately measuring risk in complex systems is an important task. Let $(\Omega, \mathcal{F}, \mathsf{P})$ be an atomless (i.e., sufficiently rich) probability space and $X:\Omega \to \mathbb{R}^d$ a random vector. The random outcome of the system is modeled by a random variable $Y=h(X)$ for some measurable function $h: \mathbb{R}^d\to \mathbb{R}$. We assume that $Y$ is accessible via a simulation oracle, but that $h$ is highly complex and not analytically tractable. In contrast, the distribution of the random vector $X$ is explicitly known and can appropriately be modified in order to increase the efficiency of the estimation. The assumption is that even if the simulation mechanism for $X$ changes, the function $h$ can still be evaluated, but its evaluation is very costly. The problem is to determine $\rho(h(X))$ by simulation where $\rho$ is a monetary risk measure. Our sign convention is that $h(X)$ counts losses as positive and gains as negative, as is customary in actuarial science. More specifically, we suppose that $\rho = \rho_g$ is a DRM associated to a distortion function $g:  [0,1] \to [0,1]$ of the form
$ \rho_g(Y) = \int_{-\infty}^0 [ g \left( \mathsf{P}(Y > y) \right) - 1 ] dy + \int_0^\infty g(\mathsf{P}(Y > y)) dy. $ For further details, we refer to Appendix~\ref{app:RM}. By \cite{Dhaene2012}, see also \cite{BettelsKimWeber2022}, DRMs can be written as mixtures of quantiles, i.e., 
\begin{align}
    \rho_g(Y) = c_1  \int_{[0,1]} q_Y^+ ( 1 -u ) dg_1(u) + c_2 \int_{[0,1]} q_Y(1 - u) dg_2(u), \label{form:DRMDecomp}
\end{align}  
where $g_1, g_2$ are right- resp. left-continuous distortion functions, $c_1 + c_2 = 1$, $c_1, c_2 \in[0,1]$, $g = c_1 g_1 + c_2 g_2$,
and $q_Y^+(u) = \sup \{ y \vert F_Y(y) \leq u\}$, $q_Y(u) = \inf\{ y \vert F_Y(y) \geq u \}$. Eq.~\eqref{form:DRMDecomp} is the starting point for the Monte Carlo simulation scheme.

The risk estimation problem has two aspects:  The quantiles, which appear as integrands in eq.~\eqref{form:DRMDecomp}, must be simulated efficiently, and the integrals must be discretized. We propose an importance sampling technique for the quantile estimation procedure based on ML estimation of the function $h$; in addition, we devise a strategy for allocating samples along the discretization to achieve good performance. The next sections explain how to design and to implement the following algorithm for the Monte Carlo estimation of DRMs.

\begin{algorithm}[h!]
\begin{algorithmic}[1]
    \State {\bf Input: } Distortion function $g$, pivot sample size $M$, sample size $N$, size of partition $m$.
    \State {\bf Output: Estimation of $\rho_g(Y)$}
    \Function{Main:}{}
        \State Set $\alpha_i = i \alpha /m$ for $i \in \{0, \dots, m\}$  and $\alpha_{m+1} = 1$;
        \State Sample $X \gets (X_1, \dots, X_M)$ from $F$ and set $Y \gets (h(X_1), \dots, h(X_M))$;
        \For{$i\in \{0, \dots, m\}$}{}
            \State Set $aux \gets$ empirical quantile of sample $Y$ at level $1- \alpha_{i}$;
            \State Set $\vartheta_i$ such that 
            $$ aux =  \frac{\sum_{j=1}^M Y_i \exp ( \vartheta_i Y_j)}{ \sum_{j=1}^M \exp( \vartheta_i Y_j)}; $$
        \EndFor
        \For{$i\in \{0, \dots, m\}$}{}
            \State Set $aux \gets$ empirical quantile of sample $Y$ at level $1- \alpha_i$;
            \State Set $aux\_c \gets \frac{1}{M} \sum_{j=1}^M \frac{dF}{dF_{\vartheta_i}} (X_j) \mathds{1}_{\{Y_j > 
            aux \}};$            
            \State Set $c_i$ such that
            $$ c_i \gets \frac{aux\_c - \alpha_i^2}{G^\prime(aux)} \cdot ( g(\alpha_{i+1}) - g(\alpha_i));$$
        \EndFor
        \For{$i\in \{0, \dots, m\}$}{}
            \State Set 
            $$ p_i \gets  \frac{\sqrt{c_i}}{\sum_{i=0}^m \sqrt{c_i}};$$
        \EndFor
        \State Choose $\hat{h}$ as the regression selected by a $k$-fold validation and calibration from $X, Y$;
        \State Set $F_i$ for $i \in \{0,1, \dots, m \}$ such that 
        $$ dF_i = \exp \left( \vartheta_i \hat{h}(x) - \hat{\psi}(\vartheta_i) \right) dF; $$
    \State Sample $\theta_1, \dots, \theta_N$ as i.i.d. copies of $\theta$
        such that $\mathsf{P}(\theta= i) = p_i$ for $i \in \{0,1, \dots, m \}$;
    \State Sample $X^\prime \gets (X^\prime_1, \dots, X^\prime_N)$ such that $X_i \sim F_{\theta_i}$ and set 
    $Y^\prime \gets (h(X_1^\prime), \dots, h(X^\prime_N))$;
    \State Set $estimate \gets 0$;
    \For{$i\in \{0, \dots, m\}$}{}
        \State {\bf Option 1:} Compare the variances of $\hat{q}_{F_i, N_i}( 1 - \alpha_i)$ and $\hat{q}_{F^\ast, N}(1 - \alpha_i)$;
        \State \hspace{19mm} Set $\hat{q}_Y(1 - \alpha_i)$ as the better performing estimator;
        \State {\bf Option 2:} Set $\hat{q}_Y(1 - \alpha_i) \gets \hat{q}_{F^\ast, N}(1 - \alpha_i)$; 
        \State Set $estimate \gets estimate + \hat{q}_Y(1 - \alpha_{i}) \cdot (g(\alpha_{i+1}) - g(\alpha_i))$;
    \EndFor 
    \State {\bf Return:} $estimate$;
    \EndFunction
\end{algorithmic}
\caption{Importance Sampling DRM Estimation Algorithm}
\label{algo:ISDRMEsti}
\end{algorithm}

\subsection{Quantile Estimation with Importance Sampling} \label{subsec:QuantEstIS}

We begin with a quantile estimation technique that incorporates importance sampling, as proposed and studied in \cite{Glynn1996} and \cite{AhnShy2011}. For this we will first assume that the function $h$ is known; ML techniques in the simulation are discussed in Section~\ref{subsec:ImpIssues}.
Let $F$ be the distribution function of $X$, and $F^\ast$ some other distribution function on $\mathbb{R}^d$ such that $F$ is absolutely continuous with respect to $F^\ast$. We are interested in a $u$-quantile of $Y=h(X)$ for $u\in (0,1)$, if $X$ has distribution function $F$. Sampling $(X_i)_{i=1,2, \dots, N}$ independently from $F^\ast$, for any $u\in (0,1)$ a quantile estimator 
of $q_Y(u)$ is given by
\begin{align}
    \hat{q}_{F^\ast, N}(u) := \inf \left\{ x \in \mathbb{R} \left\vert \frac{1}{N} \sum_{h(X_i) > x } 
    \frac{dF}{dF^\ast} (X_i) \leq 1 - u \right\} \right., \quad u \in (0,1). 
    \label{form:QuantEstDefi}
\end{align} 
Conditions for the asymptotic normality of this estimator are provided in \cite{Glynn1996} and \cite{AhnShy2011} and stated in Appendix~\ref{app:QuantEstCLT}. More precisely,  denoting by $G$ and $G^\ast$ the distribution functions of $Y=h(X)$, if X is respectively distributed according to $F$ and $F^\ast$, \cite{AhnShy2011} show the following result.

\begin{theo}\label{theo:CLT}
	Suppose that Assumption~\ref{ass:ISDist} holds.
	Then for $u \in (0,1)$ we have
	$$ \sqrt{N} \left( \hat{q}_{F^\ast, N} (u) - q_Y (u) \right) 
	\overset{d}{\underset{N \rightarrow \infty}{\longrightarrow}}  \mathcal{N}
	\left( 0, \frac{\mathsf{E}_{F^\ast} \left[ \frac{dF}{dF^\ast}(X)^2 \mathds{1}_{\{h(X) \in( q_Y(u), \infty)\}} 
	\right] - (1 - u)^2}{G^\prime(q_Y(u))^2} \right).$$
\end{theo}

This theorem can now be leveraged to construct a sampling distribution ${F^\ast}$ that improves the efficiency of the Monte Carlo simulation. A classical choice for large, rare outcomes are exponential tilts. In our setting, $X$ is sampled, and we are interested in large outcomes of $Y=h(X)$. We thus consider the candidate family of sampling distribution
\begin{align}
    dF_\vartheta (x) = \exp \left( \vartheta h(x) - \psi(\vartheta) \right) dF(x), \label{form:SampleDist}
\end{align}
where $\vartheta \in \Theta \subseteq \mathbb{R}$ for some suitable neighbourhood $\Theta$ of $0$,  $\psi(\vartheta) := \log \left( \mathsf{E}_{F} [\exp ( \vartheta h(X))] \right)$. In order to minimize the variance in Theorem~\ref{theo:CLT}, \cite{SunHong2009} minimize a suitable upper bound and obtain the condition
\begin{align}
    q_Y(u) = \mathsf{E}_{F_\vartheta} [h(X)], \label{form:OptiTheta}
\end{align} 
which is used to determine a good parameter $\vartheta$. They prove that, under suitable technical conditions, the corresponding measure change reduces the variance of the estimator. We review this in Appendix~\ref{app:QuantEstCLT}. The implementation of eq.~\eqref{form:OptiTheta} requires knowledge of the quantile $q_Y(u)$ which is the value we seek to estimate. Moreover, the exact structure of $h$ and $\psi(\vartheta)$ are unknown. An algorithmic approach based on ML and MCMC to overcome these challenges is presented in Section~\ref{subsec:ImpIssues}.

\subsection{Discretization and Optimal Allocation of the Sampling Budget} \label{subsec:OptAllo}

The estimation of eq.~\eqref{form:DRMDecomp} requires a discretization of the two integrals involving the left- and right-continuous distortion functions. This distinction is only relevant, if $q_Y^+ ( 1 -u ) \neq q_Y(1 - u)$ for some $u$. In the numerical implementation, we suppose that  $q_Y^+ ( 1 -u ) = q_Y(1 - u)$ for all $u$ that appear in the discretization. This condition is consistent with the application of Theorem~\ref{theo:CLT} and Assumption~\ref{ass:ISDist} that forms the basis of the quantile approximation that we use. The requirement is essentially that the distribution function grows locally in a neighborhood of the quantiles.
We consider the approximation 
\begin{equation}\label{eq:dicretization}
    \hat{\rho}_g(Y ) = \sum_{i=0}^m \hat{q}_{F_{\vartheta_i^\ast}, N_i}(1 - \alpha_i) ( g(\alpha_{i+1} )- g(\alpha_i))
\end{equation}
of $\rho_g(Y)$,
defined for a partition $0 = \alpha_0 < \alpha_1 < \dots < \alpha_m < \alpha_{m+1} = 1$
where $\hat{q}_{F_{\vartheta_i^\ast}, N_i}(1- \alpha_i)$ are IS estimators according to Section~\ref{subsec:QuantEstIS}
with the sampling distributions $F_{\vartheta_i^\ast}$ and an allocated sampling budget $N_i$. For example, the partition $(\alpha_i)_{i=0,1, \dots, m}$ could be chosen uniformly in the region where $g$ grows, or one could choose $\alpha_i= g^{-1} \left(\frac{i}{m+1}\right)$, $i=0,1, \dots, m+1$, to adequately cover the levels where $g$ places weight. We suppose that the technical Assumption~\ref{ass:ISDist} is satisfied and that for each $i$ the number of samples $N_i$ is chosen large enough so that the $\hat{q}_{F_{\vartheta_i^\ast}, N_i}
(1 - \alpha_i)$ is finite according eq.~\eqref{prop:finiteQuantEst}. Using the notation $$ \hat{q}_Y ( 1 -u ) := \sum_{i=0}^m \mathds{1}_{\{ u \in [ \alpha_i, \alpha_{i+1})\}} 
\hat{q}_{F_{\vartheta_i^\ast}, N_i}(1 - \alpha_i) $$ for the estimator of the quantile function, we obtain
$ \hat{\rho}_g(Y) = \int_0^1 \hat{q}_Y(1 -u ) dg(u).$
Two questions have to be considered: First, how should the available $N$ samples be allocated to each quantile at the different levels? Second, should individual IS be used for each quantile, or should a single common measure change for pooled samples be preferred?

\subsubsection{Sample Allocation to Quantiles}\label{sec:saq}

Using Jensens's inequality, Fubini's theorem and Theorem~\ref{theo:CLT}, the MSE of the estimator can approximately be bounded above as follows:
{\small
\begin{align}
& \mathsf{E} \left[ (\rho_g(Y) - \hat{\rho}_g(Y))^2 \right] 
= \mathsf{E} \left[ \left( \int_0^1 q_Y(1 -u ) - \hat{q}_Y(1-u) dg(u) \right)^2 \right]\leq 
     \int_0^1 \mathsf{E} [ (q_Y(1 -u) - \hat{q}_Y(1 -u))^2 ] dg(u)\notag  \\&
    \approx \sum_{i=0}^m \int_{\alpha_i}^{\alpha_{i+1}} ( q_Y(1 -u) - q_Y(1 - \alpha_i))^2 dg(u) +
	\underbrace{\frac{\mathsf{E}_{\vartheta_i^\ast} \left[ \frac{dF}{dF_{\vartheta_i^\ast}}(X)^2 \mathds{1}_{\{h(X) > q_Y(1 - \alpha_i)\}} 
	\right] - \alpha_i^2}{N_i G^\prime(q_Y(1 - \alpha_i))^2}
    }_{:= \frac{V \left(1 - \alpha_i, F_{\vartheta_i^\ast} \right)}{N_i}}
	(g(\alpha_{i+1} ) - g(\alpha_i) ) 	  \notag \\
    & =: \mathcal{E}(F_{\vartheta_0^\ast}, F_{\vartheta_1^\ast}, \dots, F_{\vartheta_m^\ast}, \tilde{N}) \mbox{ with } \tilde{N}= (N_i)_{i=0,1, \dots, m}.  \label{eq:ErrorBound}
\end{align}} 
In the latter sum only the second summand  $\frac{V \left(1 - \alpha_i, F_{\vartheta_i^\ast} \right)}{N_i} (g(\alpha_{i+1}) - g(\alpha_i))$ depends
on $N_i$. Hence, when minimizing the approximate upper bound, the optimal allocation is obtained by minimizing $\sum_{i=0}^m \frac{V \left(1 - \alpha_i, F_{\vartheta_i^\ast} \right)}{N_i} (g(\alpha_{i+1}) - g(\alpha_i))$ under the constraint $\sum_{i=0}^m N_i = N$. This leads (up to rounding) to the solution 
\begin{equation}\label{theo:OptAllo}
    N_i^\ast =N \frac{\sqrt{c_i}}{\sum_{j=0}^m \sqrt{c_j}}, \quad i=0,1, \dots, m,
\end{equation}  
where 
$c_j := V \left(1 - \alpha_j, F_{\vartheta_j^\ast} \right)(g(\alpha_{j+1}) - g(\alpha_j)) , \; j \in \{0,1, \dots, m \}.$ The derivation of this result can be found in Appendix~\ref{app:OptAllo}. 

If the total sample size $N$ is not known in advance, eq.~\eqref{theo:OptAllo} determines the fraction $p_i:= \frac{N_i^\ast}{N}$ of the samples generated for each quantile. The total collection of all samples for these quantiles can also be viewed as samples from the mixture sampling distribution $F^\ast := \sum_{i=0}^m p_i \cdot F_{\vartheta_i^\ast}$, 
where $F_{\vartheta_i^\ast}$ are the sampling distributions for each individual quantile constructed in Section~\ref{subsec:QuantEstIS}.

\subsubsection{Efficient Use of the Samples in the Estimation of Multiple Quantiles}\label{subsec:EffUse}

The estimation of DRMs according to eq.~\eqref{eq:dicretization} requires the estimation of quantiles at the levels $1-u$ for $u= \alpha_0, \dots, \alpha_m$. We discuss whether individual importance sampling should be used, or a single common measure change for pooled samples is preferred. We assume in our comparison that the generation of individual samples is costly, but that the evaluation of the quantile estimators for given samples is comparatively inexpensive. We suppose that samples are allocated to the individual quantiles according to eq.~\eqref{theo:OptAllo} and that $F^*$ is the mixture sampling distribution in Section~\ref{subsec:OptAllo}. For each $i$, estimators of $q_Y(1-\alpha_i)$ are $\hat q_{F_{\vartheta_i^*},N_i^*} (1-\alpha_i)$ and $\hat q_{F^*,N} (1-\alpha_i)$ with approximate variances 
\begin{align*}
	\frac{V(1 - \alpha_i, F_{\vartheta_i^\ast})}{N^*_i} &= \frac{\mathsf{E}_{F_{\vartheta_i^\ast}} 
    \left[ \frac{dF} {dF_{\vartheta_i^\ast}}(X)^2 \mathds{1}_{\{h(X) > q_Y(1 - \alpha_i) \}} \right] 
    - \alpha_i^2 } {p_i \cdot N \cdot  G^\prime(q_Y(1 -\alpha_i))^2}, \\
	\frac{V(1 - \alpha_i, F^\ast)}{N} &= \frac{\mathsf{E}_{F^\ast} \left[ \frac{dF}{dF^\ast} (X)^2 
    \mathds{1}_{\{h(X) > q_Y(1 - \alpha_i) \}} \right] - \alpha_i^2} {N G^\prime (q_Y(1 - \alpha_i))^2},
\end{align*}
respectively, if the conditions of Theorem~\ref{theo:CLT} are satisfied. Hence, by comparing 
\begin{align*}
	\mathsf{E}_{F} \left[ \frac{dF}{dF_{\vartheta_i^\ast}} (X) \mathds{1}_{\{h(X) > q_Y(1 - \alpha_i) \}} \right]  - \alpha_i^2 \quad
	\text{to} \hspace{3mm} p_i \cdot \left(
	\mathsf{E}_{F} \left[ \frac{dF}{dF^\ast} (X) 
    \mathds{1}_{\{h(X) > q_Y(1 - \alpha_i) \}} \right] - \alpha_i^2 \right),  
\end{align*}
the preferred estimator can be selected. For simplicity, we propose using the mixture distribution for all quantile levels and implement it in all case studies in Sections~\ref{sec:CaseStudies} \& \ref{sec:ALM}. The corresponding estimator is 
\begin{equation*}
    \hat{\rho}_g(Y ) = \sum_{i=0}^m \hat{q}_{F^\ast, N}(1 - \alpha_i) ( g(\alpha_{i+1} )- g(\alpha_i)).
\end{equation*}

\subsection{Machine Learning and Implementation} \label{subsec:ImpIssues}

The objective of the simulation is to estimate $\rho_g(Y)$ for $Y= h(X)$ -- with the main challenge being the high cost of evaluating $h$. In order to apply importance sampling to this situation, we propose an algorithm with the following steps: First, pivot samples are used to compute parameters that govern importance sampling on the basis of exponential changes of measure. Second, the costly function $h$ is approximated by some auxiliary function $\hat h$ that simplifies the measure changes and allows to accelerate the generation of the corresponding samples. This step typically involves acceptance-rejection methods where the auxiliary function $\hat h$ allows to avoid the costly evaluation of $Y=h(X)$. Third, quantile estimators are computed for these samples for the original random variable $Y = h(X)$.
We discuss additional design nuances and details of the approach in Appendix \ref{app:Algo1}.
Let $(X_1, h(X_1)), \dots, (X_M, h(X_M))$ be pivot samples. According to eq. \eqref{form:OptiTheta}, an estimator $\hat{\vartheta}^\ast_i$ of  ${\vartheta}^\ast_i$ can be obtained by solving
$   \hat{q}_{F, M}(1 - \alpha_i ) = \frac{\sum_{i=1}^M h(X_i) \exp( \hat{\vartheta}^\ast_i h(X_i))}{
    \sum_{i=1}^M \exp ( \hat{\vartheta}^\ast_i h(X_i))}
$  numerically;
here, $\hat{q}_{F,M}(1 - \alpha_i)$ signifies the crude quantile estimator. A plug-in estimator of $\psi(\vartheta^*_i)$ is, for example, given by
$ \hat{\psi}_i = \log \left( \frac{1}{M} \sum_{j=1}^M \exp( \hat{\vartheta}^\ast_i h(X_j)) 
\right).$ 

To facilitate the generation of additional samples, we use ML-techniques trained on the pivot samples to approximate $h$ by a less costly function $\hat{h}$. More specifically, we consider linear and polynomial predictors, linear, polynomial and Gaussian support vector machines (SVMs) and $k$-nearest neighbor ($k$-NN) regressions in numerical case studies. To compare methods and parameters, $k$-fold validation is applied; we determine an approximation that yields the smallest MSE across splits of the 
training sample. For a brief review of the methods see Appendix \ref{app:MLTools} and 
\cite{ShalevSchwartzBenDavid2014}.

Importance sampling is based on the approximation $\hat{h}$, i.e., in \eqref{form:SampleDist} we use a Radon-Nikodym density proportional to $\exp \left( \hat{\vartheta}^\ast_i \hat{h}- \hat{\psi}_i\right) =: \hat{f}_{\hat{\vartheta}^\ast_i}$. The latter function might not itself be a probability density due to a potentially incorrect normalization, since  $\hat{\psi}_i$ was estimated from $h$ instead of $\hat{h}$. The correct normalization constant could be estimated at this stage, or Markov Chain Monte Carlo (MCMC) can directly be used to genererate more samples. Unless the proposal kernel is the independence kernel, MCMC does typically not preserve the independence of simulations, but appears to be quite efficient in our case studies. The Metropolis-Hastings algorithm with random walk proposal allows to either produce samples from a density proportional to $\hat{f}_{\hat{\vartheta}^\ast_i}$ 
or proportional to the mixture
$p_0 \hat{f}_{\hat{\vartheta}_0^\ast}  + \dots + p_m \hat{f}_{\hat{\vartheta}_m^\ast}$ with
 $(p_i)_{i=0,1, \dots, m}$ according to Section \ref{sec:saq}.

For the importance sampling quantile estimator (\ref{form:QuantEstDefi}) we need to evaluate the likelhood ratio which requires knowledge of the normalizing constants. To be more precise, we have
$
    \frac{dF}{d\hat{F}_{\hat{\vartheta}^\ast_i}} = \frac{z_i}{ \exp \left( \hat{\vartheta}^\ast_i \hat{h}(x) 
    - \hat{\psi}_i \right) },$ $
    z_i = \int \exp \left( \hat{\vartheta}^\ast_i \hat{h}(x) - \hat{\psi}_i  \right) F(dx)  
$ and
$    \frac{dF}{d \left( \sum_{i=0}^m p_i \hat{F}_{\hat{\vartheta}^\ast_i}\right) } 
    = \frac{z}{\sum_{i=0}^m p_i \exp \left( \hat{\vartheta}^\ast_i \hat{h}(x) - 
    \hat{\psi}_i \right) } $,
 $   z = \sum_{i=0}^m p_i z_i.$
The normalizing factors $z_i$ and $z$ could be approximated using simulations, but this might be costly given the desired accuracy. In low dimensions, we
can alternatively either use a trapezoidal rule 
with the $N+M$ samples as grid points or apply an adaptive quadrature rule explained in
\cite{Shampine2008}.  Since $F$ is the original distribution of the factor $X$, a suitable original model design might also ensure the applicability of such an approach for high dimensions $d$. Another strategy for estimating the normalizing constant relies on estimating the density function 
from the samples drawn; in this case, we assume that $F$ has density $f$ with respect to $d$-dimensional Lebesgue measure. Consider, for example, the mixture distribution, and let $\hat{f}_{\rm{mix}}$ be the estimated density, e.g., via kernel density 
estimation. Then for all $x \in \mathbb R^d$ we have that $z \approx \frac{\sum_{i=0}^m p_i 
\exp( \hat{\vartheta}^\ast_i \hat{h}(x)) f(x)}{\hat{f}_{\rm{mix}}(x)}$. Thus, $z$ can be estimated by computing the right hand side for several $x$ and taking an average. In the implementations we chose for each application the method which performed best in test cases. While the suggested approach works quite well in the considered numerical experiments, future research needs to further optimize the algorithm to guarantee good performance for high-dimensional random vectors $X$. A successful strategy could be to choose tractable pairs of ML-hypotheses classes on the one hand and the factor sampling distribution $F$ on the other hand that facilitate the implementation of measure changes.

\subsection{Time Efficiency of IS and Crude Estimations} \label{sec:TimeEff}

The proposed IS algorithm has higher computational complexity than a crude estimator of the same DRM. In this section, we identify conditions under which the IS algorithm is preferable. Specifically, we examine its efficiency when samples are drawn i.i.d. from the mixture IS distribution and analyze the error bound in (\ref{eq:ErrorBound}), while individual IS is deferred to Appendix \ref{app:TimeEff}. As we will see, when the IS estimator attains a smaller MSE and the evaluation of $h$ is costly, its additional computational overhead is offset by its efficiency gain.\\
We denote by $T_{CR}(N_{CR}, m)$ and $T_{IS}(M, N_{IS}, m)$ the computation times for estimating the DRM $\rho_g(Y)$. Here, $N_{CR}$ is the sample size used in the crude estimator, while $M$ and $N_{IS}$ denote the pivot and importance sampling sample sizes, respectively. The parameter $m$ represents the partition size. We assume that $N_{CR}, M,$ and $N_{IS}$ are chosen so that both methods satisfy the same error bound in (\ref{eq:ErrorBound}), and that $N_{CR} > M + N_{IS}$. This assumption is stated formally in the following lemma.
\begin{lem} \label{lem:TimeEff}
   Let $F^\ast = \sum_{i=0}^m p_i F_i$, where $F_i$ are distribution functions for $i \in \{0,1,\dots,m\}$. Assume that $F$ is absolutely continuous with respect to each $F_i$, and that Assumptions \ref{ass:ISDist} hold. If samples are drawn i.i.d. from $F^\ast$, then
\[
   \mathcal{E}(F,\dots,F,\tilde{N}_{CR}) 
   = \mathcal{E}(F^\ast,\dots,F^\ast,\tilde{N}_{IS})
   \;\;\Longleftrightarrow\;\;
   N_{CR} = N_{IS}\,
   \frac{\sum_{i=0}^m V(1-\alpha_i,F)\,\bigl(g(\alpha_{i+1})-g(\alpha_i)\bigr)}
        {\sum_{i=0}^m V(1-\alpha_i,F^\ast)\,\bigl(g(\alpha_{i+1})-g(\alpha_i)\bigr)}.
\]
\end{lem}
\begin{proof}
    When all samples from $F^\ast$ are used
    for the quantile estimation, we have $\tilde{N}_{CR} = (N_{CR}, \dots, N_{CR})^T$
    and $\tilde{N}_{IS} = (N_{IS}, \dots, N_{IS})^T$. Hence:
    \begin{align*}
        & \mathcal{E}(F, \dots, F, \tilde{N}_{CR}) 
        = \mathcal{E}(F^\ast, \dots, F^\ast, \tilde{N}_{IS}) \\
        \Leftrightarrow \quad 
        & \frac{1}{N_{CR}} \sum_{i=0}^m V(1 - \alpha_i, F)\, \bigl(g(\alpha_{i+1}) - g(\alpha_i)\bigr) 
         \quad = \frac{1}{N_{IS}} \sum_{i=0}^m V(1 - \alpha_i, F^\ast)\, \bigl(g(\alpha_{i+1}) - g(\alpha_i)\bigr) \\
        \Leftrightarrow \quad 
        & N_{CR} 
        = N_{IS} \cdot \frac{\sum_{i=0}^m V(1 - \alpha_i, F)\, \bigl(g(\alpha_{i+1}) - g(\alpha_i)\bigr)}%
        {\sum_{i=0}^m V(1 - \alpha_i, F^\ast)\, \bigl(g(\alpha_{i+1}) - g(\alpha_i)\bigr)}.
    \end{align*}
\end{proof}
For comparing the computational costs of the crude and IS algorithms, we define the following quantities:
\begin{itemize}
    \item[] $t_S(N)$: computation time required to sample $(X_i)_{i \in \{1,\dots,N\}}$ from $F$.
    \item[] $t_h(N)$: computation time required to evaluate $(h(X_i))_{i \in \{1,\dots,N\}}$.
    \item[] $t_\rho(N,m)$: computation time required to compute the quantile estimates and $\hat{\rho}_g(Y)$.
    \item[] $t_{\mathrm{MIX}}(M,m)$: computation time required to compute $\vartheta_i$, $\psi(\vartheta_i)$, and $p_i$ for $i \in \{0,\dots,m\}$ from the samples $(X_i)_{i \in \{1,\dots,M\}}$.
    \item[] $t_{\mathrm{kFold}}(M)$: computation time required for $k$-fold validation based on the samples $(X_i)_{i \in \{1,\dots,M\}}$.
    \item[] $t_{\mathrm{MH}}(N)$: computation time required to sample $(X_i)_{i \in \{1,\dots,N\}}$ from the IS distribution using the Metropolis--Hastings algorithm.
    \item[] $t_{\mathrm{Norm}}(m)$: computation time required to estimate the normalizing constant.
\end{itemize}
We assume that the computation times are increasing in their arguments, and that 
$t_S(\cdot)$ and $t_h(\cdot)$ are linear functions. With these definitions, we can 
express the estimation methods together with their computation times. For the 
purpose of comparison, the estimation methods are decomposed into the following steps:
\footnotesize \\
\begin{minipage}[t]{0.41\textwidth}
    \vspace{3mm}
    \underline{Crude Method:}
\begin{itemize}[leftmargin=*, labelwidth=4.5em, align=left]
    \item[$t_S(N_{CR})$:] Draw samples $(X_i)_{i \in \{1, \dots, N_{CR}\}}$ from $F$.
    \item[$t_h(N_{CR})$:] Evaluate $h$ to obtain $(h(X_i))_{i \in \{1, \dots, N_{CR}\}}$.
    \item[$t_\rho(N_{CR}, m)$:] Compute the quantile estimators 
    $\hat{q}_{F, N_{CR}}(1-\alpha_i)$, $i \in \{0, \dots, m\}$, and the estimate 
    $\hat{\rho}_g(Y)$.
\end{itemize}
\vspace{3mm}
\end{minipage} \hspace{10mm}
\begin{minipage}[t]{0.41\textwidth}
    \vspace{3mm}
\underline{IS Method:}
\begin{itemize}[leftmargin=*, labelwidth=4.5em, align=left]
    \item[$t_S(M)$:] Draw samples $(X_i)_{i \in \{1, \dots, M\}}$ from $F$.
    \item[$t_h(M)$:] Evaluate $h$ to obtain $(h(X_i))_{i \in \{1, \dots, M\}}$.
    \item[$t_{Mix}(M, m)$:] Estimate $\vartheta_i^\ast$, $\psi(\vartheta_i^\ast)$, and $p_i$ for 
    $i \in \{0, \dots, m\}$.
    \item[$t_{kFold}(M)$:] Perform $k$-fold validation based on $(h(X_i))_{i \in \{1, \dots, M\}}$.
    \item[$t_{MH}(N_{IS})$:] Generate samples $(X_i^\prime)_{i \in \{1, \dots, N_{IS}\}}$ from the IS distribution.
    \item[$t_h(N_{IS})$:] Evaluate $h$ to obtain $(h(X_i^\prime))_{i \in \{1, \dots, N_{IS}\}}$.
    \item[$t_{Norm}(m)$:] Estimate the normalizing constant.
    \item[$t_\rho(N_{IS}, m)$:] Compute the quantile estimators $\hat{q}_{F^\ast, N_{IS}}(1-\alpha_i)$,
    $i \in \{0, \dots, m\}$, and the estimate $\hat{\rho}_g(Y)$.
\end{itemize}
    \vspace{3mm}
\end{minipage}
\normalsize
\newline
The total computation times of the crude and IS estimations are
\begin{align*}
    T_{CR}(N_{CR}, m) 
        &:= t_S(N_{CR}) + t_h(N_{CR}) + t_\rho(N_{CR}, m), \\
    T_{IS}(M, N_{IS}, m) 
        &:= t_S(M) + t_h(M) + t_{Mix}(M, m) + t_{kFold}(M) + t_{MH}(N_{IS}) \\
        &\quad + t_h(N_{IS}) + t_{Norm}(m) + t_\rho(N_{IS}, m) \\
        &= t_S(M) + t_h(M+N_{IS}) + t_{Mix}(M, m) + t_{kFold}(M) \\
        &\quad + t_{MH}(N_{IS}) + t_{Norm}(m) + t_\rho(N_{IS}, m).
\end{align*}
\begin{prop} \label{prop:TimeEff}
Suppose $N_{CR} > M+N_{IS}$. If
\[
    t_h\bigl(N_{CR}-(M+N_{IS})\bigr) 
        > t_{Mix}(M,m) + t_{kFold}(M) + t_{MH}(N_{IS}) + t_{Norm}(m),
\]
then
\[
    T_{CR}(N_{CR}, m) - T_{IS}(M, N_{IS}, m) > 0,
\]
that is, the crude method requires more computation time than IS.
\end{prop}
\begin{proof}
We may write
\begin{align*}
    T_{CR}(N_{CR},m) - T_{IS}(M,N_{IS},m)
    &= t_S(N_{CR}-M) 
       - t_h\!\left(N_{CR}-(M+N_{IS})\right) \\
    &\quad - t_{Mix}(M,m) + t_{kFold}(M) + t_{MH}(N_{IS}) + t_{Norm}(m) \\
    &\quad + \bigl[t_\rho(N_{CR},m) - t_\rho(N_{IS},m)\bigr].
\end{align*}
Since $N_{CR} > M+N_{IS}$, the monotonicity of $t_S(\cdot)$ and $t_\rho(\cdot)$ yields 
$t_S(N_{CR}-M) > 0$ and $t_\rho(N_{CR},m) - t_\rho(N_{IS},m) > 0$. 
Hence, if
\[
    t_h\!\left(N_{CR}-(M+N_{IS})\right) 
      > t_{Mix}(M,m) + t_{kFold}(M) + t_{MH}(N_{IS}) + t_{Norm}(m),
\]
then $T_{CR}(N_{CR},m) - T_{IS}(M,N_{IS},m) > 0$, as claimed.
\end{proof}

\section{Case Studies} \label{sec:CaseStudies}

In this section, we apply the developed method to various test models and distributions. The goal is to experimentally evaluate the variance reduction achieved by the proposed algorithm compared to importance sampling in the exact model, which is known in closed form for the test cases. We compare the root mean square errors (RMSEs) when estimating different DRMs that model both risk-averse and risk-seeking attitudes.

\subsection{Simulation Design} \label{subsec:StudiesDesign}

We consider the distortion function $g_{\gamma, \alpha} (u) 
=\mathds{1}_{\{u \in [0, \alpha] \}} \left( \frac{u}{\alpha} \right)^{\gamma } + \mathds{1}_{\{u \in ( 
\alpha, 1] \}}$  with $\gamma \in \{1/2, 1, 2\}$ illustrated in 
Figure~\ref{fig:distFuncIllu}, see Example~\ref{ex:DRM} in the Appendix for more details.
The concave function $g_{1/2,\alpha}$ defines a convex DRM that models a risk-averse attitude. Conversely, the 
 function $g_{2, \alpha}$ is convex on the interval $[0, \alpha]$ and models a risk-seeking attitude. The function $g_{1, \alpha}$ corresponds to the Average Value at Risk (AV@R) at level $\alpha$. 
The AV@R, also known as Expected Shortfall, is particularly important in practice, since it serves as the foundation for various solvency regimes.  For additional details and references see Appendix \ref{app:DRM}.

\begin{figure}[h]
    \centering
    \includegraphics[scale = 0.5]{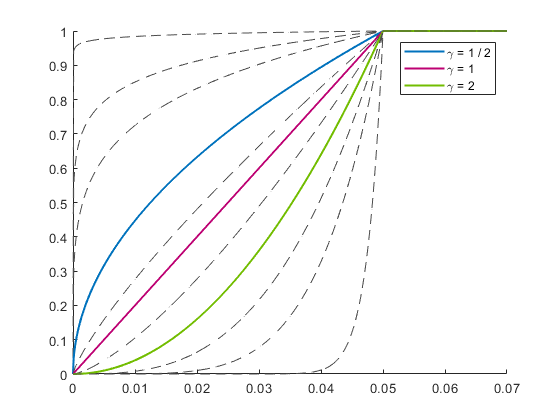}
    \caption{Different distortion functions $g_{\alpha, \gamma} (u)$ with $\alpha = 0.05$ and 
    $\gamma \in \{0.01,0.1, 0.2, 0.5, 0.8, 1, 1.4, 2, 3, 5, 20\}$. The distortion function used in the case studies are highlighted.}
    \label{fig:distFuncIllu}
\end{figure}

In our numerical experiments, we repeat Algorithm~\ref{algo:ISDRMEsti} for DRM estimation $R$ times to obtain data that can be further analyzed. For clear benchmarking, we specify the hypothesis class used in each of these experiments and compare the results across hypothesis classes. Thus, we do not perform the $k$-fold validation in line 15 of Algorithm~\ref{algo:ISDRMEsti} in each repetition, but only recalibrate within any previously selected class. Additionally, we identify the winner of a $k$-fold validation with $k=20$ based on $M=2000$ pivot samples. Numerical tests in the context of our case studies show that this determination of a ML hypothesis class is quite robust, i.e., different sets of $M=2000$ pivot samples typically lead to the selection of the same class.

We consider the following functions and distributions:
\begin{itemize}
    \item[(1)] \textit{Identity of Normal:} 
    We set $X \sim \mathcal{N}(0,1)$ and $h(x) = x$, implying $Y \sim \mathcal{N}(0,1)$. 
    The $k$-fold cross validation from the pivot samples suggests using a linear regression to approximate 
    $h$.
    \item[(2)] \textit{Sum of Normals:} We consider $X_1, X_2 \sim \mathcal{N}(0, 1)$
    with $\mathrm{Corr}(X_1, X_2) = 0.3$ and $h(x_1, x_2) = x_1 + x_2$. The $k$-fold cross validation
    suggests a linear regression. 
    \item[(3)] \textit{Product of Normals:}  
    Let $X_1, X_2 \sim \mathcal{N}(2,1)$ with $\mathrm{Corr}(X_1, X_2) = - 0.3$
    and $h(x_1, x_2) = x_1 \cdot x_2$. 
    The $k$-fold validation identifies the SVM with a polynomial kernel of degree 
    $2$ as the optimal choice for $\hat{h}$. 
    \item[(4)] \textit{Sum of Squared Normals:}
    Consider the independent random variables $X_1, X_2, X_3, X_4 \sim \mathcal{N}(0,1)$ and let
    $h(x_1, x_2, x_3, x_4) = x_1^2 + x_2^2 + x_3^2 + x_4^2$. Then $h(X_1, X_2, X_3, X_4)$ follows
    a $\chi^2$ distribution with $4$ degrees of freedom. The $k$-fold validation 
    suggests a polynomial SVM with degree $2$.     
    \item[(5)] \textit{Sine and Uniform:}
    Set $X \sim \mathrm{unif}(0,1)$ and $h(x) = x \sin(2.5 \cdot \pi \cdot x)$. This
    example is used in \cite{Altman1992} to illustrate $k$-NN 
    regression. The $k$-fold validation suggests polynomial regression with degree  
    $5$ as an approximation of $h$. 
    \item[(6)] \textit{Logistic Transformation and Exponential:}
    Letting $X \sim \exp(1)$, then $h(X)$ with $h(x) = - \log \left( e^{-x} /
    (1 - e^{-x}) \right)$ follows a $\mathrm{Logistic}(0,1)$ distribution. 
    The $k$-fold validation suggests either a SVM with Gaussian kernel or 
    $k$-NN regression with $k = 1$. 
\end{itemize}

For each of these functions and distributions, we perform numerical experiments for different ML hypothesis classes used to approximate $h$ in the IS algorithms. In particular, this analysis is also performed for the winner of the $k$-fold validations. Each experiment is repeated $R$ times. In all cases, we implement a crude estimation with $M + N$ samples as well as an estimation based on the IS method with $M$ pivot samples and $N$ samples from the mixture distribution defined in  Section~\ref{subsec:OptAllo}. As a benchmark, we determine an ``exact value" by a crude estimation with $10,000,000$ samples. From the replications we calculate for all cases an estimated RMSE.

\subsection{Results} \label{subsec:StudiesResults}

\subsubsection{Distribution of the Samples}\label{sec:SampleDistApp}

To illustrate the measure change, we consider model (5) in Figure~\ref{fig:SampleDist}. As mentioned before, option 2 in line 23 of Algorithm~\ref{algo:ISDRMEsti} was implemented in all case studies. An analogous analysis for models (1)-(4) \& (6) can be found in Figures~\ref{fig:SampleDistApp} \& \ref{fig:SampleDistApp3}. We consider the DRMs $\rho_{g_{\alpha, \gamma}}(Y)$, 
$\alpha = 0.05, \gamma \in \{1 / 2, 1 , 2\}$. The figures show the true density of $Y$. Additionally, $200$ samples from the mixture distribution (with values on the x-axis) are plotted along the probability density. The labeled quantile $q_Y(0.05)$ indicates the threshold above which samples are relevant 
for the estimation of $\rho_{g_{\alpha, \gamma}}(Y)$. The crosses mark the expectations of 
the individual importance sampling components of the mixture distribution. By design, samples and expectations are in the right tail of the distribution in the area relevant for the estimation of the DRM.

\begin{figure}
    \centering
      \begin{minipage}[c]{0.32\textwidth}
        \centering
        \includegraphics[width=1.11\textwidth]{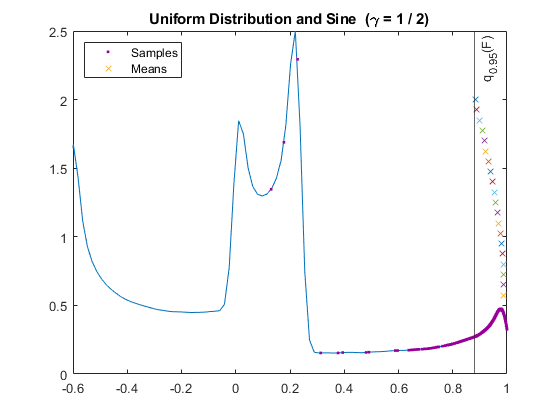}
    \end{minipage}
    \begin{minipage}[c]{0.32\textwidth}
        \centering
        \includegraphics[width=1.11\textwidth]{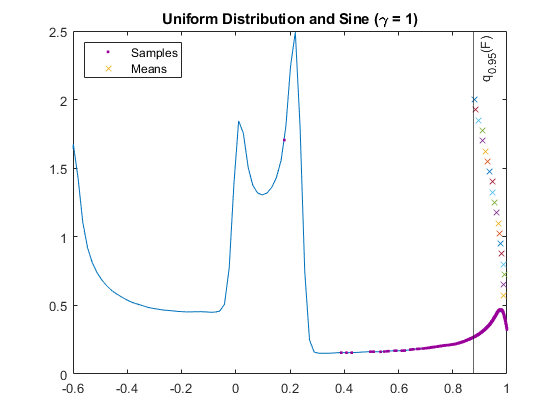}
    \end{minipage}
    \begin{minipage}[c]{0.32\textwidth}
        \centering
        \includegraphics[width=1.11\textwidth]{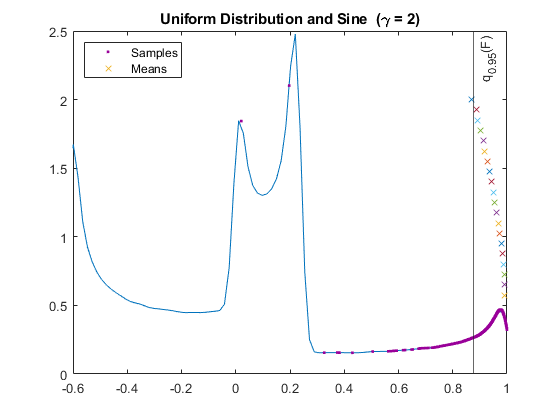}
    \end{minipage}
    \caption{An example of $200$ samples drawn from the mixture distribution plotted on the underlying 
    distribution of the model $Y$ for the case study (5). To approximate the mixture weights and
    optimal mixture components $M = 20,000$ pivot samples were drawn.}
    \label{fig:SampleDist}
\end{figure}

\subsubsection{Efficiency of the Estimations}

In this section, we discuss the efficiency of the algorithm for case studies (1)-(6). The ML approximation used in importance sampling is fixed, and we compare different hypothesis classes. The calibration of the ML regressions and the construction of the importance sampling measure change are based on $M = 2,000$ pivot samples. We choose $m = 20$ and $N = 20,000$ and run $R = 1,000$ replications of the simulation for estimating the RMSE. In Figure~\ref{fig:ErrorRatio} \& \ref{fig:ErrorRatio2}, the ratio of the RMSE between the crude estimate and the
proposed importance sampling method for the DRMs $\rho_{g_{\alpha, \gamma}}$ with $\gamma \in \{1/2, 1, 2\}$ and $\alpha \in [0.01, 0.3]$ for models (1) to (6).
The absolute estimated RMSE for the different estimation methods is shown in Figure~\ref{fig:LevelError} \& \ref{fig:LevelError3} in Appendix~\ref{app:AddPlots}.

\begin{figure}
    \centering
    \begin{minipage}[c]{0.32\textwidth}
        \centering
        \includegraphics[width=1.11\textwidth]{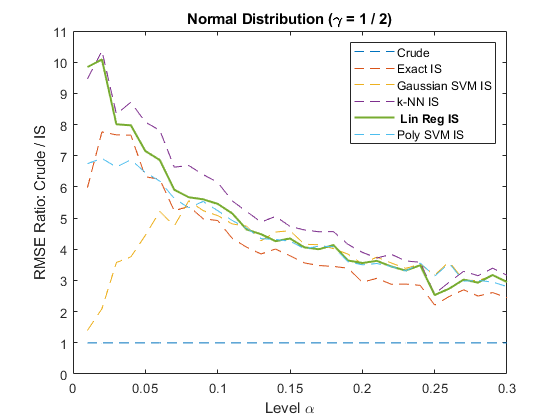}
    \end{minipage}
    \begin{minipage}[c]{0.32\textwidth}
        \centering
        \includegraphics[width=1.11\textwidth]{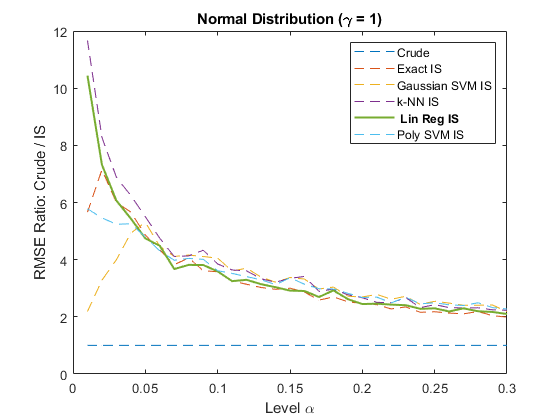}
    \end{minipage}
     \begin{minipage}[c]{0.32\textwidth}
        \centering
        \includegraphics[width=1.11\textwidth]{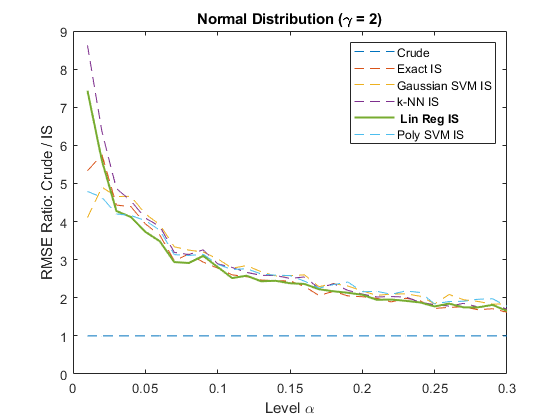}
    \end{minipage}
    \begin{minipage}[c]{0.32\textwidth}
        \centering
        \includegraphics[width=1.11\textwidth]{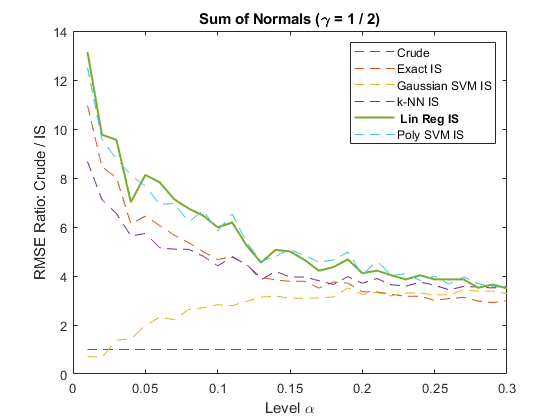}
    \end{minipage}
    \begin{minipage}[c]{0.32\textwidth}
        \centering
        \includegraphics[width=1.11\textwidth]{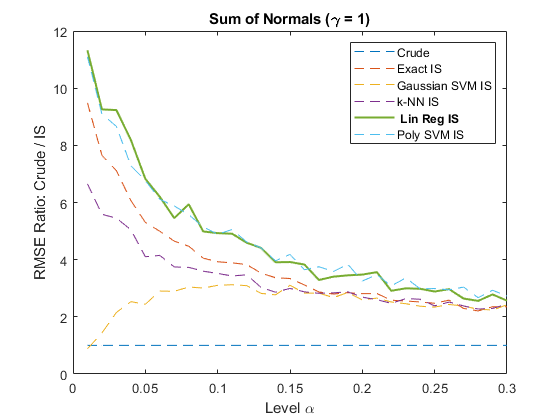}
    \end{minipage}
    \begin{minipage}[c]{0.32\textwidth}
        \centering
        \includegraphics[width=1.11\textwidth]{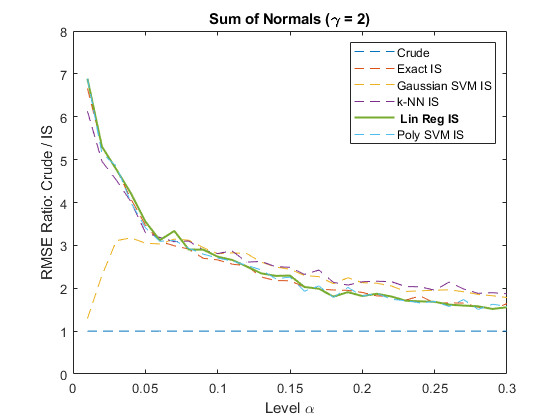}
    \end{minipage}
     \begin{minipage}[c]{0.32\textwidth}
        \centering
        \includegraphics[width=1.11\textwidth]{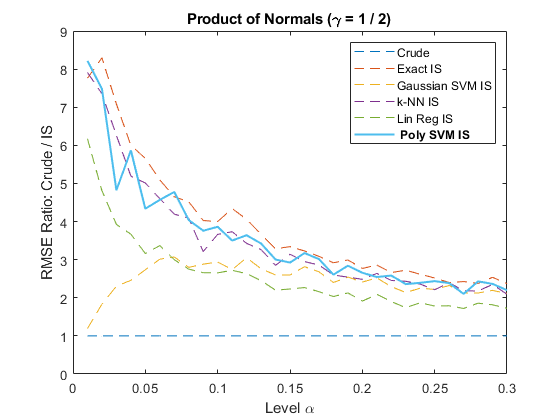}
    \end{minipage}
    \begin{minipage}[c]{0.32\textwidth}
        \centering
        \includegraphics[width=1.11\textwidth]{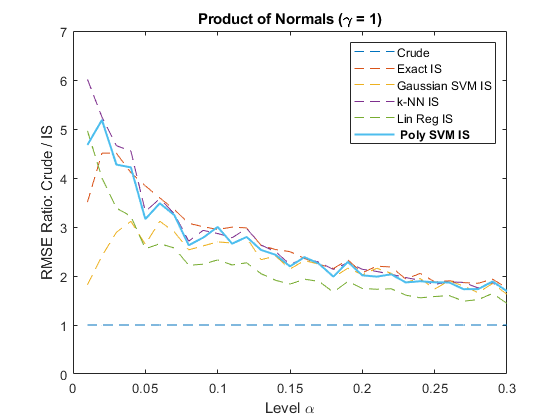}
    \end{minipage}
     \begin{minipage}[c]{0.32\textwidth}
        \centering
        \includegraphics[width=1.11\textwidth]{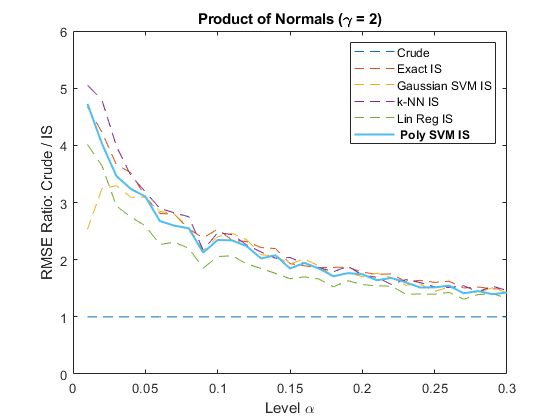}
    \end{minipage}
    \begin{minipage}[c]{0.32\textwidth}
        \centering
        \includegraphics[width=1.11\textwidth]{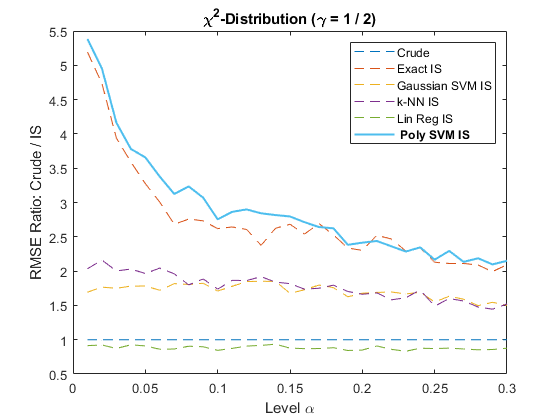}
    \end{minipage}
    \begin{minipage}[c]{0.32\textwidth}
        \centering
        \includegraphics[width=1.11\textwidth]{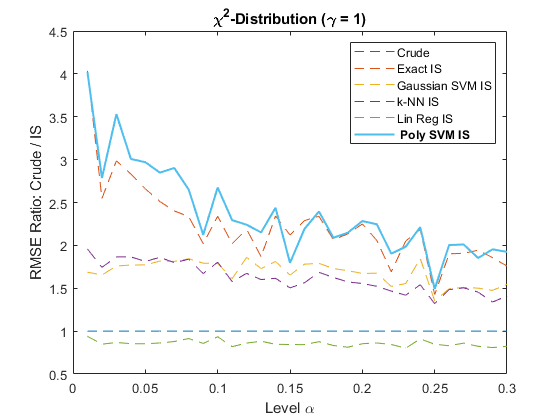}
    \end{minipage}
    \begin{minipage}[c]{0.32\textwidth}
        \centering
        \includegraphics[width=1.11\textwidth]{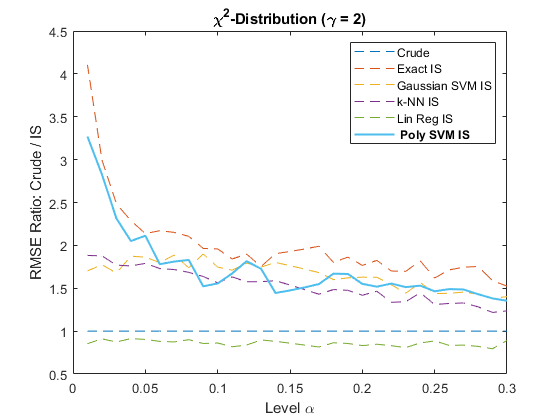}
    \end{minipage}
    \caption{Ratio of the RMSE arising between the crude method and the IS method
    when estimating $\rho_{g_{\gamma, \alpha}}$, with $\gamma \in \{1/2, 1, 2\}$, $\alpha \in [0.01, 0.3]$,
    for the models (1) to (6). The comparison is made between the crude method and the proposed 
    IS method using various approximations for the black box considered in the paper.}
    \label{fig:ErrorRatio}
\end{figure}

\begin{figure}
    \centering
    \begin{minipage}[c]{0.32\textwidth}
        \centering
        \includegraphics[width=1.11\textwidth]{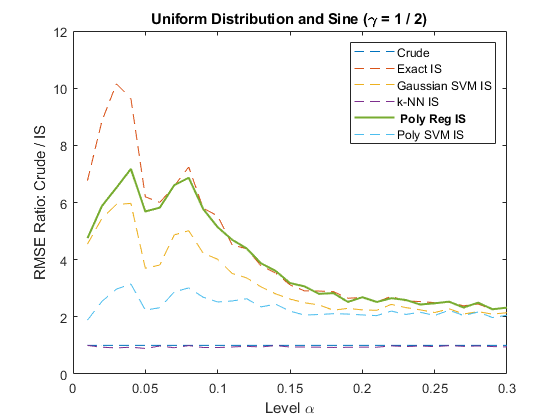}
    \end{minipage}
    \begin{minipage}[c]{0.32\textwidth}
        \centering
        \includegraphics[width=1.11\textwidth]{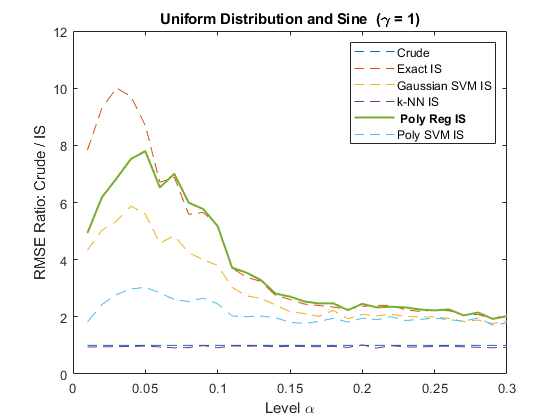}
    \end{minipage}
     \begin{minipage}[c]{0.32\textwidth}
        \centering
        \includegraphics[width=1.11\textwidth]{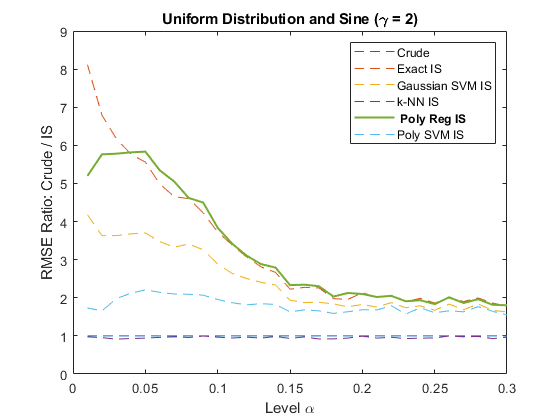}
    \end{minipage}
    \begin{minipage}[c]{0.32\textwidth}
        \centering
        \includegraphics[width=1.11\textwidth]{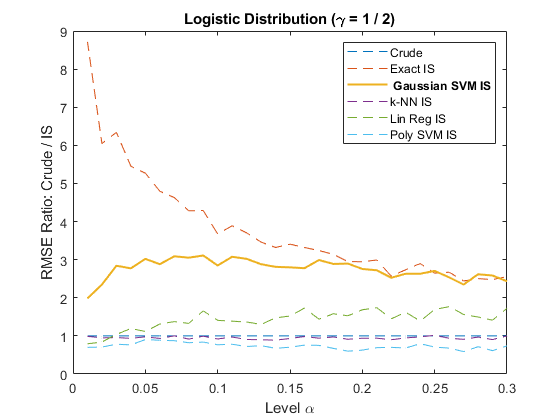}
    \end{minipage}
    \begin{minipage}[c]{0.32\textwidth}
        \centering
        \includegraphics[width=1.11\textwidth]{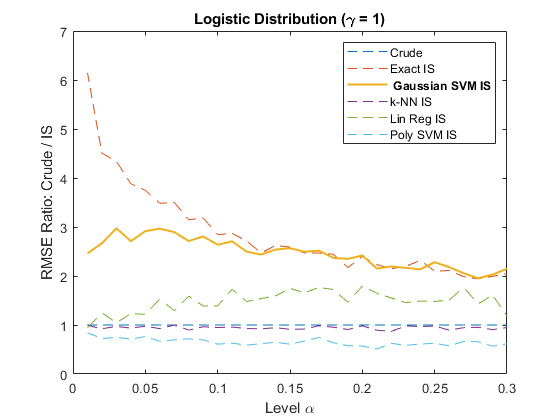}
    \end{minipage}
    \begin{minipage}[c]{0.32\textwidth}
        \centering
        \includegraphics[width=1.11\textwidth]{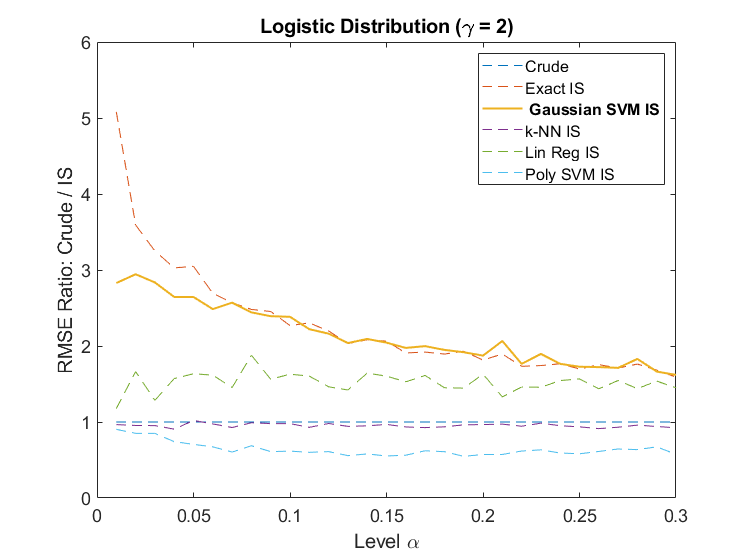}
    \end{minipage}
    \caption{Continuation of Figure \ref{fig:ErrorRatio}.}
    \label{fig:ErrorRatio2}
\end{figure}

The plots confirm that the proposed importance sampling algorithm can successfully reduce the RMSE in all cases. The efficiency of the algorithm depends significantly on the chosen ML regression model. A poorly chosen approximation can even lead to a higher error than the crude method. Interestingly, the choice based on $k$-fold validation and pivot calibration works reasonably well in all cases, despite the fact that the ML objective function does not focus on the tail. The smaller $\alpha$, the more the DRM zooms in on the tail risk due to rare events. As expected, the variance reduction becomes better the smaller $\alpha$ is. Similarly, variance reduction is also better the smaller $\gamma$, since DRMs with smaller $\gamma$ put more emphasis on tail risk. 

\subsection{Iterative Exploration of the Extreme Tail}\label{sec:itexext}

The discussed algorithm consists of two simulation steps. First, pivot samples are drawn that are used for both the choice of the ML approximation and the determination of an IS measure change. Second, samples are generated under the IS distribution and used for the estimation of the DRMs. However, if DRMs are considered that focus on particularly extreme tail events, this approach might not yet be sufficient. A possible extension to the suggested approach is the following: The samples from the IS distribution are not directly used for DRM estimation, but serve as additional data for calibrating the ML approximation a second time and the construction of a further measure change on this basis. 
By using samples from the IS distribution for the ML approximation, our goal is to obtain more samples in the right tail, thereby enhancing the ML approximation's accuracy in that region. This results in a more suitable IS distribution for the subsequent measure change in DRM estimation.
In this section, we provide a case study that takes this approach -- the iterative exploration of the extreme tail.

\subsubsection*{Simulation Desgin}

We consider again the case studies (1) - (4) outlined in Section \ref{subsec:StudiesDesign} with the same 
distortion functions $g_{\alpha, \gamma } (u)$, $\gamma \in \{1/2, 1, 2\}$
and corresponding DRMs. However, we focus on more extreme tail events by choosing $\alpha = 0.002$. 
In addition, we choose a finer partition by setting $m = 50$. In the experiment, we repeat all simulation runs $R = 2,000$ times to estimate the RMSE. We compare three simulation approaches for the estimation of the DRMs. In all cases, $27,500$ samples are used, respectively.

The first approach is a crude simulation with $27,500$ samples. The second approach is the algorithm suggested in the previous sections with $7,500$ pivot samples and $20,000$ 
samples from the IS distribution. The third approach is an iterative exploration: $5000$ pivot samples are used to calculate the IS distribution for level $\alpha^\prime = 0.01$. Then we draw from this IS distribution $2,500$ additional pivot samples. The IS distribution for $\alpha = 0.002$ is computed from the total $7,500$ pivot samples. In the last step, $20,000$ are drawn from this distribution to estimate the DRMs.

\subsection*{Results}

\begin{table}[]
    \renewcommand{\arraystretch}{1.5}
    \scriptsize
    \centering
    \begin{tabular}{| l || c | c | c || c | c | c || c | c | c || c | c | c |}
        \hline
        & \multicolumn{3}{|c|}{(1) Id. of Normals} & \multicolumn{3}{|c|}{(2) Sum of Normals} & \multicolumn{3}{|c|}{(3) Prod. of Normals} & \multicolumn{3}{|c|}{(4) Sum of Sq. Normals} \\
        \hline 
        $\gamma$ & $1/2$ & $1$ & $2$ & $1/2$ & $1$ & $2$ & $1/2$ & $1$ & $2$ & $1/2$ & $1$ & $2$ \\
        \hline
        Exact & $3.35$ & $3.16$ & $3.02$ & $5.40$ & $5.09$ & $4.88$ & $4.08$ & $3.60$ & $3.30$ & $20.58$ & $19.01$ & $17.99$  \\
        \hline
        Mean CRUDE & $3.33$ & $3.15$ & $3.03$ & $5.38$ & $5.08$ & $4.88$ & $4.02$ & $3.59$ & $3.30$& $20.45$  & $18.99$ & $18.00$  \\
        Mean IS & $3.35$ & $3.16$ & $3.02$ & $5.40$ & $5.09$ & $4.88$ & $4.07$ & $3.60$ & $3.30$ & $20.48$ & $18.91$ & $17.89$  \\
        Mean ITER IS & $3.35$ & $3.16$ & $3.02$ & $5.40$ & $5.09$ & $4.88$ & $4.06$ & $3.59$ & $3.29$ & $20.55$ & $18.98$ & $17.97$  \\
        \hline
        RMSE $\frac{\text{CRUDE}}{\text{IS}}$ & $35.61$ & $20.63$ & $14.87$ & $16.73$ & $13.09$ & $10.34$ & $10.90$ & $9.40$ & $7.48$ & $5.34$ & $3.58$ &  $2.97$ \\
        RMSE $\frac{\text{CRUDE}}{\text{ITER IS}}$  & $35.42$ & $20.81$ & $14.41$ & $35.02$ & $21.70$ & $15.04$ & $9.61$ & $8.39$ & $6.71$ & $13.55$ & $8.89$  & $7.17$ \\
        \hline
    \end{tabular}
    \caption{Results of the DRM estimation in the extreme tail.}
    \label{tab:TailStudy}
\end{table}

The results of the case study are displayed in Table \ref{tab:TailStudy}. The
exact values of the DRMs, the means over $R = 2,000$ simulation runs and the corresponding ratios of the RMSE of the two IS methods and the crude method are documented. 
Overall the iterative method typically provides the most substantial RMSE reduction,
outperforming the direct IS approach substantially in experiments (2) and (4). The direct IS approach is still more efficient than the crude method. 
Especially in (1) and (2) the reduction of the RMSE is significant in contrast to the crude method,
while in (3) and (4) the IS methods are not as efficient. When considering the mean over all simulation runs, we observe that the IS methods also reduce estimation bias. 

\section{Application to ALM} \label{sec:ALM}

We apply Algorithm~\ref{algo:ISDRMEsti} to the estimation of solvency capital in a simple asset-liability management (ALM) model of an insurance firm. Instead of the risk measure $V@R$ which forms the basis of Solvency II, we use the same DRMs that we considered in Section~\ref{sec:CaseStudies}. The suggested method could also be applied in highly complex ALM models such as those applied by major insurance groups.

\subsection{Model Description} \label{subsec:ALMModel}

Our ALM model, inspired by \cite{WeberHamm2014} and \cite{HammKnispelWeber2019}, describes a snapshot in time of an ongoing insurance business. The focus is on a one-year time horizon with dates $t=0,1$, as in Solvency II. The values of assets and liabilities are denoted by $A_t, L_t$, $t=0,1$, respectively. At each point in time, their difference is the book value of equity $E_t = A_t-L_t$, $t=0,1$, which is used for the solvency capital calculation.

The evolution of balance sheet is driven by market and insurance risks. For simplicity, we assume that reserves are constant, i.e., $L_t=v$, $\forall t$. Any changes in value are thus seen on the asset side. We assume that insurance claims are modeled by a collective model where the number of claims are given by the counting process $N$ and their severities by independent, identically distributed losses $\xi_k$. Annual total premium payments $\pi$ are received at the beginning of the year. We set 
$$ C \;=\; \sum_{k=1}^N \xi_k. $$
The random annual return of assets between dates $t=0$ and $t=1$ is denoted by $R_A$, i.e., we obtain that
$$A_1 \;=\; R_A \cdot A_0 - C + \pi .$$
In order to model the random return of assets we assume that 
$$ A_0 \;=\;  \eta^S S_0 + \eta^B B_0,$$
where $B = (B_t)_{t \in \{0,1\}}$ and
$S = (S_t)_{t \in \{0,1\}}$ are the prices of a bond and a stock and $\eta^S$ and $\eta^B$ the respective holdings. This implies that $$R_A \;=\; \frac{\eta^S S_1 + \eta^B B_1}{A_0} \;=\; \frac{\eta^S S_0}{A_0} \cdot \frac{S_1}{S_0} + \frac{\eta^B B_0}{A_0}\cdot \frac{B_1}{B_0} \;= \;b\cdot \frac{S_1}{S_0} + (1-b) \cdot\frac{B_1}{B_0},$$ where $b$ is the fraction of initial wealth invested in the stock.
Setting $B_0 = 1$, $B_1 = 1 + R_B$ for some random interest rate $R_B> -1$, and  $S_0 = 1$, 
$$ S_1 \;= \;\exp \left( \left( \mu - \frac{\sigma^2}{2} \right) \Delta t + \sigma 
\sqrt{\Delta t} Z \right),$$
we derive that
$$R_A \;=\; b \cdot  \exp \left( \left( \mu - \frac{\sigma^2}{2} \right) \Delta t + \sigma 
\sqrt{\Delta t} Z \right)  + (1-b)  \cdot (1+R_B) .$$
Solvency capital is determined in terms of risk measure applied to the change in net asset values over the time period of one year, i.e.,
$$E_1-E_0 \;=\; A_1 - \underbrace{L_1}_{=v} - A_0 + \underbrace{L_1}_{=v}\; =\; (R_A -1 ) \cdot A_0 -C + \pi .$$

\subsection{Simulation Overview} \label{subsec:ALMSimOverview}

As in Section~\ref{sec:CaseStudies}, we apply the proposed 
importance sampling method to the DRMs  $\rho_{g_{\alpha, \gamma}}$ with
distortion function $g_{\alpha, \gamma}(u) = 
\mathds{1}_{[0, \alpha ] } (u) \left( \frac{u}{\alpha} \right)^{\gamma } 
+ \mathds{1}_{ ( \alpha, 1 ] } (u)$ and $\gamma \in \{1/2, 1, 2\}$. 
In terms of the DRMs, solvency capital is $\rho_{g_{\alpha, \gamma }} ( E_0 - E_1)$. The underlying random factors 
are $R_B, Z, N, \xi_1, \xi_2, \dots$. 
We set $E_0 = 1000$ and $R_B = (V - 1 /2) / 10$, where $V$ is beta 
distributed with parameters $(2,2)$, i.e., $V \sim B(2,2)$. 
The parameters of the stock are $\mu = 0.02$, $\sigma =  0.2$, $\Delta t = 1$, and we assume that 
half of the available capital is invested into the stock, i.e., $b = 0.5$. 
For the collective model, we assume that $N$ is a Poisson random variable with parameter
$\lambda = 5$, and $(\xi_k)_{k\geq 1}$  are independent exponentially distributed random variables with
parameter $\vartheta^\prime = 10$.
The premium and reserve are $103 \%$ resp. $105 \%$ of the expected claims such that
$\pi = 1.03 \lambda \vartheta^\prime$ and $v = 1.05 \lambda \vartheta^\prime$. 

As in Section~\ref{sec:CaseStudies}, the importance sampling estimates with the different 
ML approximations used in the measure changes are performed with $M = 2,000$ pivot samples for calibration of the approximation and determination of the importance sampling mixture distribution, and $N = 20,000$ samples of the mixture distribution for DRM estimation. As comparison a crude estimation with $M + N $ 
samples is implemented. We always use a discretization with $m = 20$. The `exact value" of benchmarking is determined 
with a crude estimation with $1,000,000$ samples and used to calculate the RMSE over 
$R = 1,000$ simulation runs.

\subsection{Results} \label{subsec:ALMREs}

\begin{figure}
    \centering
    \begin{minipage}[c]{0.32\textwidth}
        \centering
        \includegraphics[width=1.11\textwidth]{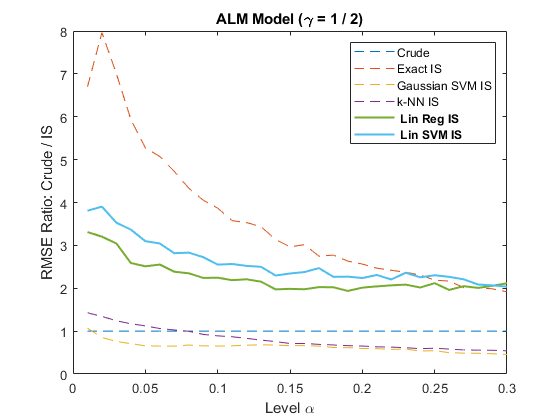}
    \end{minipage}
    \begin{minipage}[c]{0.32\textwidth}
        \centering
        \includegraphics[width=1.11\textwidth]{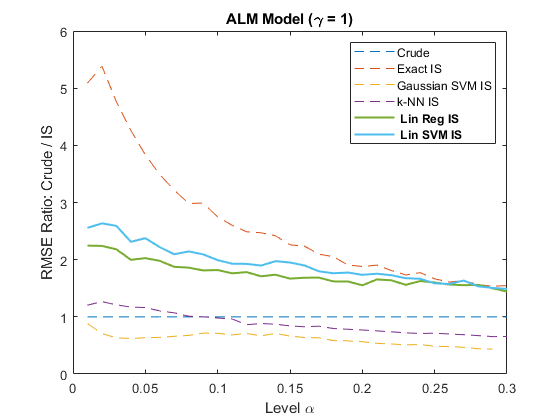}
    \end{minipage}
    \begin{minipage}[c]{0.32\textwidth}
        \centering
        \includegraphics[width=1.11\textwidth]{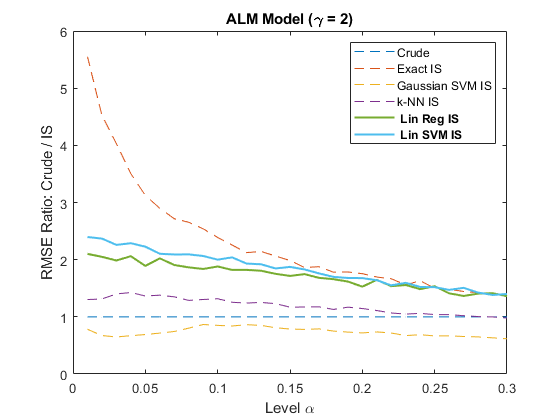}
    \end{minipage}

    \caption{Ratio of the RMSE between the crude method and importance sampling method for 
    the estimation of the considered DRMs for the evolution of the net 
    asset value in the ALM model. The importance sampling method is implemented with 
    the different approximation techniques considered in the paper. }
    \label{fig:ALM_ErrorRatio}
\end{figure}

The results of the simulations are shown in Figure~\ref{fig:ALM_ErrorRatio}, which presents the ratio of the RMSEs of the crude method and the studied importance sampling methods. For further reference, the absolute RMSE of the estimates are provided in 
Figure~\ref{fig:ALM_LevelError} in Appendix~\ref{app:AddPlots}. 
For all DRMs, importance sampling with exact knowledge of the model leads to a significant reduction in the RMSE, especially for small $\alpha$. The variance reduction becomes less pronounced as $\alpha$ increases. For $\gamma = 1$, the exact method achieves the highest observed RMSE ratio of $8.8$ for $\alpha = 0.01$. Importance sampling with linear regression leads to a maximum reduction of $5.8$ for $\alpha = 0.01$.
For the DRMs with concave ($\gamma = 1/2$) and on $[0, \alpha]$ convex ($\gamma = 2$) distortion functions, the linear SVM gives the best reduction in RMSE. For $\gamma = 2$, the best ratio obtained with full knowledge of the model is $5.55$ for $\alpha = 0.01$, and with linear SVM, the maximum reduction is $2.39$ for $\alpha = 0.01$.  For $\gamma = 1 /2$, the exact importance sampling method has the best ratio of $7.96$ for $\alpha = 0.02$ and with the linear SVM approximation of $3.9$ for $\alpha = 0.02$. 
Across all the different estimated DRMs, we see that the importance sampling methods with
Gaussian SVM and $k$-NN regression can lead to a worse RMSE than the crude method.
The worst ratio is observed in all cases with the Gaussian SVM 
with $0.51$ as $\alpha = 0.28$ for $\gamma = 1$, $0.62$ as $\alpha = 0.3$ for $\gamma = 2$ and
and $0.45$ for $\gamma = 1 / 2$ with $\alpha = 0.3$. 

In summary, the proposed method provides a good path to variance reduction. However, the ML approximation in the measure change needs to be chosen carefully, but $k$-fold validation seems to work quite well for this type of analysis. The variance reduction becomes better the more the risk measure depends on extreme tail events. In the ALM case study, the most extreme parameter $\alpha$ was $0.01$. We expect that the iterative procedure outlined in Section~\ref{sec:itexext} would also lead to further improvements in variance reduction when the very extreme tail is considered. 

\onehalfspacing
\printbibliography

\noindent {\bf Declaration of Generative AI and AI-Assisted Technologies in the Writing Process}

\noindent This paper, including its appendix, is original research conducted by the authors. Tools from OpenAI were employed solely for language refinement and clarity enhancement. The substantive content, analyses, and conclusions remain entirely the work of the authors.

\newpage    
\appendix
\doublespacing

\section{Online Appendix}\label{app}

This is an \underline{\bf online appendix} that is provided as an electronic supplement to the paper.

\subsection{Distortion Risk Measures} \label{app:RM}
We review some facts related to risk measures and the special case of DRMs. 
Let $\mathcal{X}$ denote a set of suitable (e.g., all bounded) measurable functions on the measurable space
$(\Omega, \mathcal{F})$. Elements in $\mathcal{X}$ model financial positions or insurance losses. 
We use the sign convention to interpret positive values as losses and negative values as gains. Precise technical conditions for the results below are stated in the references that we mention. 

The axiomatic definition of risk measures goes back to \cite{Artzner1999}; the notion of DRMs was developed by \cite{Wang1996} and \cite{Acerbi2002}. 
DRMs are a subclass of comonotonic risk measures. The link of comontonic risk measures and
DRMs is briefly discussed in Appendix \ref{app:DRM}.
For an excellent overview on risk measures and DRMs we refer 
to \cite{FoellmerSchied2016}. \\
A risk measure $\rho: \mathcal{X} \to \mathbb{R}$ is a functional that quantifies the risk of elements of $\mathcal{X}$:
\begin{defi}
	A mapping $\rho: \mathcal{X} \to \mathbb{R}$ is called monetary risk measure if the following properties hold:
	\begin{itemize}
		\item[(i)] \textit{Monotonicity:} \quad\quad If $X \leq Y$, $X,Y \in \mathcal{X}$, then
		$\rho(X) \leq \rho(Y)$.
		\item[(ii)] \textit{Cash-Invariance:}\quad If $X \in \mathcal{X}$ and $m \in \mathbb{R}$,
		then $\rho(X+m) = \rho(X) + m$.
	\end{itemize}
\end{defi}
Risk measures may exhibit additional properties such as quasi-convexity, which in economic terms means that diversification of positions does not increase the measurements. This property can be shown to be equivalent to convexity:
\begin{defi}
	A risk measure $\rho: \mathcal{X} \mapsto \mathbb{R}$ is called convex, 
    if for $X, Y \in \mathcal{X}$, $\lambda \in [0,1]$ 
	$$ \rho( \lambda X + (1 - \lambda)Y ) \leq \lambda \rho(X) + (1 - \lambda) \rho(Y). $$
\end{defi}
A DRM is defined as follows:
\begin{defi}
	\begin{itemize}
		\item[(i)] A non decreasing function $g: [0,1] \to [0,1]$ with $g(0) = 0$ and $g(1) = 1$ is
		called distortion function.
		\item[(ii)] Let $\mathsf{P}$ be a probability measure on $(\Omega, \mathcal{F})$ and $g$ 
		be a distortion function. The monetary risk measure $\rho_g : \mathcal{X} \to \mathbb{R}$ 
		defined by
		$$ \rho_g(X) = \int_{- \infty}^0 \left[ g \left( \mathsf{P}( \{ X > x \})  \right) - 1\right] dx + 
		\int_0^\infty g\left( \mathsf{P}( \{ X > x \}) \right) dx $$
		is called a DRM with respect to $g$. 
	\end{itemize}
\end{defi}
If the distortion function $g$ is concave we obtain a convex risk measure (see \cite{FoellmerSchied2016}):
\begin{theo} \label{theo:ConcDist}
	Consider the distortion function $g$ and the corresponding DRM $\rho_g$. If $g$ is concave, then the 
	DRM $\rho_g$ is a convex risk measure. If the underlying probability space 
	is atomless, the converse implication is also true. 
\end{theo}
DRMs are expressible as mixture of quantiles. One must focus on the details of the continuity properties of the distortion function to obtain the correct representation, as shown in \cite{Dhaene2012}.
\begin{theo} \label{theo:DRMQuantMix}
\begin{itemize}
	\item[(i)] Let $g$ be a right continuous distortion function. Then the DRM $\rho_g(X)$ is
	given by
	$$ \rho_g(X) = \int_{[0,1]} q_X^+ ( 1 - u) dg(u), $$
	where $q_X^+(u) = \sup \{x \vert F_X(x) \leq u \}$. 
	\item[(ii)] Let $g$ be a left continuous distortion function. Then the DRM $\rho_g(X)$ is given by
	$$ \rho_g(X) = \int_{[0,1]} q_X(1 -u) dg(u) = \int_{[0,1]} q_X(u) d \overline{g}(u), $$
	where $q_X(u) = \inf \{x \vert F(x) \geq u\}$ and $\overline{g}(u) = 1 - g(1 -u)$, $u \in [0,1]$. 
\end{itemize}
\end{theo}
Many important risk measures fall into the class of DRMs; for examples see \cite{ChernyMadan2008},
\cite{FoellmerSchied2016}, \cite{Weber2018}, and the Appendix~\ref{app:ex_DRM}. 
Particularly important examples will be discussed here:

\begin{exam} \label{ex:DRM}
\begin{itemize}
	\item[(i)] Let $g(u) = \mathds{1}_{(\alpha, 1]} (u)$, then $\rho_g$ is  the 
	Value at Risk at level $\alpha$, so that
	$$ \rho_g(X) = V@R_\alpha(X) = \inf \{x \vert F(x) \geq 1 - \alpha \}. $$
 	\item[(ii)] The distortion function $g(u) = \frac{u}{\alpha} \mathds{1}_{[0, \alpha]}(u) + 
 	\mathds{1}_{( \alpha, 1]}(u)$ yields the Average Value at Risk at level $\alpha$, i.e.,
	$$ \rho_g(X) = AV@R_\alpha(X) = \frac{1}{\alpha} \int_0^\alpha V@R_\lambda (X) d\lambda. $$
    \item[(iii)] The distortion function $g(u) = \left( \frac{u}{\alpha} \right)^\gamma
    \mathds{1}_{[0, \alpha]}(u) + \mathds{1}_{(\alpha, 1]}(u)$ with $\alpha \in (0,1]$ 
    and $\gamma \in \mathbb{R}_{>0}$ generalizes the distortion function of the AV@R ($\gamma = 1$) 
    and other special cases such as the hazard transform ($\gamma \geq 1$, $\alpha = 1$) and MAXV@R 
    ($\gamma \in \mathbb{N}, \alpha = 1$). \\
    If $\gamma \leq 1$, the distortion function $g_{\alpha, \gamma}$ is concave such that 
    the corresponding DRM is convex. If $\gamma > 1$ the distortion function is convex on the interval $[0,\alpha]$. 
    In this case, the resulting DRM is not convex.   
\end{itemize}
\end{exam}

\begin{rema} \label{rema:DRMDecomp}
	Every distortion function $g$ can be decomposed in the convex combination
	of a left and right continuous distortion function (see \cite{Dhaene2012}), such that 
	$g(u) = c_1 g_1(u) + c_2 g_2(u)$ with $c_1 + c_2 = 1$ and $c_1, c_2 \geq 0$.
	As a consequence, any distortion risk measure $\rho_g$ with general
	distortion function can be expressed as convex combination $\rho_g(X) = c_1 \rho_{g_1}(X) + c_2 
	\rho_{g_2}(X)$. 
	The decompositions of $g$ and $\rho_g$ is not unique, unless $g$ is a step function. 
	 \cite{BettelsKimWeber2022} point out that a decomposition of $g$ into a left and a right continuous step
	function and a continuous function is unique. 
\end{rema}	

\subsection{Comonotonic Risk Measures} \label{app:DRM}

We review the connections between comonotonic risk measures, Choquet integrals and DRMs. More details can be found in \cite{FoellmerSchied2016}.
 $(\Omega, \mathcal{F})$ is a measurable space on which the financial positions in $\mathcal{X}$
are defined.

\begin{defi}\label{def:comonotonicity} 
    \begin{itemize}
	    \item[(i)] Two measurable functions $X,Y$ on $(\Omega, \mathcal{F})$ are called comonotonic if 
	    \[ (X(\omega) - X(\omega^\prime))(Y(\omega) - Y(\omega^\prime))
	    \geq 0 \hspace{5mm} \forall (\omega, \omega^\prime) 
	    \in \Omega \times \Omega. \]
	    \item[(ii)] A risk measure $\rho: \mathcal{X} \rightarrow \mathbb{R}$ is called comonotonic if 
	    \[ \rho(X +Y) = \rho(X) + \rho(Y), \]
	    for comonotonic $X, Y \in \mathcal{X}$.
    \end{itemize}
\end{defi}

Comonotonic risk measures are expressible as Choquet integrals with respect to capacities.

\begin{defi}
\begin{itemize}
    \item[(i)] A map $c : \mathcal{F} \to  [0, \infty)$ is called monotonic set function 
    if it satisfies the following properties:
	\begin{itemize}
	    \item[a)] $c(\emptyset) = 0$. 
	    \item[b)] $A,B \in \mathcal{F}$, $A \subseteq B$ $\Rightarrow$ $c(A) \leq c(B)$.
	\end{itemize}
    If, in addition, $c(\Omega) = 1$, i.e., $c$ is normalized, then $c$ is called a 
    capacity. 
	\item[(ii)] For $X \in \mathcal{X}$ the Choquet integral of $X$ with respect to the monotone set function 
	$c$ is defined by 
	\[ \int X dc = \int_{- \infty}^0 [ c(\{X > x\}) - c(\Omega)] dx + \int_0^\infty c(\{X > x\}) dx . \]
\end{itemize}
\end{defi}

The Choquet integral coincides with the Lebesgue integral if $c$ is a $\sigma$-additive probability measure. 
The following characterization theorem can, for example, be found in Chapter 4 of \cite{FoellmerSchied2016}.

\begin{theo} \label{theo:DRMComonotonic}
    A monetary risk measure $\rho: \mathcal{X} \to \mathbb{R}$ is comonotonic, if and only if there 
    exists a capacity $c$ on $(\Omega, \mathcal{F})$ such that 
    \[ \rho(X) = \int X dc. \]
\end{theo}

DRMs are an important special case of the comonotonic risk measures. 
In this case, the capacity is defined in terms of a distorted probability measure $\mathsf{P}$. 
The resulting capacity is absolutely continuous with respect to $\mathsf{P}$, but typically not additive. 
	   
\begin{defi}    
    \begin{itemize}
	    \item[(i)] If $\mathsf{P}$ is a probability measure on $(\Omega, \mathcal{F})$ and $g$ is a distortion 
	    function, then
	    \[ c^g(A) := g(\mathsf{P}(A)) , \hspace{5mm} A \in \mathcal{F}, \]
	    is called a distorted probability.
	    \item[(ii)] A comonotonic risk measure
	   $\rho(X) = \int X dc $
	    is called a DRM, if the capacity $c$ can
	    be expressed as a distorted probability. 
	\end{itemize}
\end{defi} 

\begin{table}
 	\small
 	\renewcommand{\arraystretch}{1.4}
	\begin{tabularx}{1.0\textwidth}{|l|c|X|l|}
				\hline 
				Name & Distortion & Closed form &  Reference \\
				\hline 
				\hline
				$\mathrm{MINV@R}$ & $1 - (1 - u)^n$ & \scriptsize $- \mathsf{E}[\min \{ - X_1, \dots, - X_n \}]$ &  \scriptsize \cite{ChernyMadan2008} \\
				&& \scriptsize $= \mathsf{E}[ \max \{ X_1, \dots, X_n\}]$ & \scriptsize\cite{FoellmerSchied2016} \\
				& & & \scriptsize \cite{BannScherer2014} \\
				\hline
				$\mathrm{MAXV@R} $ & $u^{ 1 / n}$ & \scriptsize  $- \mathsf{E}[ Y_1 ]$ &\scriptsize \cite{ChernyMadan2008} \\
				&&\scriptsize \mbox{such that~}  & \scriptsize \cite{FoellmerSchied2016} \\
				& & \scriptsize $\max \{Y_1, \dots, Y_n\} \sim - X$ & \scriptsize \cite{BannScherer2014} \\
				\hline 
				$\mathrm{MINMAXV@R}$ & $1 - ( 1- u^{1 / n})^n$ 
				& \scriptsize $- \mathsf{E}[ \min \{ Y_1, \dots, Y_n\}]
				$  & \scriptsize \cite{ChernyMadan2008} \\
				& & \scriptsize \mbox{such that}& \scriptsize \cite{FoellmerSchied2016} \\
				& & \scriptsize$\max\{Y_1, \dots, Y_n \} \sim - X$& \scriptsize \cite{BannScherer2014} \\
				\hline
				$\mathrm{MAXMINV@R}$ &  $(1 - (1 - u)^n)^{1 / n}$ &\scriptsize $- \mathsf{E} [Y_1]$ & 
				\scriptsize \cite{ChernyMadan2008} \\
				& & \scriptsize\mbox{such that } & \scriptsize \cite{FoellmerSchied2016} \\
				& & \scriptsize $\max\{Y_1, \dots, Y_n\} $ & \scriptsize \cite{BannScherer2014} \\
				& & \scriptsize $\sim \min \{ - X_1, \dots, -X_n\} $ & \\
				\hline
				$RV@R$ & $\frac{u - \beta}{\alpha - \beta} \mathds{1}_{\{\beta < u \leq \alpha \}} 
				+ \mathds{1}_{\{ u > \alpha \}}$ & $ \frac{1}{\alpha - \beta} \int_\beta^{\alpha} V@R_\lambda(X) d\lambda$ & \scriptsize
				\cite{BignozziTsanakas2015} \\
				(Range $V@R$) &$0 < \beta < \alpha < 1$  & &  \scriptsize \cite{Weber2018}, \cite{LiEtAl2018} \\
				\hline
				\footnotesize Proportional   &  $u^{1 / \gamma}$ & $\int_0^\infty (1 - F(x))^{1 / \gamma} dx,$ &
				\scriptsize \cite{Wang1995,Wang1996}  \\
				\footnotesize hazard transform &  $\gamma > 1$ &if $X \geq 0$ a.s. &  \scriptsize \cite{GuillenEtAl2018}\\
				\hline
				Dual power & $1 - (1 - u)^\gamma$ & $\int_0^\infty 1 - F(x)^\gamma dx$, & \scriptsize
				\cite{WirchHardy1999} \\
				transform & $\gamma > 1$ &if $X \geq 0$ a.s. & \scriptsize \cite{GuillenEtAl2018}\\
				\hline 
				Gini's principle & $(1-\vartheta) u + \vartheta u^2$ & $\mathsf{E}[X] + \frac{\vartheta}{2} \mathsf{E}[|X-X_1|]$ & \scriptsize \cite{Yitzhaki1982},\cite{Wozabal2014}\\
				& $0 < \vartheta < 1$ & & \scriptsize \cite{GuillenEtAl2018}\\
				\hline
				Exponential & $\frac{1 - \exp(-ru)}{1-\exp(-r)}$ if $r>0$ &\centering - & \scriptsize \cite{MethniStupfler2017}\\
				transform & $u$ if $r=0$ & & \scriptsize \cite{Dowd2008}\\
				\hline
				Inverse S-shaped & $a \left[ \frac{u^3}{6} - \frac{\delta u^2}{2} + \left(\frac{\delta^2}{2}+\beta \right) u\right]$& & \scriptsize \cite{GueganHassani2014}\\
				polynomial & $a = \left( \frac{1}{6} - \frac{\delta}{2} + \frac{\delta^2}{2} + \beta \right)^{-1}$& \centering - & \scriptsize \cite{MethniStupfler2017}\\
				of degree 3 & $0 < \delta < 1, \beta \in \mathbb{R}$& & \\
				\hline
				Beta family &  $\int_0^u  \frac{t^{a-1}(1-t)^{b-1}}{B(a,b)}  dt$ & \centering - & \scriptsize \cite{SamanthiSepanski2018} \\
				& $a,b >0$ & &\scriptsize \cite{WirchHardy1999}  \\
				\hline
				Wang transform & \scriptsize $\Phi(\Phi^{-1}(u) - \Phi^{-1}(q) )$ & \centering - & \scriptsize \cite{Wang2000, Wang2001}\\
				& $0 < q < 1$ & & \scriptsize \cite{Wozabal2014}\\
				\hline
	\end{tabularx}
	\caption{Further examples of distortion risk measures of a random variable $X$.  Table 1 of the online appendix of \cite{MethniStupfler2017} provides these examples of distortion functions; we include this table of examples as a convenient reference for the reader. In the third column, $X_1, \dots, X_n$ denote independent copies of $X$, $n \in \mathbb{N}$;  $Y_1, \dots, Y_n$ are suitable iid random variables satisfying the conditions given in the third column of the table. $B$ denotes the beta function, $\Phi, \Phi^{-1}$ the distribution and quantile function of the standard normal distribution, respectively.}
	\label{tab:DRMExam}
\end{table}

\subsection{Examples of DRMs}\label{app:ex_DRM}

For further reference, we include a list of examples of distortion risk measure in Table~\ref{tab:DRMExam} which was compiled in the online appendix of \cite{MethniStupfler2017}.

\subsection{Asymptotics of Quantile Estimators in Importance Sampling} \label{app:QuantEstCLT}     

The importance sampling estimator in Section~\ref{sec:SetScene}, eq.~\eqref{form:QuantEstDefi} 
is studied, along with other alternatives, in \cite{Glynn1996}. We can rewrite the estimator in \eqref{form:QuantEstDefi} as 
    $$ \hat{q}_{F^\ast, N}(u) = \inf \left\{ x \in \mathbb{R} \; \bigg\vert \; \frac{1}{N} \sum_{i=1}^N
    \frac{dF}{dF^\ast}(X_i) \mathds{1}_{\{h(X_i) \leq x \}} \geq u \right\}.$$
 Setting $F^\ast = F$, the estimator coincides with the crude Monte Carlo estimator of quantiles, the 
    empirical quantile. 
     We analyze under which conditions the estimator in  eq.~\eqref{form:QuantEstDefi} is finite. For this purpose, we first consider a deterministic problem. Let $\xi_1, \dots, \xi_N \in \bbr$, $\gamma_i \geq 0$ for $i=1,2, \dots, N$ and $z>0$. Then
$q:= \inf \left\{x\in \bbr\,|\, \sum_{\xi_i>x} \gamma_i \leq z\right\} \in \bbr \; \Longleftrightarrow \; \sum \gamma_i > z .$
To see this, we observe that $q = - \infty$ is equivalent to $\sum_{\xi_i>x}  \gamma_i \leq z$ for all $x$, which simply means that $\sum \gamma_i \leq z$. This proves the claim, since $q=\infty$ is equivalent to $\sum_{\xi_i>x} \gamma_i >z $ for all $x$, but for large enough $x$ the sum is empty and equal to $0$, contradicting $z>0$.

The simple characterization implies for $u \in (0,1)$ that
\begin{equation}\label{prop:finiteQuantEst}
\hat{q}_{F^\ast, N}(u) \in \bbr \quad \Longleftrightarrow \quad \frac{1}{N} \sum_{i=1}^N \frac{dF}{dF^\ast} (X_i) > 1- u 
\end{equation}
Assuming the samples $(X_1, h(X_1)), \dots, (X_N, h(X_N))$ from $F^\ast$ are independent and identically distributed, we obtain by a law of large numbers that
$\frac{1}{N} \sum_{i=1}^N \frac{dF}{dF^\ast} (X_i)  \longrightarrow  \mathsf{E}_{F^\ast} \left[   \frac{dF}{dF^\ast} (X)        \right]  = 1$, thus eq. \eqref{prop:finiteQuantEst} is satisfied for $N$ large enough.

The asymptotic normality of the estimator $\hat{q}_{F^\ast, N}(u)$ can be shown, if the following assumptions hold. This is stated in Theorem~\ref{theo:CLT} of \cite{AhnShy2011}, generalizing \cite{Glynn1996}. 

\begin{ass} \label{ass:ISDist}
    Let $G$, $G^\ast$ be the distribution functions of $h(X)$, if $X$ is distributed according to  $F$, $F^\ast$, respectively. \\
    Assume that for $u \in (0,1)$ the following properties hold: 
    \begin{itemize}
        \item[(A1)] $G$ is absolutely continuous with respect to $G^\ast$. 
        \item[(A2)] $G^\ast$ is continuous at $q_Y(u)$.
        \item[(A3)] $G$ has a strictly positive first derivative at $q_Y(u)$.
        \item[(A4)] $\frac{dG}{dG^\ast}(\cdot)$ is a function of finite variation on compacts and 
        has finite negative variation on $(y, \infty)$ for all $y \in \mathbb{R}$.
        \item[(A5)] $\frac{dG}{dG^\ast}(\cdot)$ is right continuous.
        \item[(A6)] There exists a $\lambda \in (0, 1/2]$ such that 
        $$ \int_y^\infty (1 - G^\ast (x-))^{1 / 2 - \lambda} d \left\vert \frac{dG}{dG^\ast}(x) \right\vert < 
        \infty \hspace{10mm} \forall y \in \mathbb{R}. $$
    \end{itemize}
\end{ass}

\begin{rema}\label{rema:ass}
	These assumptions of \cite{AhnShy2011} are weaker than the assumptions
	of \cite{Glynn1996} to obtain the implications in Theorem
	\ref{theo:CLT}. \cite{Glynn1996} assumes that (A1) to (A3) hold and $\mathsf{E}_{G^\ast}\left[ \frac{dG}{dG^\ast}(X)^3\right] < \infty$. 
	The latter is replaced by assumption (A6), together with the technical conditions (A4) and (A5); here, 
    $\vert \cdot \vert$ denotes total variation.\footnote{If $f : \mathbb{R} \mapsto \mathbb{R}$ is a function of bounded variation, there
	exist two increasing functions $f^+, f^- : \mathbb{R} \mapsto \mathbb{R}$ with $f = f^+ - f^-$,
	and $\vert f \vert := f^+ + f^-$. 
	The former is called the Jordan decomposition of $f$ and is closely related to the Hahn decomposition of
	signed measures. For details we refer to \cite{Folland1999}.} \cite{AhnShy2011} show that 
	if $\mathsf{E}_{G^\ast} \left[  \frac{dG}{dG^\ast}^{2 + \delta}\right] < \infty$ holds for some $\delta > 0$, then 
	(A6) is satisfied for $\lambda \in ( 0, \delta / ( 4 + 2 \delta))$.
\end{rema}

\begin{prop} \label{prop:ineq}
    If Assumption \ref{ass:ISDist} holds, we obtain for $u \in (0,1)$: 
    \begin{itemize}
        \item[(i)]  $\mathsf{E}_{F^\ast} \left[ \frac{dF}{dF^\ast}(X)^2 \mathds{1}_{\{h(X) \in (q_Y(u), \infty)\}} 
        \right] \geq (1 - u)^2$ 
        \item[(ii)]  Suppose that $F$ and $F^\ast$ are equivalent. Then $$\mathsf{E}_{F^\ast} \left[ \frac{dF}{dF^\ast}(X)^2 \mathds{1}_{\{h(X) \in (q_Y(u) , \infty) \}} \right] 
        \geq (1 - G^\ast(q_Y(u))^{-1}$$
    \end{itemize}
\end{prop}
\begin{proof}
    \begin{itemize}
        \item[(i)] By assumption $G$ is continuous at $q_Y( u)$, implying $G(q_Y(u)) = u$. The inequality is thus a simple consequence of Jensen's inequality.

        \item[(ii)] The function $f(x) = 1 / x$ is convex 
        function on $\mathbb{R}_{>0}$. Hence, by Jensen's inequality we have
       $$
			\mathsf{E}_{F^\ast} \left[ \frac{dF}{dF^\ast}(X)^2 \mathds{1}_{\{ h(X) > q_Y(u) \}} \right] 
			= \mathsf{E} \left[ \left(\frac{dF^\ast}{dF}(X) \right)^{-1} \mathds{1}_{\{ h(X) > q_Y(u) \}} \right] 
			\geq \left( 1 - G^\ast( q_Y(u)) \right)^{-1}. 
		$$
    \end{itemize}
\end{proof}

We now consider the estimation of the quantile $q_Y(1 - \alpha)$, $\alpha \in (0,1)$, by 
$\hat{q}_{F_\vartheta, N}(1 - \alpha)$, the estimator defined in Section~\ref{sec:SetScene}, eq.~\eqref{form:QuantEstDefi}. 
According to Theorem~\ref{theo:CLT} we should choose $F_\vartheta$ such that 
\begin{align}
    \mathsf{E}_{F_\vartheta} \left[ \frac{dF}{dF_\vartheta} (X)^2 \mathds{1}_{\{h(X) > q_Y(1 - \alpha)\}} \right]
    \label{form:SecMom}
\end{align}   
is small. We consider exponential twists given in eq.~\eqref{form:SampleDist}. The following standard result for the cumulant generating function is useful:
\begin{lem} \label{lem:CumuDeri}
    Letting $F: \mathbb{R}^d \rightarrow [0,1]$ be the distribution function of $X$,
    $h: \mathbb{R}^d \rightarrow \mathbb{R}$ 
    a measurable function. If $\psi(\vartheta+t) = \log( \mathsf{E} [ \exp((\vartheta+t) h(X) ) ]) < \infty$ for all $t$ in some neighborhood of $0$, then 
 $ \psi^\prime(\vartheta) = \mathsf{E}_{F_\vartheta}[h(X)], $
    where $F_\vartheta$ is the family of distributions defined in (\ref{form:SampleDist}). 
\end{lem}

To find an appropriate parameter $\vartheta$ to make (\ref{form:SecMom}) small, we take an approach like in \cite{SunHong2009}. By the definition of $F_\vartheta$ 
samples in the right tail are more likely to occur when $\vartheta > 0$, which also indicates that to minimize (\ref{form:SecMom}) we should choose $\vartheta > 0$. We observe that{\small
\begin{align*}
 &  \mathsf{E}_{F_\vartheta} \left[ \frac{dF}{dF_\vartheta} (X)^2 \mathds{1}_{\{h(X) > q_Y(1 - \alpha)\}} 
    \right] = \mathsf{E} \left[ \frac{dF}{dF_\vartheta} (X) \mathds{1}_{\{h(X) > q_Y(1 - \alpha) \}}
    \right] 
    = \mathsf{E} \left[ \exp ( \psi(\vartheta) - \vartheta h(X)) \mathds{1}_{\{h(X) > q_Y(1 - \alpha)\}} 
    \right] \\ 
  & \leq \exp( \psi(\vartheta) - \vartheta q_Y(1 - \alpha)) \cdot  \mathsf{P}( h(X) > q_Y(1 - \alpha)). 
\end{align*}}
Minimizing the upper bound is then equivalent to minimizing $ \psi(\vartheta) - \vartheta q_Y(1 - \alpha)$, from which we obtain a first order condition using Lemma~\ref{lem:CumuDeri} 
\begin{align}
    q_Y(1 - \alpha) = \mathsf{E}_{F_\vartheta} [h(X)]. 
\end{align} 
In their paper, \cite{SunHong2009} show that this approach yields a strict reduction of the objective function \eqref{form:SecMom}:
\begin{theo}
	Consider the situation as described above. Assume there 
	exists $\varepsilon > 0$ such that $G_Y$ is differentiable with strictly positive derivative on 
	$(q_Y(1 - \alpha) - \varepsilon , q_Y(1-\alpha) + \varepsilon)$. 
	Further suppose that $q_Y(1 - \alpha) > \mathsf{E}[h(X)]$, $\frac{dF}{dF_{\vartheta^\ast}}(x) = 
	\exp (  \psi(\vartheta^\ast) - \vartheta^\ast h(x))$, and $\vartheta^\ast$ be chosen such that 
	$q_Y(1 - \alpha) = \mathsf{E}_{F_\vartheta^\ast} [h(X)]$. Then 
	$ \mathsf{E}_{F_{\vartheta^\ast}} \left[  \frac{dF}{dF_{\vartheta^\ast}}(X)^2
	\mathds{1}_{\{ h(X) > q_Y(1 - \alpha)\}} \right] < 
	\mathsf{E} \left[  \mathds{1}_{\{h(X) > q_Y(1 - \alpha ) \}} \right]. $
\end{theo}

\subsection{Refinements, Alternatives, and Clarifications of Algorithm 1} \label{app:Algo1}

Here, we aim to discuss various details, potential refinements, and alternatives applicable to Algorithm \ref{algo:ISDRMEsti}.

\begin{rema} \phantom{Mach neue Zeile!}
\begin{itemize}
    \item\underline{Line 12:} The quantity \( G'(x) \), required for defining \( c_i \), is assumed unknown. In the algorithm, we estimate it using a kernel density approximation based on pivot samples, which is then evaluated to obtain \( G'(\cdot) \). Other estimation methods, such as finite differences based on the empirical distribution function, could also yield a viable algorithm.
    \item \underline{Line 16:} The normalizing factor \(\hat{\psi}(\vartheta_i)\) is estimated as the sample mean over the pivot samples. Note that due to estimation errors, the resulting \(dF_i\) will not be perfectly normalized. Implementations of the MH-Algorithm are generally more efficient when using an approximated density rather than an unnormalized function. The normalization process in line 24 is discussed in Section \ref{subsec:ImpIssues}.
    \item \underline{Line 17:} Instead of randomly drawing \(\theta_i\) to select the mixture component, we could stratify the samples according to \(p_i\) by drawing \(p_i N\) samples from each \(F_i\). In our algorithm, we opted for the proposed sampling strategy.
    \item \underline{Line 18:} The quantile estimation can also be based on the evaluated approximation $\hat{h}$. While this reduces estimation accuracy, it helps avoid further evaluations of $h(X)$. In this paper, we opt not to use $\hat{h}$ in the third step, allowing us to compare the performance of the IS estimation against a crude method.
    \item \underline{Line 21:} The variance estimation, used to compare variances, relies on the sample mean of the two second moments from Section \ref{subsec:EffUse}. Conducting this comparison at the algorithm's end allows the use of both the $N$ samples from the importance sampling and the $M$ pivot samples for improved accuracy. In our algorithm, the estimation of $q_Y(1 - \alpha_i)$ reuses the estimate from line 7. Alternatively, one could employ quantile estimates based on samples from the importance sampling. These approaches all yield viable algorithms, with the details reflecting design choices in the estimation process
    to balance computation effort and accuracy of the method.
\end{itemize}
\end{rema}

\subsection{Tools from Machine Learning} \label{app:MLTools}

For convenience, we briefly review the considered ML regression techniques and the methodology of $k$-fold validation. 
An excellent introduction to ML is provided by \cite{ShalevSchwartzBenDavid2014}). In our simulation algorithm, pivot 
samples $$S =  (X_1, h(X_1)), \dots, (X_M, h(X_M)) = \left(X_i, h(X_i)\right)_{i=1, \dots, M}$$ are used as training data. 

\subsubsection{Linear Predictors}

We briefly review linear prediction; this is based on Section 9.2 of
\cite{ShalevSchwartzBenDavid2014}. For regressions, we consider the hypothesis class 
\begin{equation}\label{eq:linhyp}
    \mathcal{H}_{lin} = \left\{ x \mapsto \langle w, x\rangle + b \; \vert \;  w \in \mathbb{R}^d,\, b \in \mathbb{R}\right\} 
\end{equation} 
where $\langle \cdot, \cdot \rangle$ is the Euclidean inner product. One approach is to determine $w^\ast,b^\ast$ 
by empirical risk minimization (ERM) with quadratic loss function. 
The empirical risk is thus given by
$L_S(h) = \sum_{i=1}^M ( \langle w, X_i \rangle+ b- h(X_i))^2 $, where $h \in \mathcal{H}_{lin}$
is the predictor corresponding to $w, b$. Optimal $w^\ast, b^\ast$ are determined by
$ w^\ast, b^\ast = \argmin_{w,b} \sum_{i=1}^M( \langle w, X_i \rangle+ b- h(X_i))^2. $ As is well known, the first order conditions leads to linear problem. \cite{ShalevSchwartzBenDavid2014} discuss the application of linear programming and perceptrons (cf. \cite{Rosenblatt1958}).

\subsubsection{Polynomial Predictors}

Again, we refer for more details, \cite{ShalevSchwartzBenDavid2014}, Section 9.2.2. To illustrate the main idea, we assume that the dimension of the training patterns is $d = 1$. The hypothesis class of polynomial predictors with degree $k$ is given as
$$ \mathcal{H}_{poly}^k = \left\{x \mapsto p(x) \vert w \in \mathbb{R}^{k+1} \right\} $$ 
where 
$ p(x) = w_0 + w_1 x + w_2 x^2 + \dots + w_k x^k. $
Obviously, polynomial predictors can be seen as the application of linear hypotheses to features which are obtained as a transformations of the original input patterns, in this case leading to monomials as features. Namely, setting $\psi(x) = (1, x, x^2, \dots, x^k)$, we have $p(x) = \langle w, \psi(x) \rangle$. We 
can thus apply the same methods on the transformed sample $S^\prime = \left(\psi(X_i), h(X_i) \right)_{i=1, \dots, M}$ as in the case of linear 
predictors with ERM specified by
$ w^\ast = \argmin_{w} \sum_{i=1}^M ( \langle w_i , \psi(X_i) \rangle + h(X_i))^2 . $

\subsubsection{Support Vector Machines}

Support vector machines can be used for classification and regression purposes. A good overview on classification is 
Chapter 15 of \cite{ShalevSchwartzBenDavid2014}. An early extension to regression tasks is \cite{DruckerEtAl1996}. For more details see and Chapter 6 in \cite{Vapnik1999}, Chapter 2 in \cite{HuangKecManKopriva2006}, or the tutorial article \cite{SSch2004} which form the basis for our brief review.

\paragraph{Linear Support Vector Machine Regression}

First, we consider again the linear predictor hypothesis class \eqref{eq:linhyp}.
A support vector machine regression considers the optimization problem 
\begin{align*}
    & (w^\ast, b^\ast) = \argmin_{w,b} \frac{1}{2} \Vert w \Vert^2 \\
    \text{subject to } & \vert h(X_i) - \langle w, X_i \rangle - b \vert \leq \varepsilon,
\end{align*}
where $\varepsilon > 0$ is a parameter controlling the tolerated distance of the samples 
to the predictor; within the tolerance bound, the flatness of the solution is minimized. As the solution to the optimization problem above may not exist, the soft margin concept introduces the slack variables $\xi, \Bar{\xi} \in \mathbb{R}^M$ and considers instead
\begin{align*}
    &(w^\ast, b^\ast, \xi^\ast, \Bar{\xi}^\ast )\; = \; \argmin_{w, b, \xi, \Bar{\xi}} \; \left\{ \frac{1}{2} \Vert w \Vert^2 \, + \, C \cdot \sum_{i=1}^M 
    (\xi_i + \Bar{\xi}_i)\right\}, \\
    \text{subject to } & h(X_i) - (\langle w, X_i \rangle - b) \leq \varepsilon + \xi_i, \\
    &  (\langle w, X_i \rangle + b ) - h(X_i) \leq \varepsilon + \Bar{\xi}_i, \\
    & \xi_i ,\Bar{\xi}_i \geq 0.
\end{align*}
To approximate the solution of the soft margin optimization we use in this paper a sequential optimization described in \cite{Platt1998} and \cite{FanPaiChihJoachims2005}.

\paragraph{The Kernel Trick}

The summary in this section is based on
chapter 16 of \cite{ShalevSchwartzBenDavid2014}, Section 6.3 of
\cite{Vapnik1999}, Section 2.2 of \cite{HuangKecManKopriva2006}, and \cite{SSch2004}.
When generalizing support vector machines to nonlinear predictors, the same approach as outlined for polynomial predictors can be taken. Instead of considering linear hypotheses on the input space, one considers instead the concatination of a unknown linear function and a known mapping from the input space to a feature space. Machine learning then determines a suitable linear predictor on the feature space. Good feature spaces can be very high-dimensional and the algorithm might become infeasible. 

In the case of support vector machines, a computationally cheaper way is available which relies on the following observation. If linear predictors are learnt on an Euclidian space using support vector optimization, the solution can be determined if scalar products of all elements of the domain of the linear predictors can be computed. Consider, for example, the input space $\mathbb{R}$ and the transformation to features $\psi : \mathbb{R} \mapsto \mathbb{R}^m$. Replacing the original training samples $S = \left(X_i, h(X_i) \right)_{i=1\dots, M}$ by $\hat{S} = \left( \psi(X_i), h(X_i)\right)_{i=1\dots, M}$,  we seek a support vector linear predictor
computed from $\hat{S}$. Since the solution can be computed from the knowledge of scalar products of features  $\langle \psi(x), \psi(x^\prime)\rangle =  K(x, x^\prime) $ which are labeled by inputs, it suffices to specify the corresponding kernel $K: \mathbb{R} \times \mathbb{R} \rightarrow \mathbb{R}$, but explicit knowledge of $\psi$ is not required. In this article, 
we again use the sequential optimization described, e.g., in \cite{FanPaiChihJoachims2005} with two commonly used kernel functions:

\begin{exam}
\begin{itemize}
    \item[(i)] \textit{Polynomial kernel:}
    The kernel $K(x, x^\prime) = (1 + \langle x, x^\prime \rangle )^k$ corresponds to 
    $$ K(x, x^\prime) = \langle \psi(x), \psi(x^\prime) \rangle = 
    \sum_{J \in \{0,1\}^k} \prod_{i=1}^k x_{J_i}\prod_{i=1}^k x^\prime_{J_i}, $$
    where we define $x_0 = x_0^\prime = 1$. Then $\psi(x)$ has as components monomials up to 
    degree $k$, and the SVM will learn a polynomial predictor. 
    
    \item[(ii)] \textit{Gaussian kernels:} 
    The kernel
    $$ K(x, x^\prime ) = \exp \left( - \frac{\Vert x- x^\prime \Vert^2  }{2 \sigma}  \right), $$
    for $\sigma > 0$, is called Gaussian kernel. The Gaussian kernel corresponds to the embedding 
    $\psi(x)$ with the components 
    $$ \psi(x)_i = \frac{1}{\sqrt{i!}} \exp \left( \frac{x^2}{2} \right) x^i. $$
\end{itemize}
\end{exam}

\subsubsection{$k$-fold Cross Validation}

Based on Section 11.2 of \cite{ShalevSchwartzBenDavid2014}, we briefly describe k-fold cross validation. The training of different methods was already discussed, now we need a strategy to select among these methods. For this purpose, the training sample $S$ is partitioned into subsets $S_1, \dots, S_k$, each of size $M/k$ (where $k$ divides $M$ which can easily be realized in the implementation), such that $S_j := \left(X_i, h(X_i) \right)_{i = j \cdot \frac{M}{k} + 1, \dots, (j+1) \cdot \frac{M}{k}}$. Assume that $r \in \{1, \dots, R\}$ enumerates the different methods considered and/or parameters of these methods, and let $A_r(S)$ be the output of the algorithm trained on the training data $S$ resulting in the predictor $h_r$. For each $r$ the algorithm can alternatively be trained on the training data $S \setminus S_j$, 
$j \in \{1, \dots, k\}$ with output hypothesis $h_{r,j}$. The individual predictors $h_{r,j}$ are validated on the remaining fold of 
training data, i.e., 
$$ \text{error}(r) = \frac{1}{k} \sum_{i=1}^k L_{S_i} ( h_{r,i}) = \frac{1}{k} \sum_{i=1}^k \sum_{(x,h(x)) \in S_i}
l(h(x), h_{r,i}(x)), $$
where $l(\cdot, \cdot)$ is the considered loss function. In our implementation, we use $l(x, y) = (x-y)^2$ for the purpose of error measurement, although this loss function is not used in SVMs or $k$-NN. From the estimated errors of the predictors $h_r$ we can then choose the one which is performing best. The hypothesis classes from Sections \ref{sec:CaseStudies} \& 
\ref{sec:ALM} are displayed in Table \ref{tab:HypothesisClasses}.

\begin{table}[]
    \centering
    \begin{tabular}{|l|c | c|}
        \hline 
        Hypothesis Class & Hyperparameter &  Stop Criterion  \\
        \hline 
        \hline
        Linear Predictors & - & - \\
        \hline
        Polynomial Predictors & ordered increasing in & Overfitting observed 
        \\
         of degree $q_1$ & $q_1 \in \{2,3, \dots \}$ & \\
        \hline
        Linear SVM & - & - \\
        \hline
        Polynomial SVM  &  ordered increasing in & Overfitting observed 
        \\
       of degree $q_2$ & $q_2 \in \{2,3,\dots \}$  & Fitting computational unfeasible \\
        \hline
        Gaussian SVM & - & - \\
        \hline
        $k$-NN Regression & ordered increasing in & Overfitting observed \\
        &$k \in \{1,2,3, \dots\}$ & \\
        \hline
    \end{tabular}
    \caption{Overview of the hypothesis classes and order of the hypothesis classes
    considered in the $k$-fold validation. The stop criterion determines the largest 
    hyperparameter considered for the hypothesis classes.}
    \label{tab:HypothesisClasses}
\end{table}

\subsection{Proofs and Calculations} \label{app:ProofAnd Calc}

\subsubsection{Appendix to Section \ref{sec:saq}} \label{app:OptAllo}

\begin{proof}[Auxiliary computations] 
Suppose that Assumption \ref{ass:ISDist} holds. For large enough $N_i$ we use the approximation from Theorem \ref{theo:CLT} for all $i$, i.e., 
$$ \hat{q}_{F_i, N_i}(1 - \alpha_i) \quad \sim \quad  \mathcal{N} \left( q_Y(1 - \alpha_i), \frac{\mathsf{E}_{F_i} \left[ \frac{dF}{dF_i}
(X)^2 \mathds{1}_{\{h(X) > q_Y(1 - \alpha_i) \}} \right] - \alpha_i^2}{ N G^\prime(q_Y(1- \alpha_i))^2} \right) .$$
For $u \in [\alpha_i, \alpha_{i+1})$ we have {\small
\begin{align*}
	&\mathsf{E}\left[  (q_Y(1 -u) - \hat{q}_Y(1 -u))^2 \right] 
	\; = \;
	\mathsf{E} \left[ (q_Y(1 -u) - \hat{q}_{F_i, N_i} ( 1 - \alpha_i) )^2 \right] \\
	\approx\quad & \mathsf{E} \left[ \left(q_Y(1 -u) -  q_Y( 1- \alpha_i) - 
	\sqrt{\frac{\mathsf{E}_{F_i} 
	\left[ \frac{dF}{dF_i} (X)^2  \mathds{1}_{\{h(X) > q_Y(1 - \alpha_i)\}} \right] - \alpha_i^2}{N_i G^\prime( q_Y(1 - 
	\alpha_i))^2}} Z_i \right)^2 \right] \\
	=\quad& (q_Y(1 -u) - q_Y(1 -\alpha_i))^2 - 2 ( q_Y(1 -u) - q_Y(1 - \alpha_i)) 
	\sqrt{ \frac{\mathsf{E}_{F_i} 
	\left[ \frac{dF}{dF_i} (X)^2 \mathds{1}_{\{ h(X) > q_Y(1 - \alpha_i)\}} \right] - \alpha_i^2}{N_i G^\prime( q_Y(1 - 
	\alpha_i))^2}} \mathsf{E}[Z_i]  \\
	+\quad& \frac{\mathsf{E}_{F_i} 
	\left[ \frac{dF}{dF_i} (X)^2  \mathds{1}_{\{ h(X) > q_Y(1 - \alpha_i)\}} \right] - \alpha_i^2}{N_i G^\prime( q_Y(1 - 
	\alpha_i))^2} \mathsf{E}[Z_i^2] \\
	= \quad &  ( q_Y(1 -u) - q_Y(1 - \alpha_i))^2 + \frac{\mathsf{E}_{F_i} 
	\left[ \frac{dF}{dF_i} (X)^2 \mathds{1}_{\{ h(X) > q_Y(1 - \alpha_i)\}} \right] - \alpha_i^2}{N_i G^\prime( q_Y(1 - 
	\alpha_i))^2}
\end{align*}}
where $Z_i$, $i \in \{0,1, \dots, m \}$ are i.i.d. standard normals. With this we obtain{\small
\begin{align*}
	& \int_0^1 \mathsf{E} [(q_Y(1 -u) - \hat{q}_Y(1 -u))^2 ] dg(u) \; = \;  \sum_{i=0}^m \int_{\alpha_i}^{\alpha_{i+1}} 
	\mathsf{E} [(q_Y(1 -u) - \hat{q}_{F_i, N_i}(1 - \alpha_i))^2] dg(u)   \\
	\approx\quad &  \sum_{i=0}^m \int_{\alpha_i}^{\alpha_{i+1}} ( q_Y(1 -u) - q_Y(1 - \alpha_i))^2 dg(u) +
	\frac{\mathsf{E}_{F_i} \left[ \frac{dF}{dF_i}(X)^2 \mathds{1}_{\{h(X) > q_Y(1 - \alpha_i)\}} 
	\right] - \alpha_i^2}{N_i G^\prime(q_Y(1 - \alpha_i))^2}
    (g(\alpha_{i+1}) - g(\alpha_i) ) .
\end{align*} }
\end{proof}

\begin{proof}[Proof of Equation \eqref{theo:OptAllo}]
	Let 
	$$ c_i := \frac{\mathsf{E}_{F_i} \left[ \frac{dF}{dF_i} (X)^2 \mathds{1}_{\{h(X) > q_Y(1 - \alpha_i)\}} \right]
	- \alpha_i^2}{G^\prime(q_Y(1 - \alpha_i))}(g(\alpha_{i+1}) - g(\alpha_i)) . $$
	The optimization problem becomes to minimize $\sum_{i=0}^m \frac{c_i}{N_i}$ under the constraint $\left( \sum_{i=0}^m N_i\right) - N  = 0$
	with  $c_i \geq 0$, since $g$ is increasing and according to Proposition \ref{prop:ineq}. 
	The Lagrangian for this optimization problem is
	$ \mathcal{L}(N_0, N_1, \dots, N_m; \lambda) = \sum_{i=0}^m \frac{c_i}{N_i} + \lambda 
	\left( \sum_{i=0}^m N_i - N \right) $
	with gradient
	$$ \nabla \mathcal{L}(N_0, N_1, \dots, N_m; \lambda)  = \begin{pmatrix}
	 - \frac{c_0}{N_0^2} + \lambda &
	 - \frac{c_1}{N_1^2} + \lambda &
	 \dots &
	 - \frac{c_m}{N_m^2} + \lambda  &
	 \sum_{i=0}^m N_i - N 
	\end{pmatrix}^T \overset{!}{=} 0. $$
	We rewrite the first $m+1$ equations as
	$ \sqrt{\frac{c_i}{\lambda}} = N_i$, $i = 0,1, \dots, m$,
	and plug this into the last equation to obtain
	$  \sum_{i=0}^m \sqrt{\frac{c_i}{\lambda}} = N$ which is equivalent to $\lambda = \left( \frac{1}{N} \sum_{i=0}^m \sqrt{c_i} \right)^2. $
	This yields  the critical point 
	$ N_i = N \frac{\sqrt{c_i}}{\sum_{i=0}^m \sqrt{c_i}}$, $i = 0,1, \dots, m$.
	To show that this is the minimum under the constraint it suffices to verify that
	$\mathcal{L}(N_0, N_1, \dots, N_m; \lambda)$ is a convex function in $(N_0, N_1, \dots, N_m)$. 
	Rewriting 
	$ \mathcal{L}(N_0, N_1, \dots, N_M ; \lambda) = 
	\sum_{i=0}^m \frac{c_i}{N_i} + \lambda N_i - \frac{N}{m+1}, $
	we observe that the functions
	$ \mathcal{L}_i( N^\prime) = \frac{c_i}{N^\prime} + \lambda N^\prime - \frac{N}{m+1} $
	are each the sum of two convex functions and therefore itself convex if $N^\prime \in \mathbb{R}_+$. 
	It follows that $\mathcal{L}(N_0, N_1, \dots, N_M ; \lambda)$
	is a convex function, implying that the critical point is a minimum. 
\end{proof}

\subsection{Conditional Sampling}

In view of the error bound in (\ref{eq:ErrorBound}), note that the second moments are affected only by the change of measure in the tail. This raises the question whether the estimation error can be reduced by conditioning the sampling distributions on the tail.
To address this question, we define the conditional sampling distribution as follows:
$$ dF_{\vartheta,v}(x) = \exp(\vartheta h(x) - \psi(\vartheta, v)) \mathds{1}_{\{h(x) \geq q_Y(v)\}}
dF(x),$$
where $\psi(\vartheta, v) = \log \left( \mathsf{E}_F \left[ \exp(h(X)) \mathds{1}_{\{h(X) 
\geq q_Y(v)\}} \right] \right)$ and $v \in [0,1)$. \\
To define a conditional mixture sampling distribution, let $\alpha_{m'}$ be the largest element of the partition such that 
$g(\alpha_{m'+1}) - g(\alpha_{m'}) > 0$. Assume further that $\sum p_i = 1$, with $p_i = 0$ for $i < m'$ and $p_i > 0$ for $i \ge m'$. 
For $v \le 1 - \alpha_{m'}$, the conditional mixture sampling distribution is given by
\[
dF^\ast_{v}(x) = \sum_{i=m'}^m p_i \, dF_{\vartheta_i, v}(x).
\]
For both individual and mixture IS estimators, conditioning on the tail reduces the estimation error, as shown in the following proposition.
\begin{prop}
    \begin{itemize}
        \item[(i)] Let $v_i \leq w_i \leq 1 - \alpha_i$, for all $i \in \{0,1,\dots, m \}$. Then: 
        \begin{align*}
           & \mathcal{E}(F_{\vartheta_0, v_0}, F_{\vartheta_1, v_1}, \dots, F_{\vartheta_m, v_m}, \tilde{N}_{Ind}) 
            - \mathcal{E}(F_{\vartheta_0, w_0}, F_{\vartheta_1, w_1}, \dots, F_{\vartheta_m, w_m}, \tilde{N}_{Ind})  \\
            =& \sum_{i=0}^m \frac{g(\alpha_{i+1}) - g(\alpha_i)}{N_i G^\prime(q_Y(1 - \alpha_i))^2} 
            \mathsf{E}_F\left[ \exp(\vartheta_i h(X)) \mathds{1}_{\{h(X) \in (q_Y(v_i), q_Y(w_i)] \} } \right]
            \\
            \cdot & \mathsf{E}_F \left[ \exp ( - \vartheta_i h(X) \mathds{1}_{\{h(X) > q_Y(1 - \alpha_i) \}} \right] 
            \geq 0,
        \end{align*}
        where $\tilde{N}_{Ind} = (N_0, N_1, \dots, N_m)^T$ and $\sum N_i = N$. 
        \item[(ii)] Let $v \leq w \leq 1 - \alpha_{m^\prime}$, then
        \begin{align*}
             &\mathcal{E}(F^\ast_{v}, F^\ast_{v}, \dots, F^\ast_{v}, \tilde{N}_{Mix}) - 
             \mathcal{E}(F^\ast_{w}, F^\ast_{w}, \dots, F^\ast_{w}, \tilde{N}_{Mix}) \\
             =& \sum_{i=m^\prime}^{m} \frac{g(\alpha_{i+1}) - g(\alpha_i)}{
             N  G^\prime(q_Y(1 - \alpha_i))^2 } \\
             & \mathsf{E}_F \left[  \frac{1}{\sum_{i=m^\prime}^m p_i \exp(
             \vartheta_i h(X) - \psi(\vartheta_i, v))} 
             - \frac{1}{\sum_{i=m^\prime}^m p_i \exp( \vartheta_i h(X) - \psi(\vartheta_i,  w)) } \right] \\
             &\geq 0,
        \end{align*}
         where $\tilde{N}_{Mix} = (N, \dots, N)^T$.
    \end{itemize}
\end{prop}
\begin{proof}
    \begin{itemize}
        \item[(i)]For $v_i \leq w_i \leq 1-\alpha_i$, $i=0,\dots,m$, we compute
        {\small
        \begin{align*}
            & \mathcal{E}(F_{\vartheta_0, v_0}, F_{\vartheta_1, v_1}, \dots, F_{\vartheta_m, 
            v_m}, \tilde{N}) 
            - \mathcal{E}(F_{\vartheta_0, w_0}, F_{\vartheta_1, w_1}, \dots, F_{\vartheta_m, 
            w_m}, \tilde{N})  
            = \sum_{i=0}^m \frac{g(\alpha_{i+1}) - g(\alpha_i)}{N_i G^\prime(q_Y(1 - 
            \alpha_i))^2} \\
            &\cdot \left( \underbrace{\mathsf{E}_F \left[ \frac{dF}{dF_{\vartheta_i, v_i}}
            (X) \mathds{1}_{\{h(X) > q_Y(1 - \alpha_i)\}}\right] -  
            \mathsf{E}_F \left[ \frac{dF}{dF_{\vartheta_i, w_i}}
            (X) \mathds{1}_{\{h(X) > q_Y(1 - \alpha_i)\}}\right]}_{(1)}
            \right).
        \end{align*}
        }
        By Lemma \ref{lem:CondSampIneq}, $(1)\geq 0$, and since the prefactor is nonnegative, the inequality follows. Moreover, with
        $$ \psi(\vartheta_i, v) = \log ( \mathsf{E}_F [\exp(h(X)) \mathds{1}_{\{h(X) \geq q_Y(v) \}} ] ),$$
        we can rewrite
        \begin{align*}
            (1) &= \mathsf{E}_F \left[ \exp\left( - \vartheta_i h(X) \right)  
            \left( \exp ( \psi(\vartheta_i, v_i)) - \exp( \psi( \vartheta_i, w_i)) \right) 
            \mathds{1}_{\{h(X) > q_Y(1 - \alpha_i) \}} \right]\\
            &= \left( \exp ( \psi(\vartheta_i, v_i)) - \exp( \psi( \vartheta_i, w_i)) \right)  \mathsf{E}_F \left[ \exp( - \vartheta_i h(X) )
            \mathds{1}_{\{h(X) > q_Y(1 - \alpha_i) \}} \right] \\
            &= \mathsf{E}_F \left[ \exp( \vartheta_i h(X)) \mathds{1}_{\{h(X) \in (q_Y(v_i) 
            , q_Y(w_i) ] \}} \right] \mathsf{E}_F \left[ \exp( - \vartheta_i h(X) )
            \mathds{1}_{\{h(X) > q_Y(1 - \alpha_i) \}} \right].
        \end{align*}
        \item[(ii)] The argument for part (ii) is analogous.
    \end{itemize}
\end{proof}
\begin{lem} \label{lem:CondSampIneq}
    Let $w \leq v \leq 1 - \alpha$. Then:
    $$ \mathsf E_F \left[  \frac{dF}{dF_{\vartheta, v}} (X) \mathds{1}_{\{h(X) > q_Y(1 - \alpha) \}} \right]
    \geq \mathsf E_F \left[  \frac{dF}{dF_{\vartheta, w}} (X) \mathds{1}_{\{h(X) > q_Y(1 - \alpha) \}} \right] $$
    and 
    $$ \mathsf E_F \left[ \frac{dF}{dF^\ast_v}(X) \mathds{1}_{\{h(X) > q_Y(1 - \alpha) \}} \right] 
    \geq \mathsf E_F \left[ \frac{dF}{dF^\ast_w} (X) \mathds{1}_{\{h(X) > q_Y(1 - \alpha) \}}  \right].$$
\end{lem}
\begin{proof}
    To show the first inequality, we have $v \leq 1 - \alpha$ and therefore
    $$ \frac{dF}{dF_{\vartheta, v}} (X) \mathds{1}_{\{h(X) > q_Y(1 - \alpha)\}} 
    = \exp \left( \psi(\vartheta, v) - \vartheta h(X) \right) \mathds{1}_{\{h(X) > q_Y(1 - \alpha)\}},
    $$
    where we set $0 / 0 = 0$. As $\psi(\vartheta, v) = \log \left( \mathsf E_F \left[ \exp( \vartheta h(X)) 
    \mathds{1}_{\{ h(X) > q_Y(v)\}} \right] \right)$ we have for $v \leq w$ that 
    $$ \psi(\vartheta, v) \geq \psi(\vartheta, w).$$
    With the monotonicity of the integral the first inequality of the lemma follows.
    The second inequality is calculated analogously.
    To show the first inequality, note that since $v \leq 1 - \alpha$, we have
    $$ \frac{dF}{dF_{\vartheta, v}} (X) \mathds{1}_{\{h(X) > q_Y(1 - \alpha)\}} 
    = \exp \left( \psi(\vartheta, v) - \vartheta h(X) \right) \mathds{1}_{\{h(X) > q_Y(1 - \alpha)\}},
    $$
    where we define $0 / 0 = 0$.
    Given that $\psi(\vartheta, v) = \log \left( \mathsf E_F \left[ \exp( \vartheta h(X)) 
    \mathds{1}_{\{ h(X) > q_Y(v)\}} \right] \right)$, it follows for $v \leq w$ that 
    $$ \psi(\vartheta, v) \geq \psi(\vartheta, w).$$
    Utilizing the monotonicity of the integral, the first inequality of the lemma follows directly.
    The second inequality can is shown following the same reasoning.
\end{proof}
\begin{rema}
    \item[(i)] 
  The effectiveness of an IS estimation of a DRM with conditioned sampling 
distributions depends critically on the feasibility of efficient simulation 
methods. In particular, acceptance–rejection based on $F_{\vartheta_i}$ as 
proposal may entail a substantial increase in computational cost, since it 
typically requires the generation of many candidate draws. Hence, any practical 
advantage of conditioning hinges on whether the induced sampling overhead is 
offset by the corresponding variance reduction.
    \item[(ii)] In practice, the quantiles of the conditioned sampling distributions are 
unknown and must be estimated from the pivot samples produced by the algorithm. 
Since these estimators are themselves random, instabilities may occur if 
$\hat{q}_{F,M}(v_i) > q_Y(1 - \alpha_i)$ for some $v_i \leq 1 - \alpha_i$. 
To mitigate this risk, the choice of $v_i$ should avoid values too close to 
$1 - \alpha_i$.
\end{rema}

\begin{figure}
    \centering
    \includegraphics[width=1\textwidth]{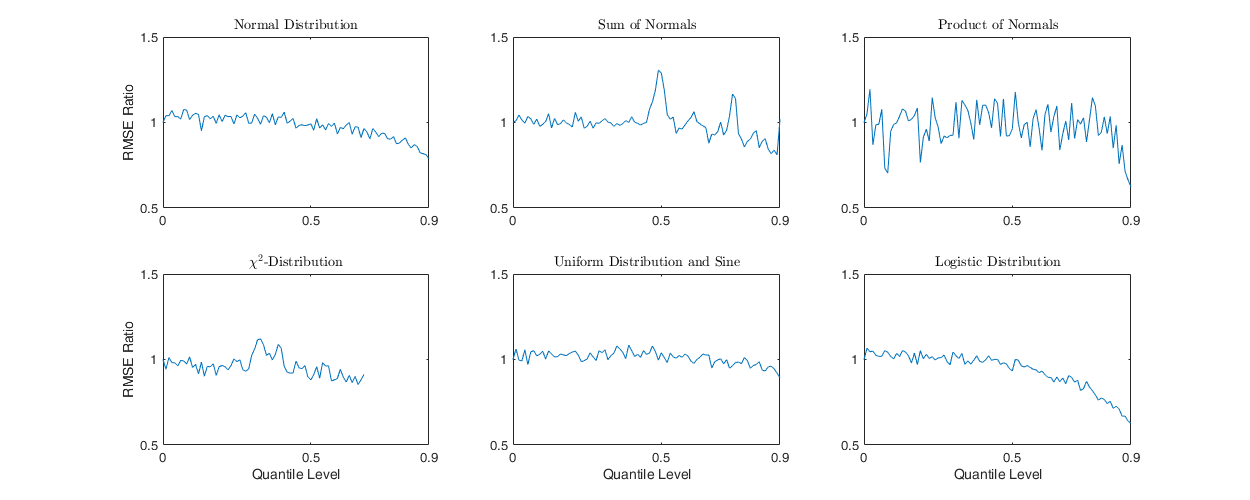}
    \caption{RMSE ratio between unconditional and conditional importance sampling 
    methods for estimating the distortion risk measure $\rho_{g_{\gamma, \alpha}}(Y)$ with 
    $\gamma = 1$ and $\alpha = 0.05$. A total of $N = 20,000$ samples are drawn 
    from the unconditional importance sampling distribution based on the exact function $h$, and only those samples exceeding the estimated quantiles are considered in the DRM estimation.
    The plots depict the RMSE ratio as a function of 
    conditioning levels across the six case studies from Section \ref{sec:CaseStudies}. 
    Some curves are truncated levels due to instability in the 
    estimations.
    }
    \label{fig:condExactMSE_1}
\end{figure}
\begin{figure}
    \centering
    \includegraphics[width=1\textwidth]{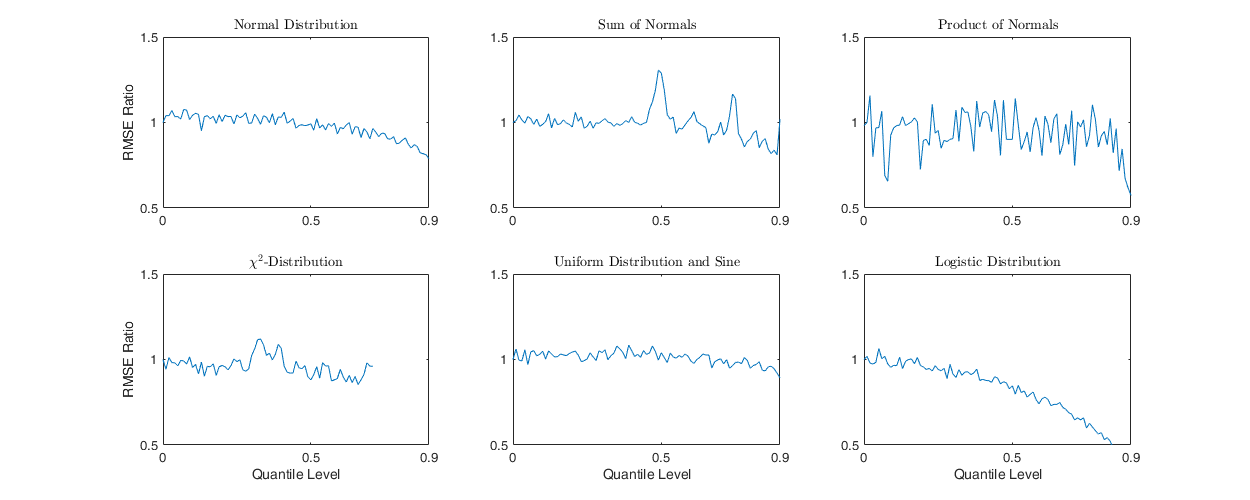}
    \caption{RMSE ratio between unconditional and conditional importance sampling 
    methods for estimating the distortion risk measure $\rho_{g_{\gamma, \alpha}}(Y)$ with 
    $\gamma = 1$ and $\alpha = 0.05$. Exactly $N = 20,000$ samples are drawn 
    from the conditional importance sampling distribution based on the exact function $h$.
    The plots depict the RMSE ratio as a function of 
    conditioning levels across the six case studies from Section \ref{sec:CaseStudies}. 
    Some curves are truncated levels due to instability in the 
    estimations.
    }
    \label{fig:condExactMSE_2}
\end{figure}
\begin{figure}
    \centering
    \includegraphics[width=1\textwidth]{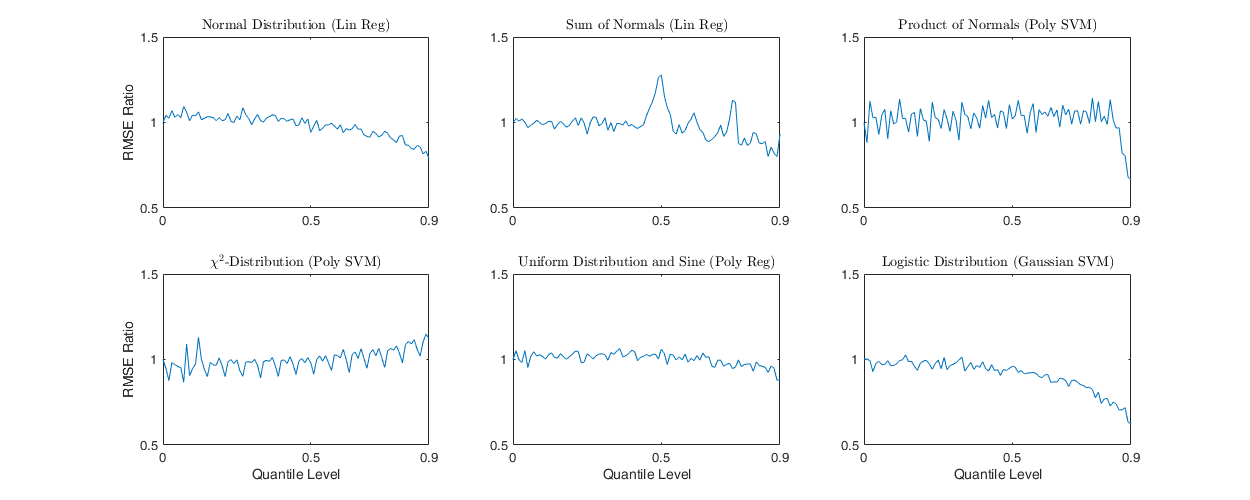}
    \caption{RMSE ratio between unconditional and conditional importance sampling 
    methods for estimating the distortion risk measure $\rho_{g_{\gamma, \alpha}}(Y)$ with 
    $\gamma = 1$ and $\alpha = 0.05$. A total of $N = 20,000$ samples are drawn 
    from the unconditional importance sampling distribution based on the best performing approximation $\hat{h}$ of $h$ from the case studies in Section 
    \ref{sec:CaseStudies}, and only those samples exceeding the estimated quantiles are considered in the DRM estimation.
    The plots depict the RMSE ratio as a function of 
    conditioning levels across the six case studies from Section \ref{sec:CaseStudies}. 
    Some curves are truncated levels due to instability in the 
    estimations.
    }
    \label{fig:condApprox_1}
\end{figure}
\begin{figure}
    \centering
    \includegraphics[width=1\textwidth]{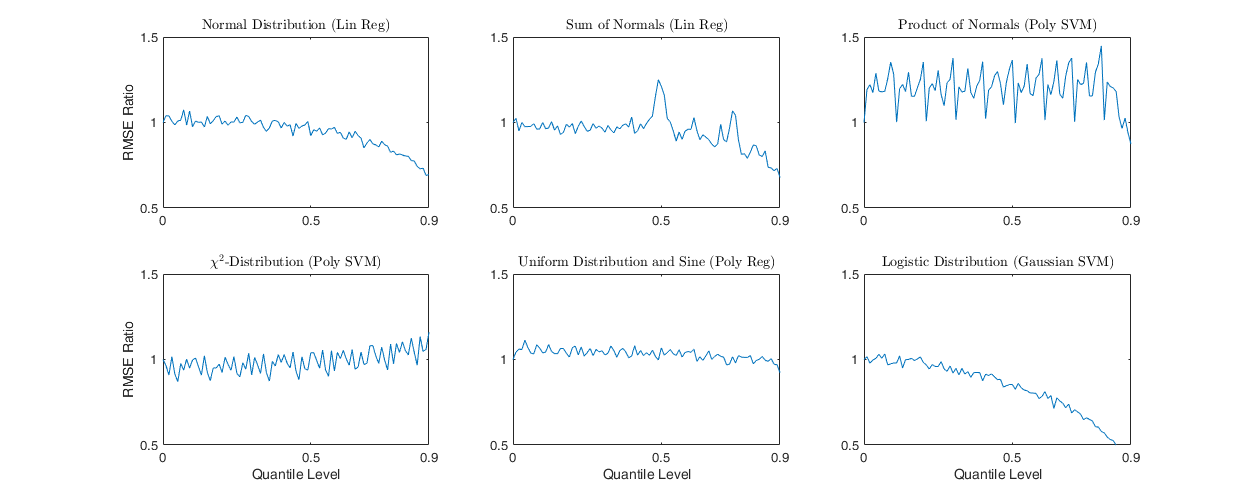}
    \caption{RMSE ratio between unconditional and conditional importance sampling 
    methods for estimating the distortion risk measure $\rho_{g_{\gamma, \alpha}}(Y)$ with 
    $\gamma = 1$ and $\alpha = 0.05$. Exactly $N = 20,000$ samples are drawn 
    from the conditional importance sampling distribution based on the best performing approximation $\hat{h}$ of $h$ from the case studies in Section \ref{sec:CaseStudies}.
    The plots depict the RMSE ratio as a function of 
    conditioning levels across the six case studies from Section \ref{sec:CaseStudies}. 
    Some curves are truncated levels due to instability in the 
    estimations.
    }
    \label{fig:condApprox_2}
\end{figure}

We now turn to numerical case studies of conditional importance sampling, adopting the framework of 
Section~\ref{sec:CaseStudies} for estimating the distortion risk measure 
$\rho_{g_{\gamma,\alpha}}(Y)$ with parameters $\gamma = 1$ and $\alpha = 0.05$ 
across the six case studies.
For each importance sampling method, we generate $M = 5{,}000$ pivot samples.
 These samples serve to estimate the quantile $q_Y(1-v)$ by 
$\hat{q}_{F,M}(1-v)$, which is then used for conditioning the sampling 
distribution:
\begin{itemize}
\item In the indirect conditional sampling method, the algorithm 
draws $N = 20{,}000$ samples from the unconditional importance sampling 
distribution and retains only those exceeding $\hat{q}_{F,M}(1-v)$ for the 
DRM estimation. 
\item In contrast, the direct conditional sampling method 
generates $N = 20{,}000$ samples that exceed $\hat{q}_{F,M}(v)$ 
for the DRM estimation. 
\end{itemize}
The remaining algorithms are implemented in 
accordance with Section~\ref{sec:CaseStudies}. 

\subsubsection{Indirect Conditional Importance Sampling: Results}

The results of the case studies are summarized in Figures \ref{fig:condExactMSE_1} and 
\ref{fig:condApprox_1}.
Figure \ref{fig:condExactMSE_1} presents the ratio of the RMSE for the indirect conditional importance sampling estimation compared to the unconditional approach, based on the exact knowledge of $h$.
In the normal distribution, sum of normals, uniform distribution and sine, and logistic distribution case studies, the RMSE ratio generally decreases with $v$, indicating that the RMSE reduction improves as the conditioning quantile increases.
Notably, in the logistic distribution case, the RMSE is reduced by nearly $40\%$ due to conditioning.
In contrast, for the product of normals, the RMSE ratio exhibits significant fluctuations, preventing a clear trend to be observed. For the $\chi^2$-distribution, the estimation becomes unstable at higher values of $v$, making it impossible to determine a reliable ratio.
Figure \ref{fig:condApprox_1} displays results where the importance sampling distributions are based on the best-performing approximations of $h$ from the case studies in Section \ref{sec:CaseStudies}. Here, the RMSE ratio decreases with $v$ in all but one case study, as the conditioning quantile increases. The sole exception is the $\chi^2$-distribution, where the RMSE ratio appears to increases with increasing $v$.

\subsubsection{Direct Conditional Importance Sampling: Results}

The results of the case studies are summarized in Figures \ref{fig:condExactMSE_2} and \ref{fig:condApprox_2}. Figure \ref{fig:condExactMSE_2} displays the RMSE ratio for the direct conditional importance sampling approach, based on the exact knowledge of $h$.
Similar to the indirect conditional importance sampling method, in the normal distribution, sum of normals, uniform distribution and sine, and logistic distribution case studies, the RMSE ratio decreases with $v$. This indicates a reduction in RMSE as the conditioning quantile becomes larger. For the product of normals, the ratio fluctuates without showing a clear trend. At higher $v$ values, estimating the DRM for the $\chi^2$-distribution becomes unstable.
Figure \ref{fig:condApprox_2} illustrates the RMSE ratio using importance sampling distributions based on the best-performing approximation of $j$ from the case studies in Section \ref{sec:CaseStudies}. Once again, the RMSE ratio decreases with $v$ in all but one case study, corresponding to an increase in the conditioning quantile. The $\chi^2$-distribution is the exception, with an increasing RMSE.

\subsection{Choosing the Size of the Partition}

The partition size \((\alpha_i)_{i \in \{0, 1, \dots, m \}}\) plays a crucial in estimating the DRM. A finer partition reduces discretization error in DRM estimation, yet enlarging the partition size may increase estimation error for individual quantiles. Thus, choosing \(m\) involves balancing these two error types.\\
This trade-off is clearly illustrated in the subsequent proposition:

\begin{prop} \label{prop:PartSize}
    Consider the partition $0 = \alpha_0 < \alpha_1 < \dots < \alpha_m = 1$, 
    with $\beta \in (\alpha_j,\alpha_{j+1})$, and sampling distributions $F_{\vartheta_0}, F_{\vartheta_1}, 
    \dots, F_{\vartheta_m}, 
    F_{\vartheta_\beta}$ satisfying Assumption \ref{ass:ISDist}. \\
    Define $\tilde{N} = (N_0, N_1, \dots, N_m)$, $\tilde{N}^\prime = (N_0, N_1, 
    \dots, N_{j-1}, N^\prime_j, N^\prime_{j+1}, N_{j+1}, \dots, N_m)$, where
    $N^\prime_j + N^\prime_{j+1} = N_j$. \\
    Then 
    \begin{align*}
        \mathcal{E}(F_{\vartheta_0}, F_{\vartheta_1}, \dots, F_{\vartheta_m}, \tilde{N}) 
        \geq \mathcal{E}(F_{\vartheta_0}, F_{\vartheta_1},\dots, F_{\vartheta_j}, F_{\vartheta_\beta}, F_{\vartheta_{j+1}} , \dots,  F_{\vartheta_m}, 
        \tilde{N}^\prime), 
    \end{align*}
    if 
    {\small 
    \begin{align*}
        &\int_{\beta}^{\alpha_{j+1}} (q_Y(1 - u) - q_Y(1 - \beta))^2 - (q_Y(1-u) - q_Y(1 - 
        \alpha_j))^2 dg(u) \\
        \geq &  \frac{\mathsf{E}_{F_{\vartheta_j}} \left[ \frac{dF}{dF_{\vartheta_j}}(X)^2 
        \mathds{1}_{\{h(X) > q_Y(1 - \alpha_j\}} \right] - \alpha_j^2 }{G^\prime(q_Y(1 - \alpha_j))^2 }\left( \frac{1}{N_j} ( g(\alpha_{j+1}) - g(\alpha_j))) 
        - \frac{1}{N_j^\prime} (g(\beta) - g(\alpha_j))\right) \\
        &-\frac{\mathsf{E}_{F_\beta} \left[ \frac{dF}{dF_\beta}(X)^2 
        \mathds{1}_{\{h(X) > q_Y(1 - \beta\}} \right] - \beta^2 }{G^\prime(q_Y(1 - \beta))^2 }
        \frac{g(\alpha_{j+1}) - g(\beta)}{N^\prime_{j+1}}
    \end{align*}
    } 
\end{prop}

\begin{proof} 
    With $q_Y(u)$ being monotone increasing and $N_j \geq N_j^\prime$, we have:
    {\small
    \begin{align*}
        &\mathcal{E}(F_{\vartheta_0}, F_{\vartheta_1}, \dots, F_{\vartheta_m}, \tilde{N}_{IS}) 
        - \mathcal{E}(F_{\vartheta_0}, F_{\vartheta_1}, \dots, F_{\vartheta_m}, F_{\beta}, 
        \tilde{N}^\prime_{IS}) \\
        =& \sum_{i=0}^m \int_{\alpha_i}^{\alpha_{i+1}} (q_Y(1 - u) - q_Y(1 - \alpha_i))^2 
        dg(u) + \frac{\mathsf{E}_{F_{\vartheta_i}} \left[ \frac{dF}{dF_{\vartheta_i}} (X)^2 
        \mathds{1}_{\{h(X) > q_Y(1 - \alpha_i)\}} \right] - \alpha_i^2 }{N_i 
        G^\prime(q_Y(1 - \alpha_i))^2 } (g(\alpha_{i+1}) - g(\alpha_i))  \\
        -& \sum_{\substack{i = 0 \\ i \neq j }}^m  \int_{\alpha_i}^{\alpha_{i+1}} (q_Y(1 - u) - q_Y(1 - \alpha_i))^2 
        dg(u) + \frac{\mathsf{E}_{F_{\vartheta_i}} \left[ \frac{dF}{dF_{\vartheta_i}} (X)^2 
        \mathds{1}_{\{h(X) > q_Y(1 - \alpha_i)\}} \right] - \alpha_i^2 }{N_i 
        G^\prime(q_Y(1 - \alpha_i))^2 } (g(\alpha_{i+1}) - g(\alpha_i)) \\
        -& \int_{\alpha_j}^\beta (q_Y(1 -u) - q_Y(1 - \alpha_j))^2 dg(u) - 
        \int_{\beta}^{\alpha_{j+1}} (q_Y(1 -u) -q_Y(1 - \beta))^2 dg(u) \\
        -& \frac{\mathsf{E}_{F_{\vartheta_j}} \left[ \frac{dF}{dF_{\vartheta_j}} (X)^2 
        \mathds{1}_{\{h(X) > q_Y(1 - \alpha_j)\}} \right] - \alpha_i^2 }{N^\prime_j 
        G^\prime(q_Y(1 - \alpha_j))^2 } (g(\beta) - g(\alpha_i)) \\
        -&\frac{\mathsf{E}_{F_{\vartheta_\beta}} \left[ \frac{dF}{dF_{\vartheta_\beta}} (X)^2 
        \mathds{1}_{\{h(X) > q_Y(1 -  \beta)\}} \right] - \beta^2 }{N^\prime_{j+1} 
        G^\prime(q_Y(1 - \alpha_j))^2 } (g(\alpha_{j+1}) - g(\beta))        \\
        =& \underbrace{\int_{\beta}^{\alpha_{j+1}} (q_Y(1-u) - q_Y(1 - \alpha_j))^2 dg(u) 
        - \int_{\beta}^{\alpha_{j+1}} (q_Y(1-u) - q_Y(1 - \beta))^2 dg(u)}_{\geq 0}   \\
        +& \frac{\mathsf{E}_{F_{\vartheta_j}} \left[ \frac{dF}{dF_{\vartheta_j}}(X)^2 
        \mathds{1}_{\{h(X) > q_Y(1 - \alpha_j\}} \right] - \alpha_j^2 }{G^\prime(q_Y(1 - \alpha_j))^2 }\left( \frac{1}{N_j} ( g(\alpha_{j+1}) - g(\alpha_j))) 
        - \frac{1}{N_j^\prime} (g(\beta) - g(\alpha_j))\right)\\
        -& \underbrace{ \frac{\mathsf{E}_{F_\beta} \left[ \frac{dF}{dF_\beta}(X)^2 
        \mathds{1}_{\{h(X) > q_Y(1 - \beta\}} \right] - \beta^2 }{G^\prime(q_Y(1 - \beta))^2 }
        \frac{g(\alpha_{j+1}) - g(\beta)}{N^\prime_\beta}}_{\geq 0}.
    \end{align*}
    }
\end{proof}
In studying the convergence rates of the MSE relative to partition size and sample size, we can assert:

\begin{prop} \label{prop:ConvIneq}
    Assume $Y \in L^\infty$ with a continuous quantile function $q_Y(\cdot)$ and 
    a continuous distortion function $g$. Define the partition 
    $\alpha_i$, $i \in \{0,1,\dots, m+1\}$, such that $g(\alpha_{i+1}) - g(\alpha_i) = 
    \frac{1}{m+1}$. 
    Let $\vartheta(\cdot) : [0,1] \mapsto \mathbb{R}$ and 
    $\tilde{N} = (N_0, N_1, \dots, N_m)$.
    Then the estimation error satisfies: 
   \begin{align*}
       \mathcal{E} (F_{\vartheta(\alpha_0)}, F_{\vartheta(\alpha_1)},
       \dots, F_{\vartheta(\alpha_m)}, \tilde{N}) \leq \frac{C_m }{m+1} + \frac{D}{m+1} \sum_{i=0}^m \frac{1}{N_i},
   \end{align*}
   where:
   \begin{itemize}
       \item $C_m \geq 0$ depends only on the partition and distribution of $Y$, with $\lim_{m \rightarrow \infty} C_m = 0$. 
       \item $D \geq 0$ is a constant depending solely on $\vartheta(\cdot)$, 
       independent of $m$ and $N$. 
   \end{itemize}
   If $N_i = \frac{N}{m+1}$, then:
   \begin{align*}
       \mathcal{E} (F_{\vartheta(\alpha_0)}, F_{\vartheta(\alpha_1)},
       \dots, F_{\vartheta(\alpha_m)},  \tilde{N})
       \leq \frac{C_m }{m+1} + \frac{(m+1)D}{N}.
   \end{align*}
\end{prop}
\begin{proof}
    Consider the estimation error: 
    {\footnotesize
    \begin{align*}
        & \mathcal{E}( F_{\vartheta(\alpha_0)}, F_{\vartheta(\alpha_1)}, 
        \dots, F_{\vartheta(\alpha_m)}, \tilde{N}) \\
        =& \underbrace{\sum_{i=0}^m \int_{\alpha_i}^{\alpha_{i+1}} 
        (q_Y(1 - u) - q_Y(1 - \alpha_i))^2 dg(u) }_{(i)}
        + \underbrace{\sum_{i=0}^m \frac{\mathsf{E}_{F_{\vartheta(\alpha_i)}} \left[ 
        \frac{dF}{dF_{\vartheta(\alpha_i))}} (X)^2 \mathds{1}_{\{h(X) > 
        q_Y(1 - \alpha_i) \}}  \right] - \alpha_i^2}{N_i G^\prime(q_Y(1 
        - \alpha_i))^2 } (g(\alpha_{i+1}) - g(\alpha_i))}_{(ii)}.
    \end{align*}}
    We analyze the discretization error $(i)$ and estimation error $(ii)$ as follows:
    \begin{itemize}
        \item[(i)]
            For the discretization error, since $q_Y(1 - u)$) is decreasing in $u$ and  $g(\alpha_{i+1})- g(\alpha_i) = \frac{1}{m+1}$, we find:
            \begin{align*}
                &\sum_{i=0}^m \int_{\alpha_i}^{\alpha_{i+1}} 
                (q_Y(1 - u) - q_Y(1 - \alpha_i))^2 dg(u) \\
                \leq & 
                \sum_{i=0}^m \int_{\alpha_i}^{\alpha_{i+1}} (
                q_Y(1 - \alpha_i) - q_Y(1 - \alpha_{i+1}))^2 dg(u) \\
                =& \sum_{i=0}^m (q_Y(1 - \alpha_i) - q_Y(1 - \alpha_{i+1}))^2 
                (g(\alpha_{i+1})- g(\alpha_i)) \\
                =& \frac{1}{m+1} \sum_{i=0}^m ( q_Y(1 - \alpha_{i}) - q_Y(1 -
                \alpha_{i+1}))^2 \\
                \leq & \frac{1}{m+1} \max_{i \in \{0,1, \dots, m\}} 
                ( q_Y(1 - \alpha_{i}) - q_Y(1 -
                \alpha_{i+1})) \sum_{i=0}^m ( q_Y(1 - \alpha_{i}) - q_Y(1 -
                \alpha_{i+1})) \\
                \leq & \frac{1}{m+1} C_m^\prime \sum_{i=0}^m q_Y(1 - \alpha_i) - 
                q_Y(1 - \alpha_{i+1}) \\
                \leq & \frac{1}{m+1} C_m^\prime 2 \Vert Y \Vert_\infty =: \frac{C_m}{m+1}.
            \end{align*}
            Here, $C^\prime_m$ is defined by:
            \begin{align*}
                C^\prime_m := \sup_{\substack{0 \leq v \leq u \leq 1 \\ 
                v -u \leq \max_i \alpha_{i+1} - \alpha_i  }} 
                q_Y(1 - v) - q_Y(1 -u).
            \end{align*}
            The continuity of $q_Y(1 - u)$ implies $\lim_{m 
            \rightarrow \infty} C_m^\prime = 0$ and hence $\lim_{m 
            \rightarrow \infty} C_m = 0$.  
        \item[(ii)] 
        For the estimation error, define:
        $$ D := \sup_{\alpha \in [0,1]} 
        \frac{\mathsf{E}_{F_{\vartheta(\alpha)}} \left[ 
            \frac{dF}{dF_{\vartheta(\alpha)}} (X)^2 \mathds{1}_{\{h(X) > 
            q_Y(1 - \alpha) \}}  \right] - \alpha}{G^\prime(q_Y(1 
            - \alpha))^2 }$$
        Given $g(\alpha_{i+1}) - g(\alpha_i) = \frac{1}{m+1}$, we get:
        \begin{align*}
            & \sum_{i=0}^m \frac{\mathsf{E}_{F_{\vartheta(\alpha_i)}} \left[ 
            \frac{dF}{dF_{\vartheta(\alpha_i)}} (X)^2 \mathds{1}_{\{h(X) > 
            q_Y(1 - \alpha_i) \}}  \right] - \alpha_i^2}{N_i G^\prime(q_Y(1 
            - \alpha_i))^2 }  (g(\alpha_{i+1}) - g(\alpha_i))\\
            =& \frac{1}{m+1}\sum_{i=0}^m \frac{\mathsf{E}_{F_{\vartheta(\alpha_i)}} \left[ 
            \frac{dF}{dF_{\vartheta(\alpha_i)}} (X)^2 \mathds{1}_{\{h(X) > 
            q_Y(1 - \alpha_i) \}}  \right] - \alpha_i^2}{N_i G^\prime(q_Y(1 
            - \alpha_i))^2 } \\
            \leq & \frac{D}{m+1} \sum_{i=0}^m \frac{1}{N_i}. 
        \end{align*}
        Assuming $N_i = \frac{N}{m+1}$, we find:
        $$ \frac{D}{m+1} \sum_{i=0}^m \frac{1}{N_i} = \frac{(m+1) D}{N}. $$
    \end{itemize}
    Combining results $(i)$ and $(ii)$, we deduce bounds 
    for $\mathcal{E}( F_{\vartheta(\alpha_0)}, 
    F_{\vartheta(\alpha_1)}, \dots, F_{\vartheta(\alpha_m)}, \tilde{N}) $
    as stated in the proposition. 
\end{proof}
\begin{rema}
\begin{itemize}
    \item[(i)]
    The inequality in Proposition \ref{prop:PartSize} is influenced by several factors, 
    , such as $G^\prime(\cdot)$, the density of the distribution of $h(X)$, the chosen distortion function $g(\cdot)$, and the reduction in the second moment of the individual sampling distributions. Generally, refining the partition improves estimation performance only if the reduction in discretization error outweighs the additional estimation error from individual quantiles.
    \item[(ii)] 
    The inequality in Proposition \ref{prop:ConvIneq} highlights a trade-off between discretization and estimation errors in choosing partition size $m$ and sample sizes $N_i$. Specifically:

    \begin{itemize}
        \item Increasing $m$ refines the partition, reducing discretization error. However, this raises estimation error unless $N_i$ are also increased.
        \item Enlarging $N_i$ decreases estimation error, improving quantile accuracy, but leaves discretization error unaffected since it's only dependents on $m$.
    \end{itemize} 
    Thus, minimizing the total error requires a balance between 
    $m$ and $N_i$. The aim is to ensure the reduced discretization error from a finer partition doesn't result in excessive estimation error due to small sample sizes.
    \item[(iii)]
    The convergence rate in Proposition \ref{prop:ConvIneq} hinges on how quickly $C_m$ approaches zero. Given that we have only assumed $L^\infty$ and the continuity of the quantile function, we cannot infer the rate at which $C_m$ vanishes. The term $C_m$ reflects the local variation of the quantile function, which is not solely dictated by its boundedness, continuity, and monotonicity. To determine the rate of convergence accurately, it is essential to consider the modulus of continuity of the quantile function.
\end{itemize}
\end{rema} 

\subsection{Estimation Error under Black Box Approximations}

To investigate the impact of errors from the use of ML approximations on the precision of IS estimations for DRMs, we 
present the following result. This result demonstrates that approximations of $h$ induce errors in the likelihood ratios used in IS estimation, which 
subsequently affect DRM estimation. 
We quantify this error below.
\begin{prop}
Let $\hat{h}$ be an approximation of $h$. Define
\begin{align*}
    d\hat{F}_{\vartheta_i}(x) &= \exp( \vartheta_i \hat{h}(x) - \hat{\psi}(\vartheta_i)) 
    dF(x) \\
    d \hat{F}^\ast &=  \sum_{i=0}^m p_i d\hat{F}_{\vartheta_i}(x),
\end{align*}
where $i \in \{0,1, \dots, m \}$ and $\hat{\psi}(\vartheta) = \log( \mathsf{E}_F [
\exp( h(X)) ])$. Under assumptions \ref{ass:ISDist}: 
\begin{itemize}
    \item[(i)]
        \begin{align*}
            & \mathcal{E}(F_{\vartheta_0}, F_{\vartheta_1}, \dots, F_{\vartheta_m}, 
            \bar{N}_{Ind}) - \mathcal{E}(\hat{F}_{\vartheta_0}, \hat{F}_{\vartheta_1}, \dots, 
            \hat{F}_{\vartheta_m}, \bar{N}_{Ind}) \\
            =& \sum_{i=0}^m \frac{g(\alpha_{i+1}) - g(\alpha_i)}{N_i G^\prime( q_Y(1 - \alpha_i))^2} \\ 
            &\mathsf{E}_F \left[  \exp( \psi(\vartheta_i) - \vartheta_i h(x) ) 
            \left( 1 - \exp( \hat{\psi}(\vartheta_i) - \psi(\vartheta_i) - \vartheta_i 
            ( h(x) - \hat{h}(x)))\right)  \mathds{1}_{\{h(X) > q_Y(1 - \alpha_i)} \right],
        \end{align*}
        where $\bar{N}_{Ind} = (N_0, N_1, \dots, N_m)^T$ and $\sum N_i = N$. 
    \item[(ii)]
    \begin{align*}
        & \mathcal{E}(F^\ast, F^\ast, \dots, F^\ast, \tilde{N}_{Mix}) - 
        \mathcal{E}(\hat{F}^\ast, \hat{F}^\ast, \dots, \hat{F}^\ast, \tilde{N}_{Mix})) \\
        =&  \sum_{i=0}^m \frac{g(\alpha_{i+1}) - g(\alpha_i)}{N G^\prime(q_Y(1 - \alpha_i))^2}
        \\
        &\cdot \mathsf{E}_F \left[ \left( \frac{1}{\sum_{j=0}^m 
        p_j \exp( \vartheta_j h(x) - \psi(\vartheta_i))} - 
        \frac{1}{\sum_{j=0}^m p_j \exp( \vartheta_j \hat{h}(x) - \hat{\psi}(\vartheta_i))}
        \right) \mathds{1}_{\{h(X) > q_Y(1 - \alpha_i)\}} \right],
    \end{align*}
    where $\bar{N}_{Mix} = (N, N, \dots, N)^T$. 
\end{itemize}
\end{prop}
\begin{proof}
    For arbitrary sampling distributions $F_i$, $F^\prime_i$, $i \in \{0,1, \dots, m \}$
    satisfying assumption \ref{ass:ISDist}, 
    and a sample allocation $\tilde{N} = (N_0, N_1, \dots, N_m)^T$, we have:
    \begin{align*}
        &\mathcal{E}(F_0, F_1, \dots, F_m, \tilde{N}) - \mathcal{E}(F_0^\prime, F_1^\prime, 
        \dots, F_m^\prime, \tilde{N}) \\
        =& \sum_{i=0}^m \frac{g(\alpha_{i+1}) - g(\alpha_i)}{N_i G^\prime(q_Y(1 - \alpha_i))^2}
        \mathsf{E}_F \left[ \left( \frac{dF}{dF_i}(X) - \frac{dF}{dF_i}(X)
        \right) \mathds{1}_{\{h(X) > q_Y(1 - \alpha_i)\}} \right].
    \end{align*}
    Furthermore, the differences in the Radon-Nikodym derivatives are given as:
    \begin{align*}
        \frac{dF}{dF_{\vartheta}}(X) - \frac{dF}{d\hat{F}_{\vartheta}}(X) 
        &= \exp( \psi(\vartheta) - \vartheta h(x)) - \exp( \hat{\psi}(\vartheta) 
        - \vartheta \hat{h}(x)) \\
        &=  \exp( \psi(\vartheta) - \vartheta_i h(x) ) 
            \left( 1 - \exp( \hat{\psi}(\vartheta) - \psi(\vartheta) - \vartheta 
            ( h(x) - \hat{h}(x)))\right) \\
       \frac{dF}{dF^\ast}(X) - \frac{dF}{d\hat{F}^\ast}(X) 
        &  = \frac{1}{\sum_{j=0}^m p_j \exp( \vartheta_j h(x) - \psi(\vartheta_j))}
        - \frac{1}{\sum_{j=0}^m p_j \exp( \vartheta_j \hat{h}(x) - \hat{\psi}(\vartheta_j))}.
    \end{align*}
    By substituting these  in $\mathcal{E}(F_0, F_1, \dots, F_m, \tilde{N}) - \mathcal{E}(F_0^\prime, 
    F_1^\prime, \dots, F_m^\prime, \tilde{N})$  we obtain the formulas from the proposition. 
\end{proof}

\subsection{Time Efficiency for Individual Quantile Estimations} \label{app:TimeEff}

In Section~\ref{sec:TimeEff}, we analyzed the time efficiency of the crude and mixture IS estimators.
We now turn to the time efficiency of individual importance sampling, as quantified by the error bound~(\ref{eq:IndErrorBound}).
We denote by $T_{CR}(N_{CR},m)$ and $T_{IS}(M,N_{IS},m)$ the computation 
times of the crude and the importance sampling estimators, respectively. 
The parameters have the following roles: $N_{CR}$ is the crude sample size, 
$M$ is the number of pivot samples, $N_{IS}$ the importance sampling size, 
and $m$ the partition size.
We determine the sample sizes under the condition that the crude and 
importance sampling estimators share the same approximate upper error bound 
according to (\ref{eq:IndErrorBound}).
\begin{lem} \label{lem:IndTimeEff}
    Let $F_0, F_1,\dots, F_m$ be distribution functions. 
    Assume that for all $i \in \{0,1,\dots, m \}$ $F$ is absolute continuous 
    with respect to $F_i$ and Assumptions \ref{ass:ISDist} hold. 
    Then, if the samples are drawn i.i.d. from the $F_i$ according to the allocation $\tilde{N}_{IS} = (N_0, N_1, \dots, N_m)$ we have: 
    $$
       \mathcal{E}(F, \dots, F, \tilde{N}_{CR}) = \mathcal{E}(F_0, F_1, 
        \dots, F_m, \tilde{N}_{IS}) 
        \quad \Longleftrightarrow \quad N_{CR} = \frac{\sum_{i=0}^m \tilde{V}(1 -\alpha_i, F) (g(\alpha_{i+1})- g(\alpha_i))^2  }{\sum_{i=0}^m \frac{\tilde{V}(1 -\alpha_i, F_{i})}{N_i} (g(\alpha_{i+1})- g(\alpha_i))^2}.
    $$
\end{lem}
\begin{proof}
    We have:
    \begin{align*}
        & \tilde{\mathcal{E}}(F, \dots, F, \tilde{N}_{CR}) = \tilde{\mathcal{E}}
        (F_0, F_1, \dots, F_m, \tilde{N}_{IS}) \\
        \Longleftrightarrow \quad & \sqrt{\frac{1}{N_{CR}}\sum_{i=0}^m \tilde{V}(1 -\alpha_i, F)(g(\alpha_{i+1})- g(\alpha_i))^2 } = \sqrt{\sum_{i=0}^m \frac{\tilde{V}(1 -\alpha_i, F_{i})}{N_i} (g(\alpha_{i+1})- g(\alpha_i))^2 } \\
        \Longleftrightarrow  \quad & N_{CR} =\frac{\sum_{i=0}^m \tilde{V}(1 -\alpha_i, F) (g(\alpha_{i+1})- g(\alpha_i))^2  }{\sum_{i=0}^m \frac{\tilde{V}(1 -\alpha_i, F_{i})}{N_i} (g(\alpha_{i+1})- g(\alpha_i))^2}.
    \end{align*}
\end{proof}
By accounting for the same computational components as in Section 
\ref{sec:TimeEff}, we obtain the following result:
\begin{prop}
Assume $N_{CR} > M + N_{IS}$, and let $N_{CR}$ be chosen according to Lemma~\ref{lem:IndTimeEff}. Then
\[
   T_{CR}(N_{CR},m) - T_{IS}(M,N_{IS},m) > 0
\]
whenever
\[
   t_h\big(N_{CR} - (N_{IS}+M)\big) > t_{Mix}(M,m) + t_{kFold}(M) 
   + t_{MH}(N_{IS}) + t_{Norm}(m).
\]
\end{prop}
\begin{proof}
    This is analogous to the proof of Proposition \ref{prop:TimeEff}.
\end{proof}

\subsection{Computational Resources}
All case studies were implemented in MATLAB and executed on the cluster system of Leibniz Universität Hannover, using nodes of varying specifications. In every case study, the dominant contribution to computation time arose from the evaluation of the normalizing constant.
\\
For the case studies of Section \ref{sec:CaseStudies} using the quadrature methods of Section \ref{subsec:ImpIssues}, the corresponding ranges of computation times are reported in Table \ref{tab:CompTimes}.
\begin{table}[h]
    \centering
\begin{tabularx}{\linewidth}{| X | X | X |}
    \hline
    Method & Fastest Calculation & Slowest Calculation\\
    \hline 
    Exact IS & $1$s - Normal Distribution & $15$s - Unif. Dist. and Sine \\
    Gaussian SVM IS & $67$s - Normal Distribution & $1463$s - $\chi^2$-Distribution \\
    $k$-NN IS & $7$s - Logistic Distribution & $540$s - Product of Normals \\
    Lin Reg IS & $3$s - Logistic Distribution &  $30$s - $\chi^2$-Distribution \\
    Poly SVM IS & $2$s - Normal Distribution & $750$s - $\chi^2$-Distribution. \\
    \hline
\end{tabularx} 
\caption{Slowest and fastest calculation times for the case studies in Section
\ref{sec:CaseStudies}.}
\label{tab:CompTimes}
\end{table}
The estimation times are sensitive to the choice of hyperparameters in the ML components, in particular for polynomial SVMs and $k$-NN regression.
\\
Using kernel-smoothing density estimates at $100$ randomly chosen points for the normalizing constant reduces computation times to below $90$s in all cases.
\\
The estimations for the iterative exploration of the extreme tail in Section \ref{sec:itexext}
were implemented with the quadrature formulas outlined in Section \ref{subsec:ImpIssues}. 
For the identity of normals, the DRM estimation required about $5$s for both the non-iterative and iterative IS methods.
 Estimating the sum of normals took $5$ seconds with the non-iterative IS and $15$ seconds with the iterative IS. For the product of normals, computation times were $48$ seconds for the non-iterative IS and $120$ seconds for the iterative IS. Finally, estimating the sum of squared normals required $493$ seconds for the non-iterative IS and $1107$ seconds for the iterative IS. \\
In the ALM case studies discussed in Section \ref{sec:ALM}, the estimation of the DRMs with exact IS required $116$s, with the Gaussian SVM 
IS $245$s, with the $k$-NN IS $99$s, with the linear regression IS $131$s and with the linear 
SVM $284$s.

\subsection{A Sharper Error Bound for Individual Importance Sampling} \label{app:SharperErrorBound}

The error bound (\ref{eq:ErrorBound}) does not take advantage of the independence of 
the quantile estimates when individual IS is applied. Consequently, the error bound
in (\ref{eq:ErrorBound}) remains applicable to both individual
and pooled sample allocation, as discussed in Section \ref{subsec:EffUse}. 
However, by concentrating on IS sampling and leveraging the 
independence of the quantile estimates, we can derive a more precise inequality for the MSE in the DRM estimation.
\begin{prop}
    Consider the estimator $\hat{\rho}_g(Y)$ defined in (\ref{eq:dicretization}) and the individual sample allocation. Then:
    \begin{align}
        \sqrt{\mathsf{E}\left[ ( \rho_g(Y) - \hat{\rho}_g(Y))^2 \right]} 
        \lessapprox DE(m, \tilde{N}) +\sqrt{  \sum_{i=0}^m  \frac{\tilde{V}(1 - \alpha_i, F_{\vartheta_{i}^\ast})}{N_i} (g(\alpha_{i+1})-g(\alpha_i))^2}
        , \label{eq:IndErrorBound} \\
        =: \tilde{\mathcal{E}}(F_{\vartheta_0^\ast}, F_{\vartheta_1^\ast},\dots,  F_{\vartheta_m^\ast}, \tilde{N}, m), \notag
    \end{align}
    where:
    \begin{align*}
        &DE(m, \tilde{N}) = \sum_{i=0}^m q_Y(1 - \alpha_i) (g(\alpha_{i+1}) - g(\alpha_i)) - \rho_g(Y),\\
        &\tilde{V}(1 - \alpha_i, F_{\vartheta_{i}^\ast}) =  \frac{\mathsf{E}_{F_{\vartheta_i^\ast}} \left[ \frac{dF}{dF_{\vartheta_i^\ast}}(X)^2 \mathds{1}_{\{h(X) > q_Y(1 - \alpha_i)\}} 
	\right] - \alpha_i^2}{ G^\prime(q_Y(1 - \alpha_i))^2}.
    \end{align*}
\end{prop}
\begin{proof}
    Let $\bar{q}_Y(1 - u) = \sum_{i=0}^m \mathds{1}_{\{u \in [\alpha_i, \alpha_{i+1})\}} q_Y(1 - \alpha_i)$.
    Applying the triangle inequality, we obtain:
    \begin{align*}
    &\sqrt{\mathsf{E} \left[ (\rho_g(Y) - \hat{\rho}_g(Y))^2] \right]} \\
    \leq &\underbrace{\sqrt{ \mathsf{E} \left[ \left( \int_0^1 q_Y(1 - u) - \bar{q}_Y(1 -u) dg(u) \right)^2 \right]}}_{=: DE(m, \tilde{N})} 
    + \underbrace{\sqrt{\mathsf{E} 
    \left[\left( \int_0^1 \bar{q}_Y(1 - u) - \hat{q}_Y(1 - u) dg(u) \right)^2 \right]}
    }_{ (\ast) }.
\end{align*}
For the discretization error, we express:
\begin{align*}
    DE(m, \tilde{N}) &= \left( \mathsf{E} \left[ \left(\sum_{i=0}^m \int_{[\alpha_i, \alpha_{i+1})}q_Y(1 - u) - q_Y(1 - \alpha_i) dg(u) \right)^2 \right] \right)^{\frac{1}{2}} \\
    &=\left\vert \sum_{i=0}^m \int_{[\alpha_i, \alpha_{i+1})} q_Y
    (1-u) - q_Y(1 - \alpha_i) dg(u) \right\vert \\
    &= \sum_{i=0}^m q_Y(1 - \alpha_i) (g(\alpha_{i+1}) - g(\alpha_i)) - 
    \int_0^1 q_Y(1 -u) dg(u) \\
    &= \sum_{i=0}^m q_Y(1 - \alpha_i) (g(\alpha_{i+1}) - g(\alpha_i)) - \rho_g(Y).
\end{align*}
For the estimation error, using Theorem \ref{theo:CLT} and the independence of the quantile estimators, we have: 
\begin{align*}
 (\ast) &=\mathsf{E} \left[ \left( \sum_{i=0}^m q_Y(1 - 
        \alpha_i) - \hat{q}_{F_{\vartheta_i^\ast}, N_i}(1 - \alpha_i)) ( g(\alpha_{i+1}) - g(\alpha_i)\right) ^2 \right] \\
        &= \sum_{i=0}^m \frac{\mathsf{E}_{F_{\vartheta_i^\ast}} \left[ \frac{dF}{dF_{\vartheta_i^\ast}}(X)^2 \mathds{1}_{\{h(X) > q_Y(1 - \alpha_i)\}} 
	       \right] - \alpha_i^2}{N_i G^\prime(q_Y(1 - \alpha_i))^2}
	       (g(\alpha_{i+1} ) - g(\alpha_i) )^2.
\end{align*}
\end{proof}
Following the approach to derive (\ref{theo:OptAllo}), we obtain the optimal allocation by minimizing $\sum_{i=0}^m \frac{\tilde{V}(1 - \alpha_i, F_{\vartheta_i^\ast})}{N_i} (g(\alpha_{i+1}) - g(\alpha_i))^2$ under the constraint $\sum_{i=0}^m N_i = N$. Using calculations similar to those in Appendix \ref{app:OptAllo}, we find:
$$ \tilde{N}_i^\ast = N \frac{\tilde{c}_i}{\sum_{j=0}^m \tilde{c}_j}, \quad i = 0,1,\dots, m,$$
where
$$ \tilde{c}_i =  \frac{\sqrt{\mathsf{E}_{F_{\vartheta_i^\ast}} \left[ \frac{dF}{dF_{\vartheta_i^\ast}}(X)^2 \mathds{1}_{\{h(X) > q_Y(1 - \alpha_i)\}} 
	       \right] - \alpha_i^2}}{ G^\prime(q_Y(1 - \alpha_i))}
(g(\alpha_{i+1} ) - g(\alpha_i) ).$$

\subsection{Additional Plots}\label{app:AddPlots}

\begin{figure}[H]
    \centering
    \begin{minipage}[c]{0.32\textwidth}
        \centering
        \includegraphics[width=1.11\textwidth]{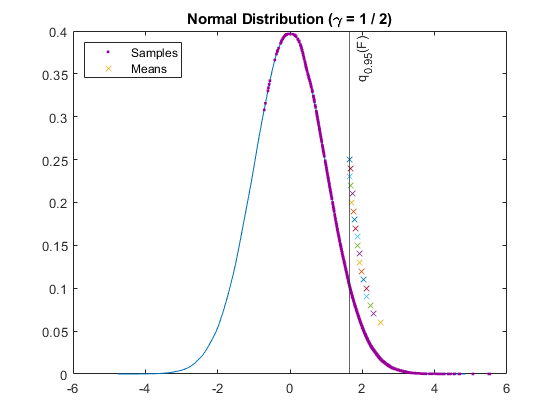}
    \end{minipage}
    \begin{minipage}[c]{0.32\textwidth}
        \centering
        \includegraphics[width=1.11\textwidth]{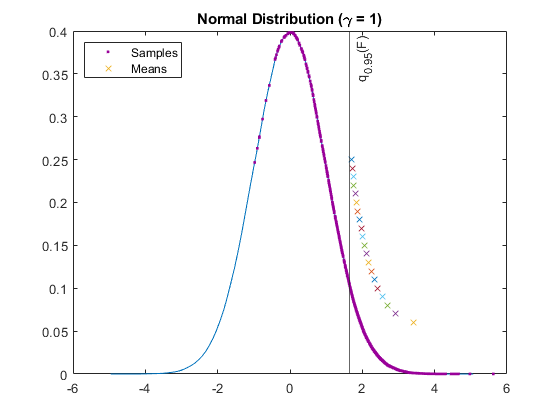}
    \end{minipage}
     \begin{minipage}[c]{0.32\textwidth}
        \centering
        \includegraphics[width=1.11\textwidth]{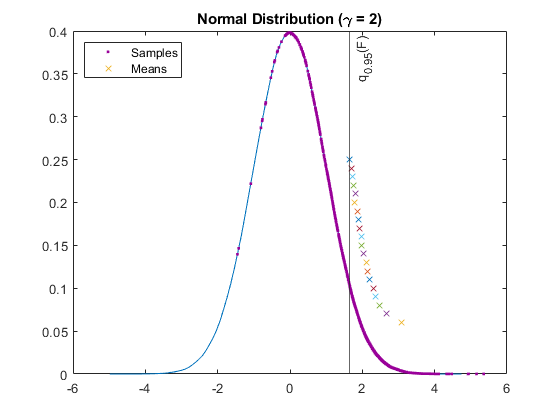}
    \end{minipage}
    \begin{minipage}[c]{0.32\textwidth}
        \centering
        \includegraphics[width=1.11\textwidth]{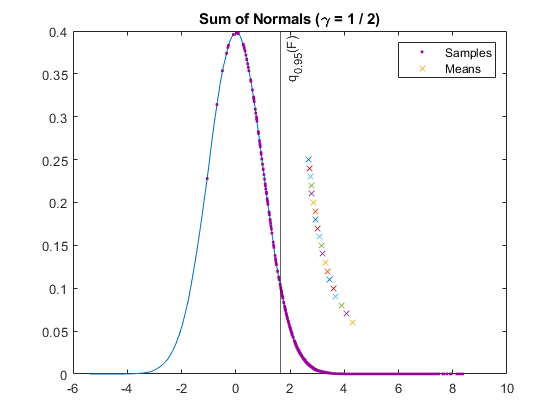}
    \end{minipage}
    \begin{minipage}[c]{0.32\textwidth}
        \centering
        \includegraphics[width=1.11\textwidth]{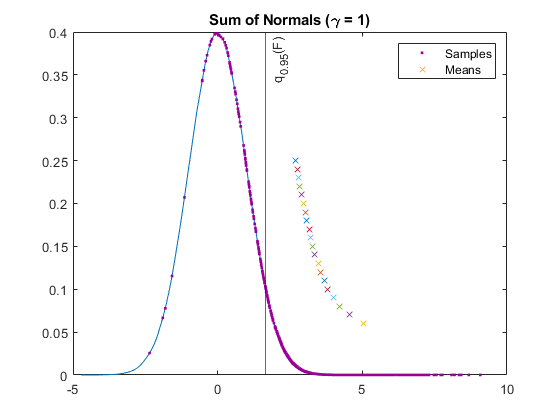}
    \end{minipage}
    \begin{minipage}[c]{0.32\textwidth}
        \centering
        \includegraphics[width=1.11\textwidth]{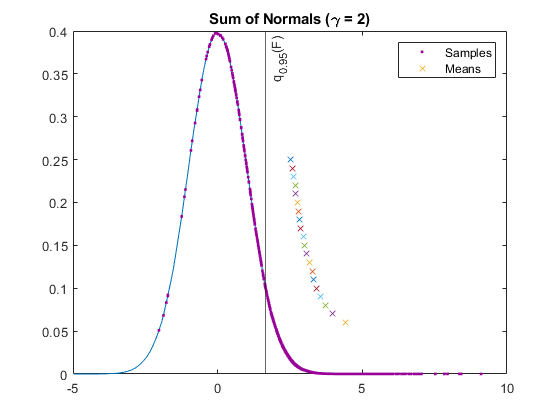}
    \end{minipage}
    \begin{minipage}[c]{0.32\textwidth}
        \centering
        \includegraphics[width=1.129\textwidth]{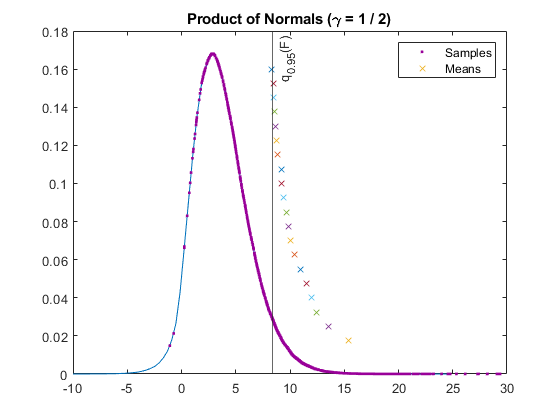}
    \end{minipage}
    \begin{minipage}[c]{0.32\textwidth}
        \centering
        \includegraphics[width=1.11\textwidth]{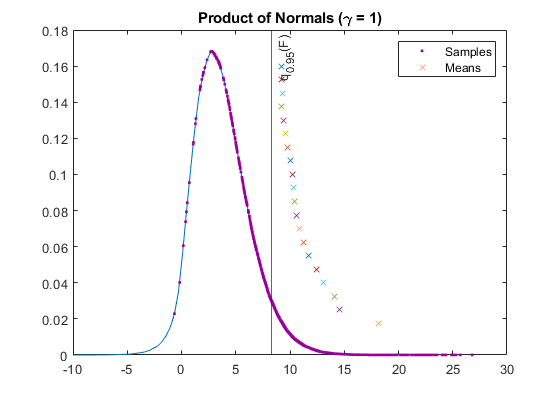}
    \end{minipage}
    \begin{minipage}[c]{0.32\textwidth}
        \centering
        \includegraphics[width=1.11\textwidth]{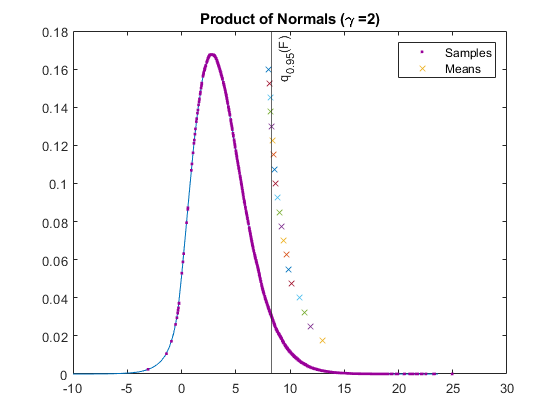}
    \end{minipage}
    \begin{minipage}[c]{0.32\textwidth}
        \centering
        \includegraphics[width=1.11\textwidth]{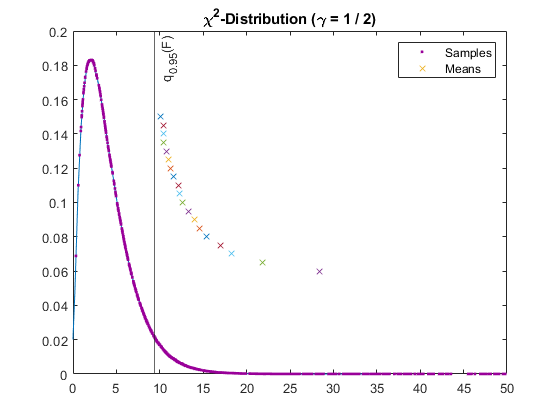}
    \end{minipage}
    \begin{minipage}[c]{0.32\textwidth}
        \centering
        \includegraphics[width=1.11\textwidth]{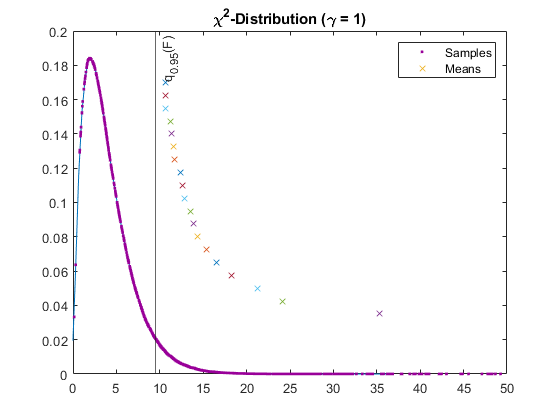}
    \end{minipage}
    \begin{minipage}[c]{0.32\textwidth}
        \centering
        \includegraphics[width=1.11\textwidth]{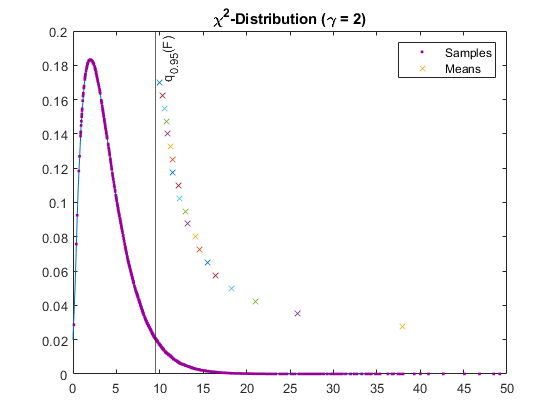}
    \end{minipage}
    \begin{minipage}[c]{0.32\textwidth}
        \centering
        \includegraphics[width=1.11\textwidth]{Picture/UnifTrans/SampleDist_BB_conv.png}
    \end{minipage}
    \begin{minipage}[c]{0.32\textwidth}
        \centering
        \includegraphics[width=1.11\textwidth]{Picture/UnifTrans/SampleDist_BB_avar.png}
    \end{minipage}
    \begin{minipage}[c]{0.32\textwidth}
        \centering
        \includegraphics[width=1.11\textwidth]{Picture/UnifTrans/SampleDist_BB_conc.png}
    \end{minipage}
    \caption{$200$ samples drawn from the mixture distribution plotted on the underlying 
    distribution of the model $Y$ for the case studies (1) to (6). To approximate the mixture weights and
    optimal mixture components $M = 20,000$ pivot samples were drawn. For the 
    estimation of the quantile and DRM $N = 100,000$ samples are drawn 
    from the mixture distribution. For further details, see Section \ref{sec:SampleDistApp}.}
    \label{fig:SampleDistApp}
\end{figure}

\begin{figure}[H]
    \centering
    \begin{minipage}[c]{0.32\textwidth}
        \centering
        \includegraphics[width=1.11\textwidth]{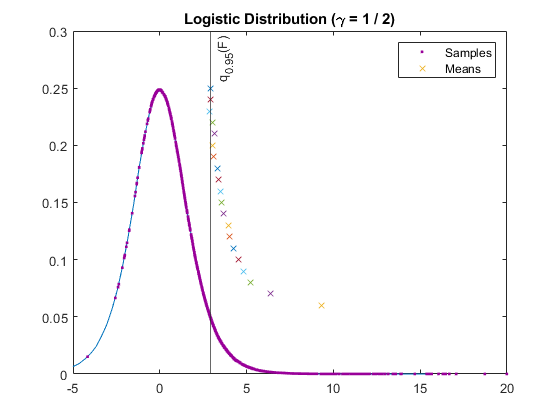}
    \end{minipage}
    \begin{minipage}[c]{0.32\textwidth}
        \centering
        \includegraphics[width=1.11\textwidth]{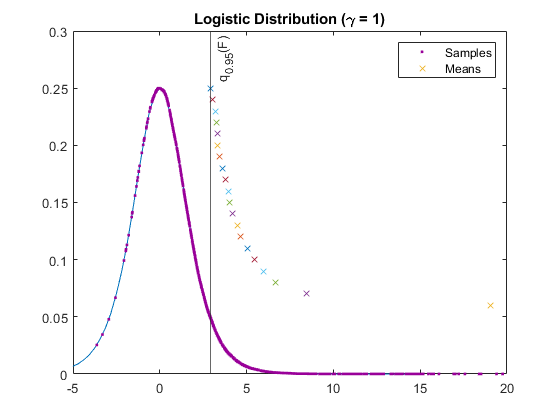}
    \end{minipage}
    \begin{minipage}[c]{0.32\textwidth}
        \centering
        \includegraphics[width=1.11\textwidth]{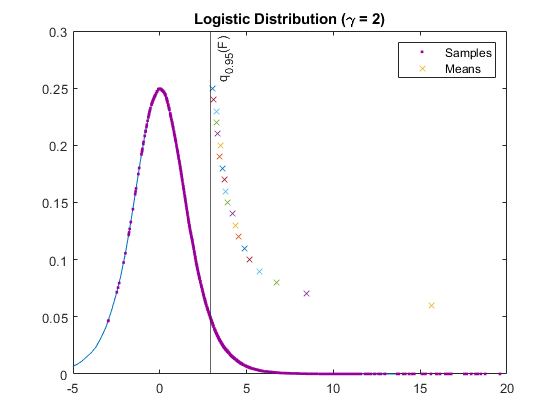}
    \end{minipage}
    \caption{Continuation of Figure \ref{fig:SampleDistApp}.}
      \label{fig:SampleDistApp3}
\end{figure}

\begin{figure}[H]
    \centering
    \begin{minipage}[c]{0.32\textwidth}
        \centering
        \includegraphics[width=1.11\textwidth]{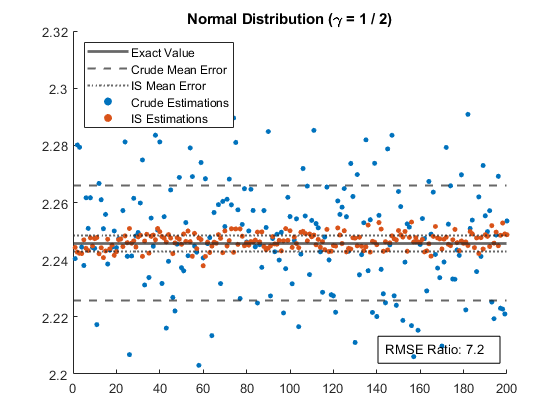}
    \end{minipage}
    \begin{minipage}[c]{0.32\textwidth}
        \centering
        \includegraphics[width=1.11\textwidth]{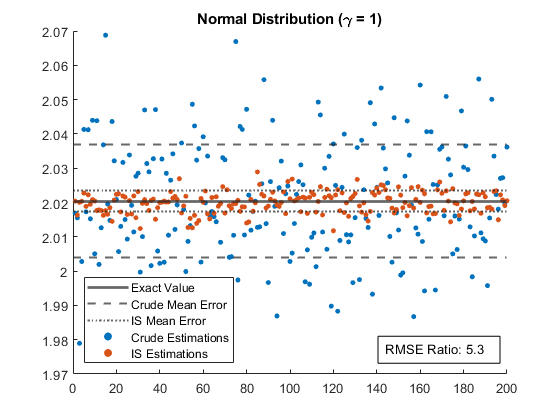}
    \end{minipage}
        \begin{minipage}[c]{0.32\textwidth}
        \centering
        \includegraphics[width=1.11\textwidth]{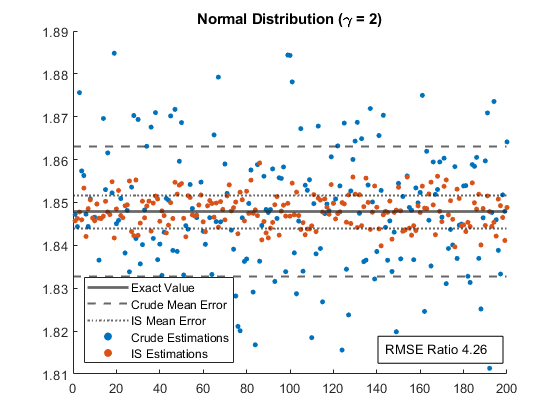}
    \end{minipage}
    \begin{minipage}[c]{0.32\textwidth}
        \centering
        \includegraphics[width=1.11\textwidth]{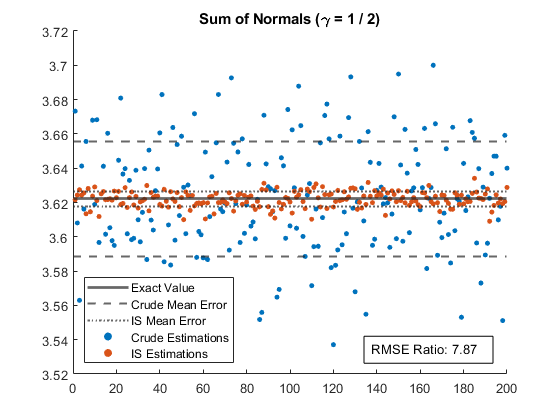}
    \end{minipage}
    \begin{minipage}[c]{0.32\textwidth}
        \centering
        \includegraphics[width=1.11\textwidth]{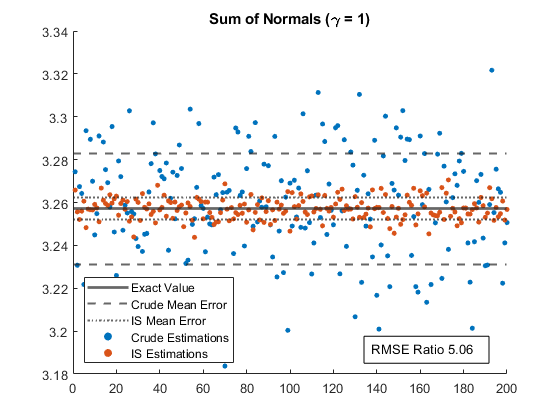}
    \end{minipage}
        \begin{minipage}[c]{0.32\textwidth}
        \centering
        \includegraphics[width=1.11\textwidth]{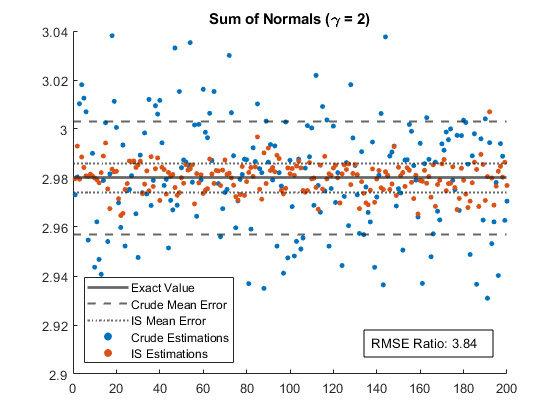}
    \end{minipage}
        \begin{minipage}[c]{0.32\textwidth}
        \centering
        \includegraphics[width=1.11\textwidth]{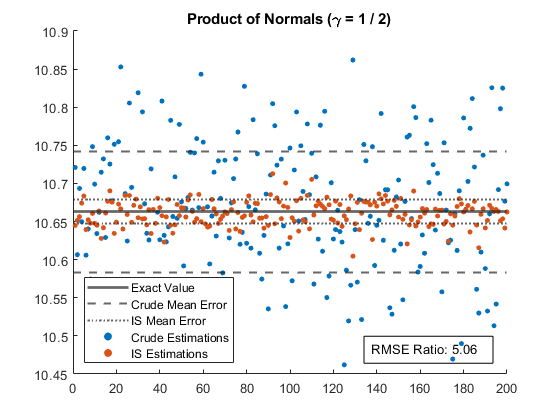}
    \end{minipage}
    \begin{minipage}[c]{0.32\textwidth}
        \centering
        \includegraphics[width=1.11\textwidth]{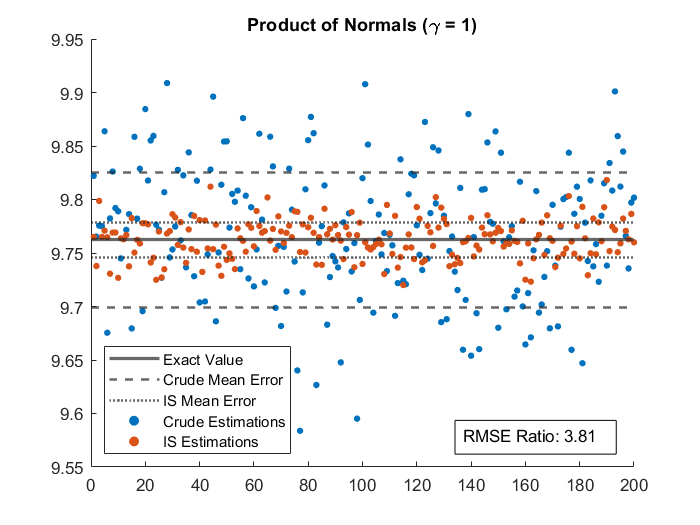}
    \end{minipage}
    \begin{minipage}[c]{0.32\textwidth}
        \centering
        \includegraphics[width=1.11\textwidth]{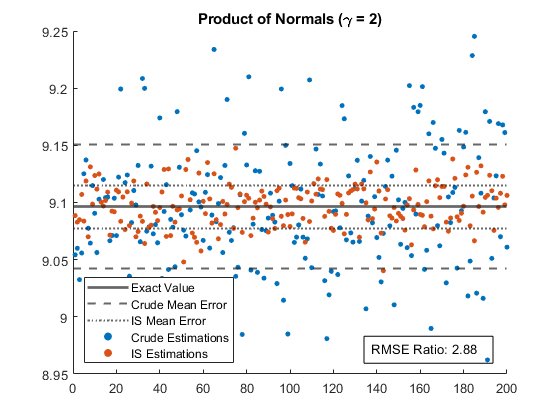}
    \end{minipage}
    \caption{200 estimations with a crude Monte Carlo estimation and the proposed importance sampling 
    method for the models (1) to (6) and the considered DRMs $\rho_{g_{\gamma, \alpha}}$, 
    $\gamma \in \{1/2, 1, 2\}, \alpha = 0.05$.
    Also shown is the ``exact value", which is calculated with
    a crude Monte Carlo estimation over $10,000,000$ samples, the estimated 
    root mean square error of the estimation around the exact value and the ratio of the 
    root mean square error of the crude method and importance sampling method. }
    \label{fig:ErrorDist}
\end{figure}

\begin{figure}[H]
    \centering
          \begin{minipage}[c]{0.32\textwidth}
        \centering
        \includegraphics[width=1.11\textwidth]{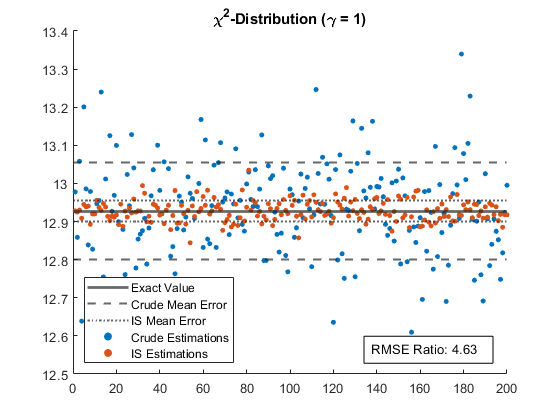}
    \end{minipage}
    \begin{minipage}[c]{0.32\textwidth}
        \centering
        \includegraphics[width=1.11\textwidth]{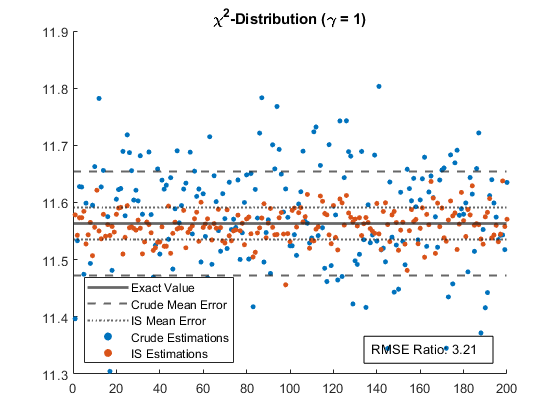}
    \end{minipage}
        \begin{minipage}[c]{0.32\textwidth}
        \centering
        \includegraphics[width=1.11\textwidth]{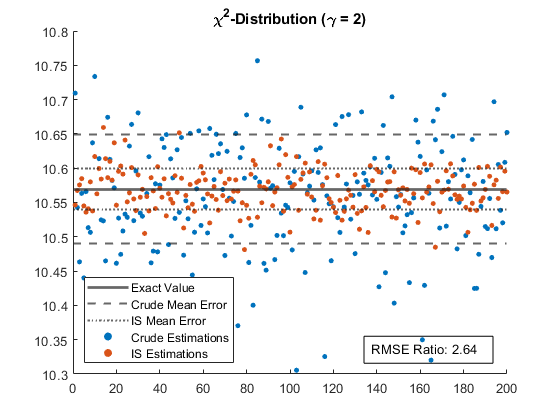}
    \end{minipage}
    \begin{minipage}[c]{0.32\textwidth}
        \centering
        \includegraphics[width=1.11\textwidth]{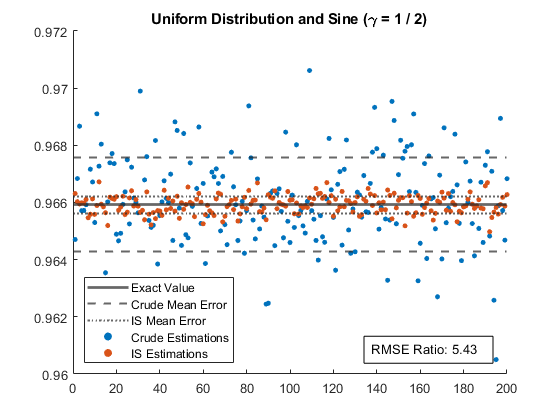}
    \end{minipage}
    \begin{minipage}[c]{0.32\textwidth}
        \centering
        \includegraphics[width=1.11\textwidth]{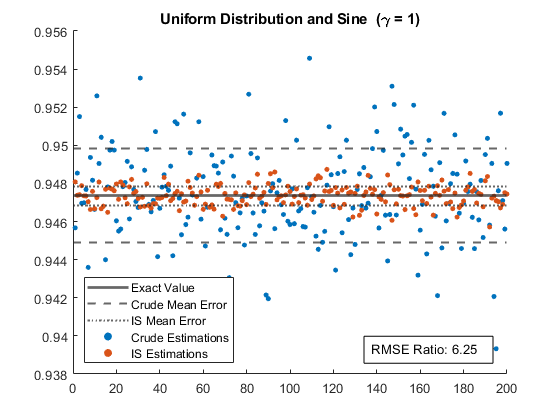}
    \end{minipage}
    \begin{minipage}[c]{0.32\textwidth}
        \centering
        \includegraphics[width=1.11\textwidth]{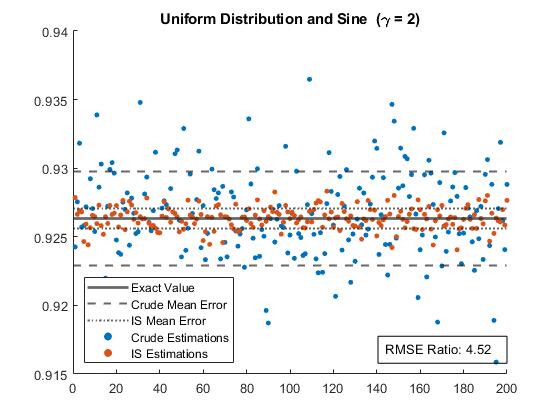}
    \end{minipage}
    \begin{minipage}[c]{0.32\textwidth}
        \centering
        \includegraphics[width=1.11\textwidth]{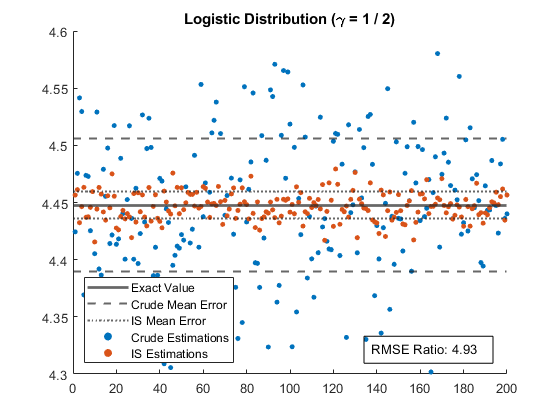}
    \end{minipage}
    \begin{minipage}[c]{0.32\textwidth}
        \centering
        \includegraphics[width=1.11\textwidth]{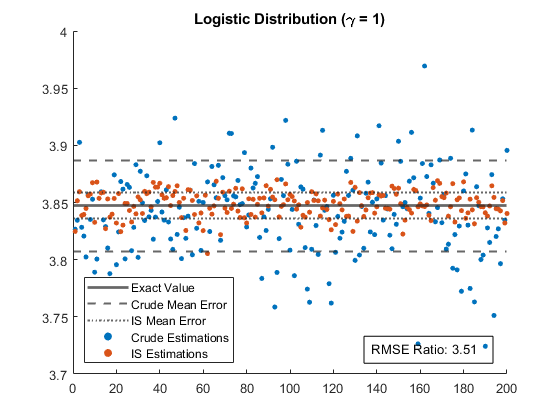}
    \end{minipage}
    \begin{minipage}[c]{0.32\textwidth}
        \centering
        \includegraphics[width=1.11\textwidth]{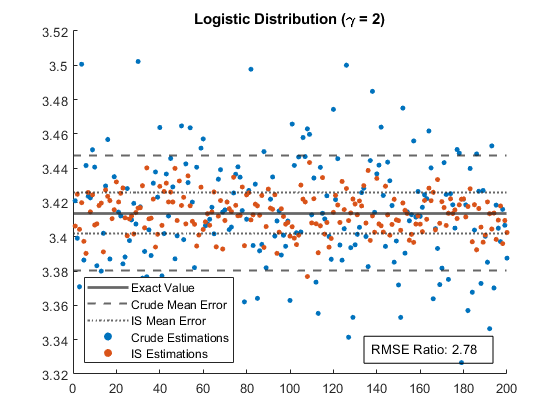}
    \end{minipage}
    \caption{Continuation of Figure \ref{fig:ErrorDist}.}
\end{figure}

\begin{figure}[H]
    \centering
    \begin{minipage}[c]{0.32\textwidth}
        \centering
        \includegraphics[width=1.11\textwidth]{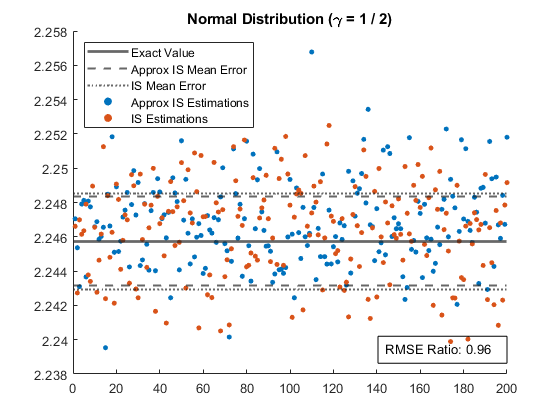}
    \end{minipage}
    \begin{minipage}[c]{0.32\textwidth}
        \centering
        \includegraphics[width=1.11\textwidth]{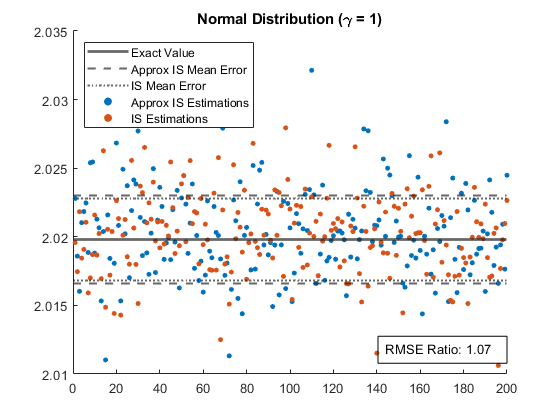}
    \end{minipage}
        \begin{minipage}[c]{0.32\textwidth}
        \centering
        \includegraphics[width=1.11\textwidth]{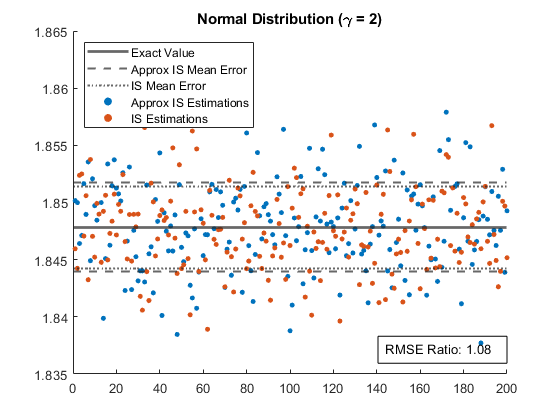}
    \end{minipage}
    \caption{200 estimations with the importance sampling method with exact knowledge of the
    model and the importance sampling method with an approximation of the model chosen through $k$-fold validation for the case studies (1) to (6) and the considered DRMs $\rho_{g_{\gamma, \alpha}}$, 
    $\gamma \in \{1/2, 1, 2\}, \alpha = 0.05$.
    Also shown is the ``exact value", which is calculated with
    a crude Monte Carlo estimation over $10,000,000$ samples, the estimated 
    root mean square error of the estimation around the exact value and the RMSE ratio between the two importance sampling methods. }
    \label{fig:ErrorComp}
\end{figure}

\hfill

\begin{figure}[H]
    \centering
      \begin{minipage}[c]{0.32\textwidth}
        \centering
        \includegraphics[width=1.11\textwidth]{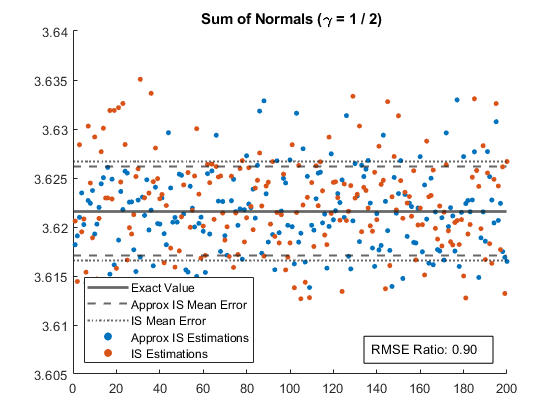}
    \end{minipage}
    \begin{minipage}[c]{0.32\textwidth}
        \centering
        \includegraphics[width=1.11\textwidth]{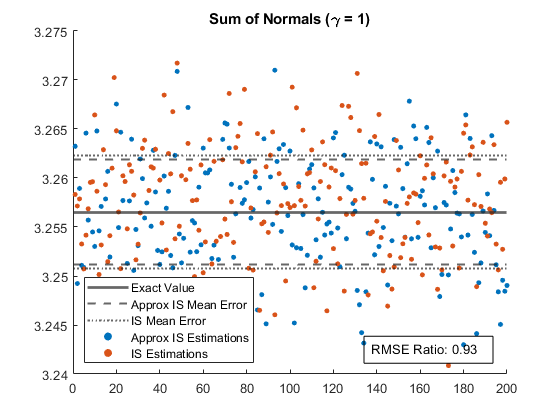}
    \end{minipage}
        \begin{minipage}[c]{0.32\textwidth}
        \centering
        \includegraphics[width=1.11\textwidth]{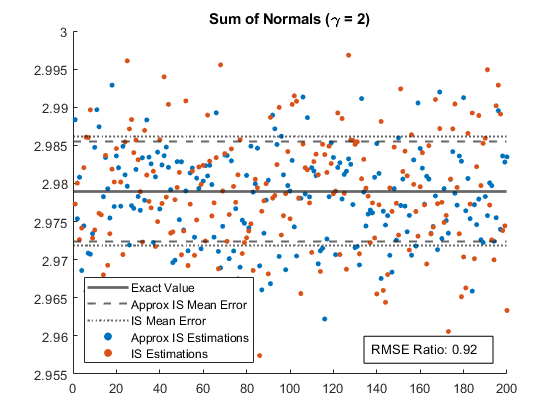}
    \end{minipage}
    \begin{minipage}[c]{0.32\textwidth}
        \centering
        \includegraphics[width=1.11\textwidth]{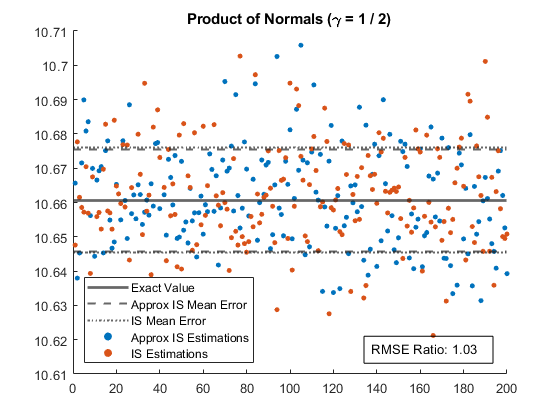}
    \end{minipage}
    \begin{minipage}[c]{0.32\textwidth}
        \centering
        \includegraphics[width=1.11\textwidth]{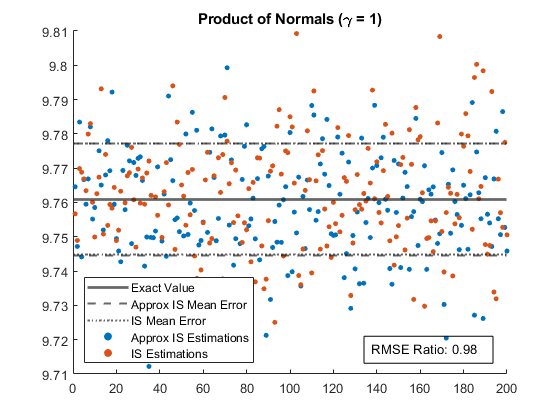}
    \end{minipage}
        \begin{minipage}[c]{0.32\textwidth}
        \centering
        \includegraphics[width=1.11\textwidth]{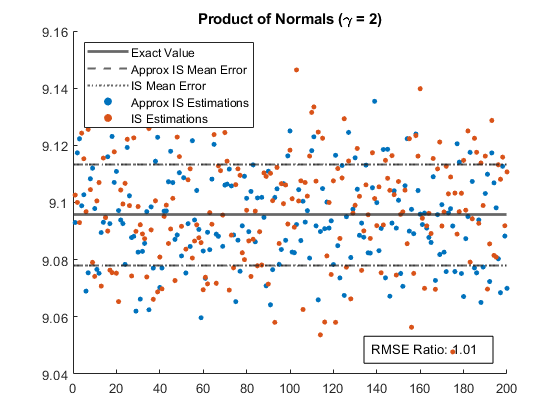}
    \end{minipage}
        \begin{minipage}[c]{0.32\textwidth}
        \centering
        \includegraphics[width=1.11\textwidth]{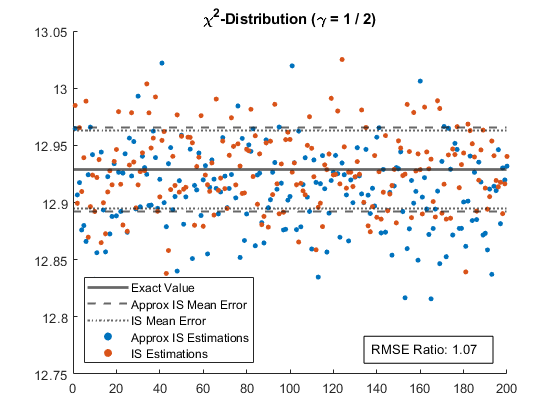}
    \end{minipage}
    \begin{minipage}[c]{0.32\textwidth}
        \centering
        \includegraphics[width=1.11\textwidth]{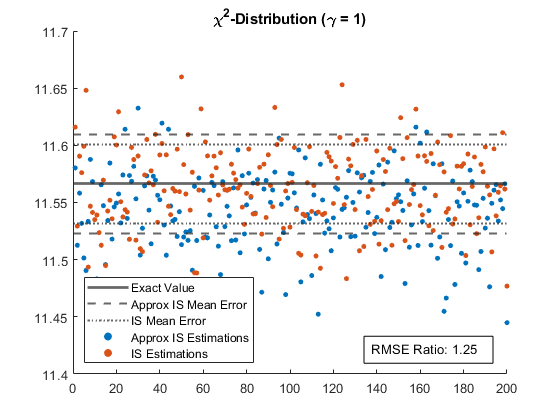}
    \end{minipage}
        \begin{minipage}[c]{0.32\textwidth}
        \centering
        \includegraphics[width=1.115\textwidth]{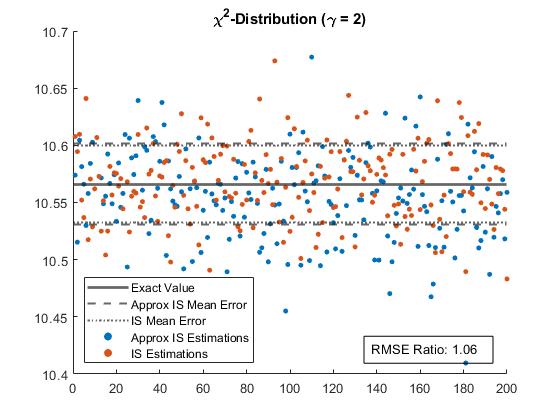}
    \end{minipage}
       \begin{minipage}[c]{0.32\textwidth}
        \centering
        \includegraphics[width=1.11\textwidth]{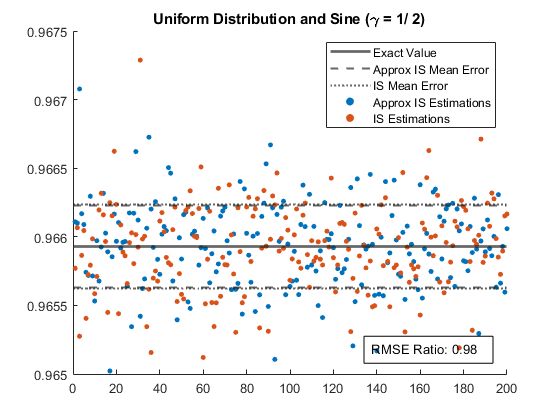}
    \end{minipage}
    \begin{minipage}[c]{0.32\textwidth}
        \centering
        \includegraphics[width=1.11\textwidth]{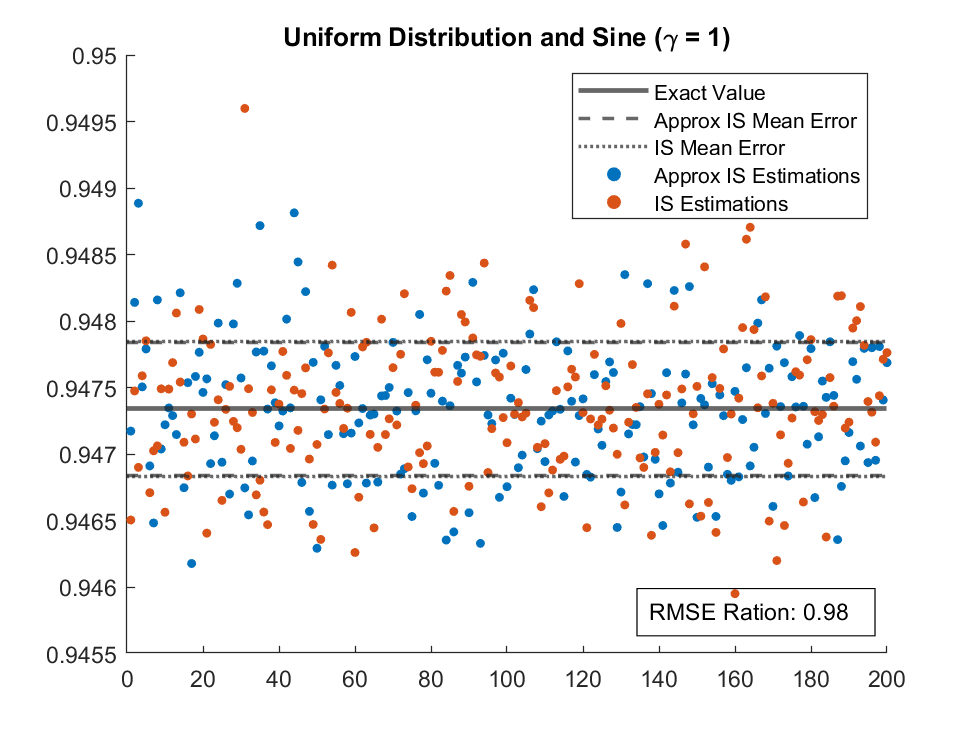}
    \end{minipage}
    \begin{minipage}[c]{0.32\textwidth}
        \centering
        \includegraphics[width=1.11\textwidth]{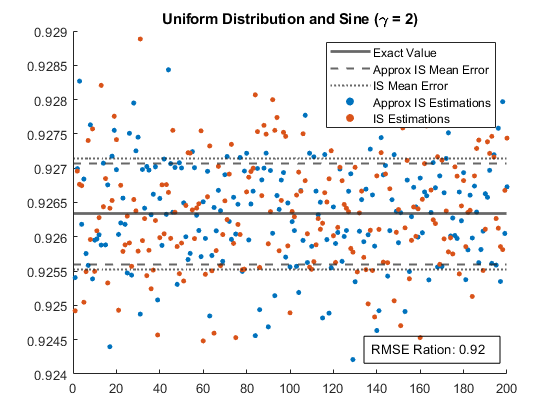}
    \end{minipage}
    \begin{minipage}[c]{0.32\textwidth}
        \centering
        \includegraphics[width=1.11\textwidth]{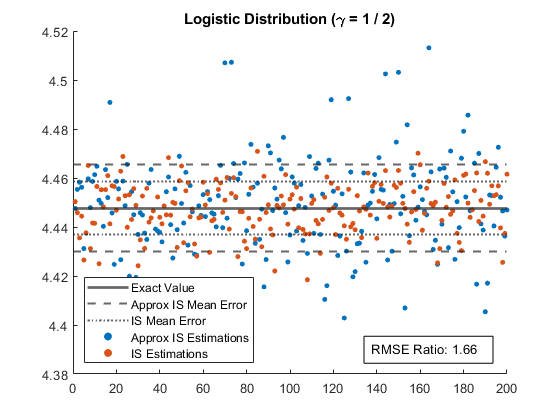}
    \end{minipage}
    \begin{minipage}[c]{0.32\textwidth}
        \centering
        \includegraphics[width=1.11\textwidth]{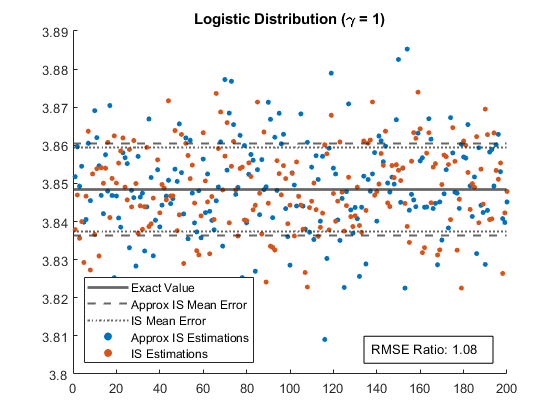}
    \end{minipage}
    \begin{minipage}[c]{0.32\textwidth}
        \centering
        \includegraphics[width=1.11\textwidth]{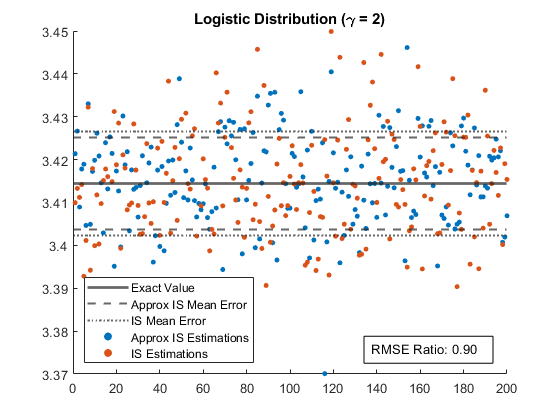}
    \end{minipage}
    \caption{Continuation of Figure \ref{fig:ErrorComp}. }
\end{figure}

\begin{figure}[H]
    \centering
    \begin{minipage}[c]{0.32\textwidth}
        \centering
        \includegraphics[width=1.11\textwidth]{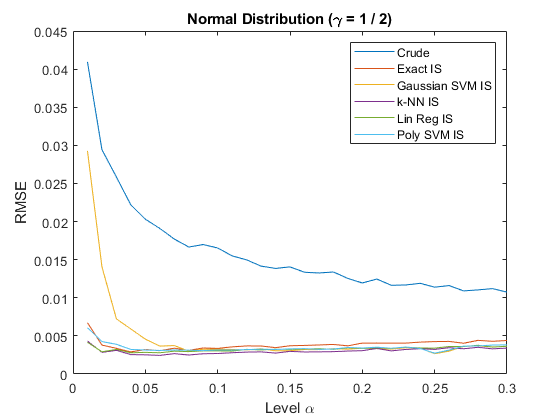}
    \end{minipage}
    \begin{minipage}[c]{0.32\textwidth}
        \centering
        \includegraphics[width=1.11\textwidth]{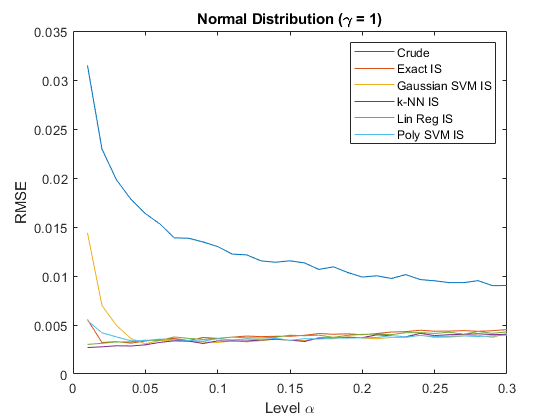}
    \end{minipage}
    \begin{minipage}[c]{0.32\textwidth}
        \centering
        \includegraphics[width=1.11\textwidth]{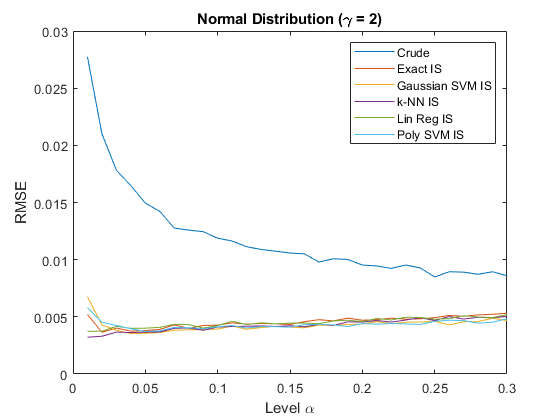}
    \end{minipage}
        \begin{minipage}[c]{0.32\textwidth}
        \centering
        \includegraphics[width=1.11\textwidth]{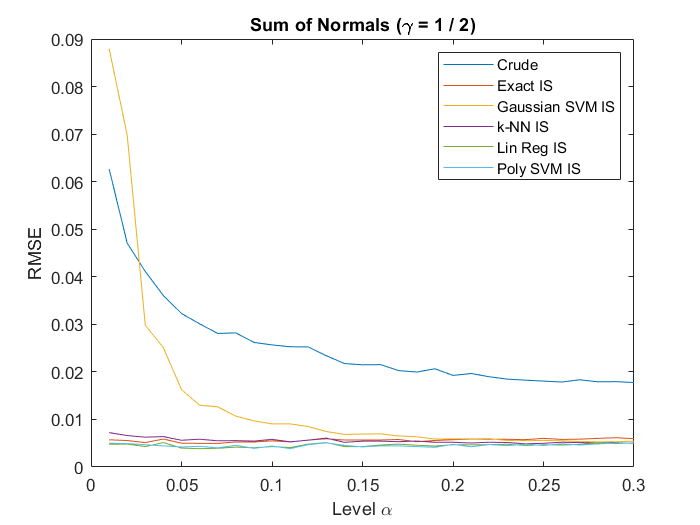}
    \end{minipage}
    \begin{minipage}[c]{0.32\textwidth}
        \centering
        \includegraphics[width=1.11\textwidth]{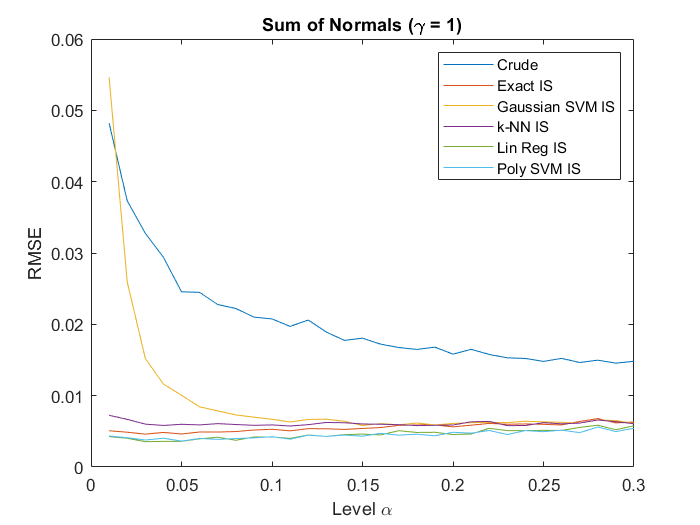}
    \end{minipage}
    \begin{minipage}[c]{0.32\textwidth}
        \centering
        \includegraphics[width=1.11\textwidth]{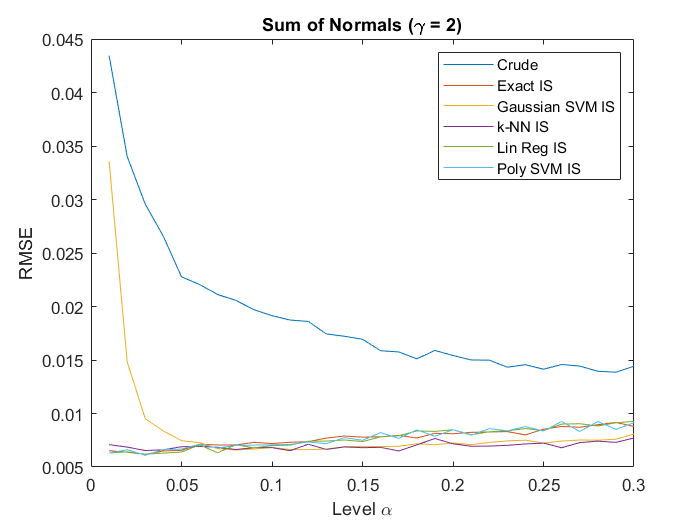}
    \end{minipage}
        \begin{minipage}[c]{0.32\textwidth}
        \centering
        \includegraphics[width=1.11\textwidth]{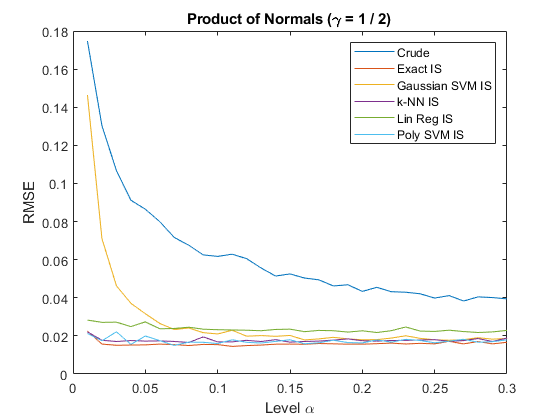}
    \end{minipage}
    \begin{minipage}[c]{0.32\textwidth}
        \centering
        \includegraphics[width=1.11\textwidth]{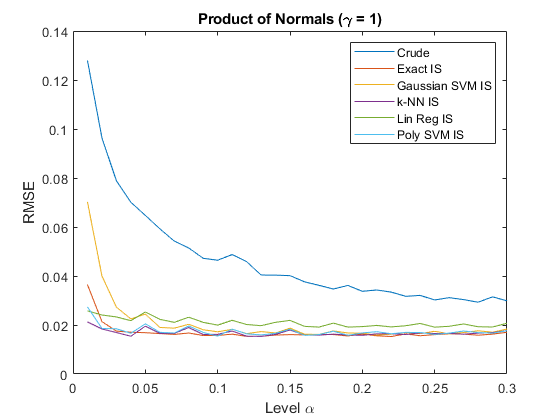}
    \end{minipage}
    \begin{minipage}[c]{0.32\textwidth}
        \centering
        \includegraphics[width=1.11\textwidth]{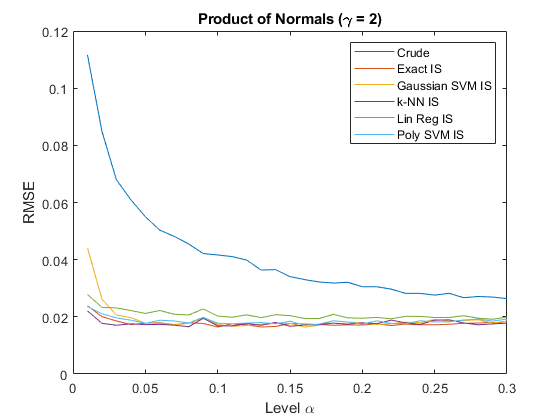}
    \end{minipage}
        \begin{minipage}[c]{0.32\textwidth}
        \centering
        \includegraphics[width=1.11\textwidth]{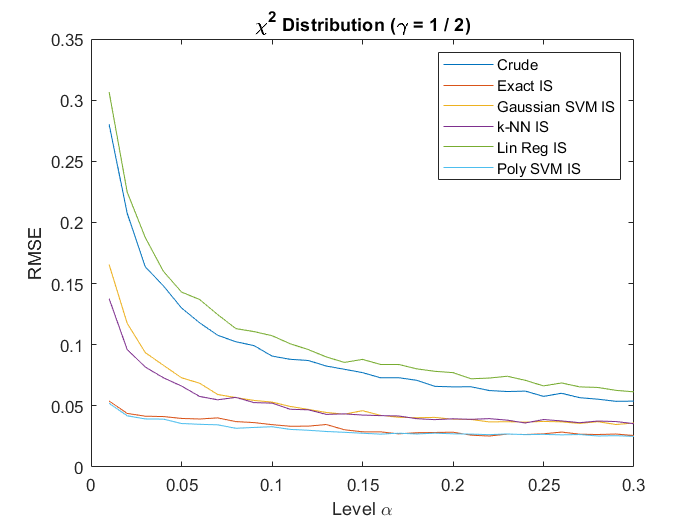}
    \end{minipage}
    \begin{minipage}[c]{0.32\textwidth}
        \centering
        \includegraphics[width=1.11\textwidth]{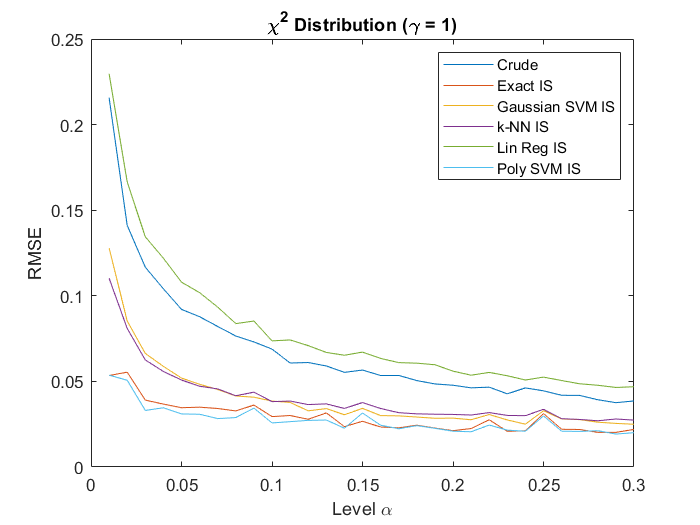}
    \end{minipage}
    \begin{minipage}[c]{0.32\textwidth}
        \centering
        \includegraphics[width=1.11\textwidth]{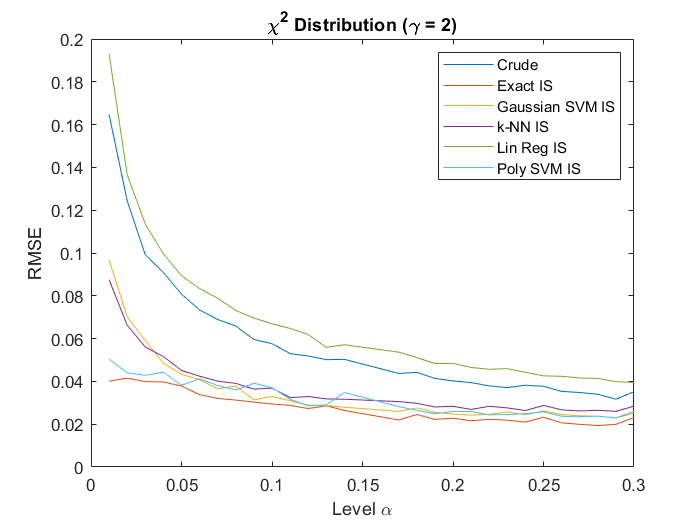}
    \end{minipage}
    \caption{Root Mean Square Error (RMSE) for estimating the DRMs $\rho_{g_{\gamma, \alpha}}$, with $\gamma \in \{1/2, 1, 2\}$, $\alpha \in [0.01, 0.3]$, for the models (1) to (6).
    The DRMs are estimated with a crude Monte Carlo method and the proposed importance sampling method using different approximations of the black box models used 
    in the paper. }
    \label{fig:LevelError}
\end{figure}

\begin{figure}[H]
    \centering
    \begin{minipage}[c]{0.32\textwidth}
        \centering
        \includegraphics[width=1.11\textwidth]{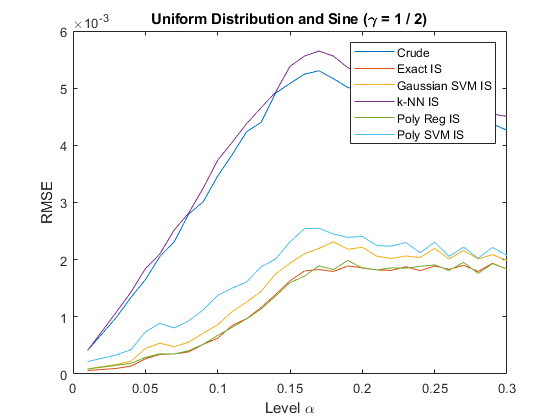}
    \end{minipage}
    \begin{minipage}[c]{0.32\textwidth}
        \centering
        \includegraphics[width=1.11\textwidth]{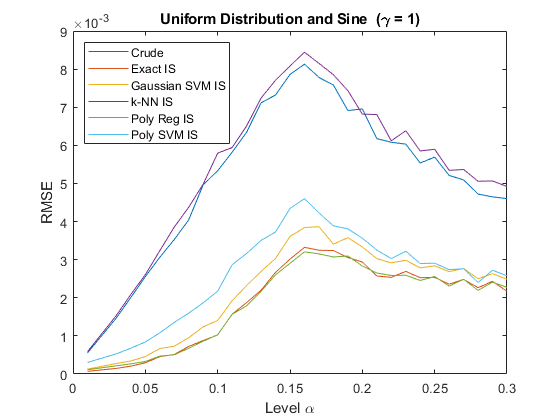}
    \end{minipage}
    \begin{minipage}[c]{0.32\textwidth}
        \centering
        \includegraphics[width=1.11\textwidth]{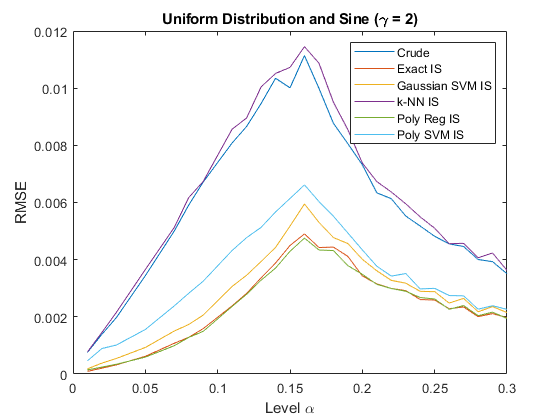}
    \end{minipage}
    \begin{minipage}[c]{0.32\textwidth}
        \centering
        \includegraphics[width=1.11\textwidth]{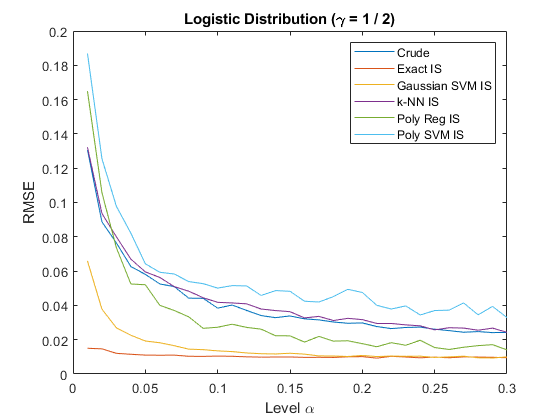}
    \end{minipage}
    \begin{minipage}[c]{0.32\textwidth}
        \centering
        \includegraphics[width=1.11\textwidth]{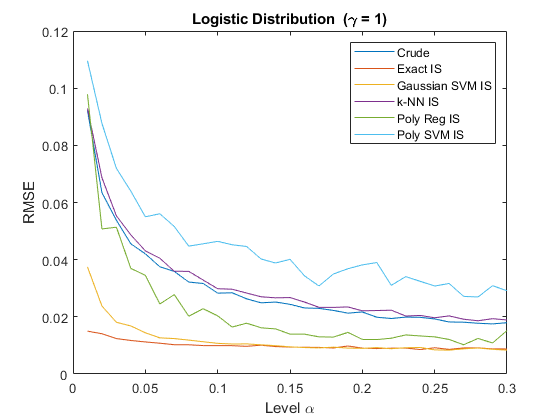}
    \end{minipage}
    \begin{minipage}[c]{0.32\textwidth}
        \centering
        \includegraphics[width=1.11\textwidth]{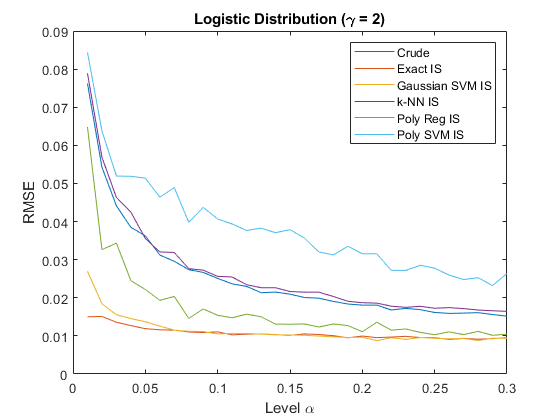}
    \end{minipage}
    \caption{Continuation of Figure \ref{fig:LevelError}. }    \label{fig:LevelError3}
\end{figure}

\begin{figure}[H]
    \centering
    \begin{minipage}[c]{0.32\textwidth}
        \centering
        \includegraphics[width=1.11\textwidth]{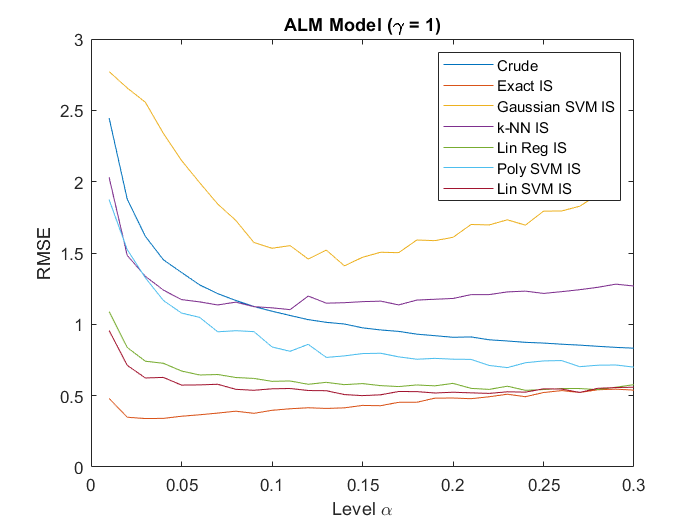}
    \end{minipage}
    \begin{minipage}[c]{0.32\textwidth}
        \centering
        \includegraphics[width=1.11\textwidth]{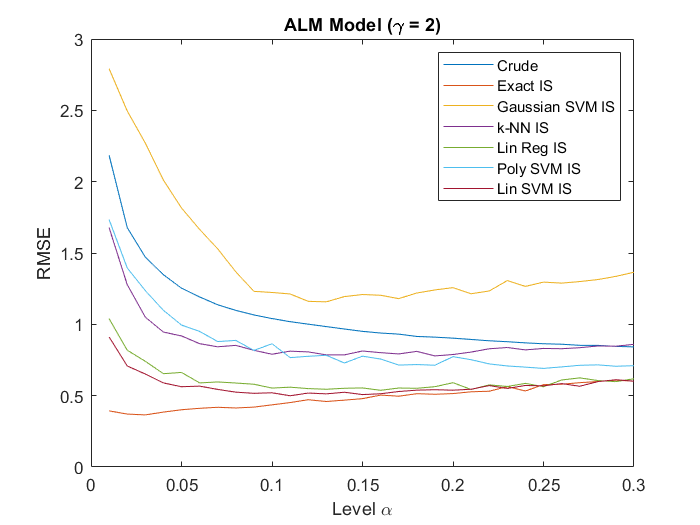}
    \end{minipage}
    \begin{minipage}[c]{0.32\textwidth}
        \centering
        \includegraphics[width=1.11\textwidth]{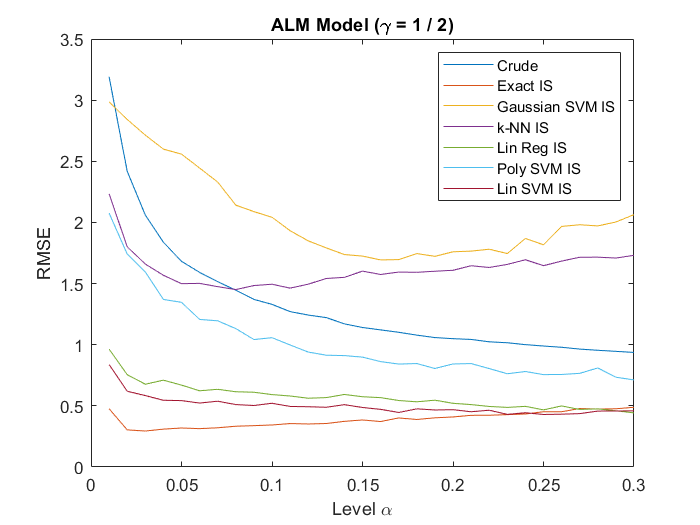}
    \end{minipage}
    \caption{RMSE of the crude method and various importance sampling methods 
    of the considered DRMs for the evolution of the net 
    asset value in the ALM model. The importance sampling methods are implemented with 
    the different approximation techniques considered in the paper. }
    \label{fig:ALM_LevelError}
\end{figure}

\end{document}